\documentclass[11pt, letterpaper]{article}
\usepackage[utf8]{inputenc}
\usepackage{fullpage,times}

\usepackage{amsmath,amsfonts,amsthm,amssymb,multirow}
\usepackage{thm-restate}
\usepackage{times}
\usepackage{graphicx}
\usepackage{floatpag}
\usepackage{algorithm}
\usepackage[noend]{algpseudocode}
\usepackage{bbold}
\usepackage{enumitem}
\usepackage{subcaption} 
\usepackage{calc}
\usepackage{tikz}
\usetikzlibrary{decorations.markings,decorations.pathreplacing,arrows, automata, positioning, quotes, shapes,backgrounds, arrows.meta}
\tikzstyle{vertex}=[circle, draw, inner sep=0pt, minimum size=4pt, fill = black]

\usepackage{graphicx}
\usepackage{tabularx}
\usepackage[hidelinks]{hyperref}
\usepackage{comment}
\makeatletter
\newcommand{\multiline}[1]{%
  \begin{tabularx}{\dimexpr\linewidth-\ALG@thistlm}[t]{@{}X@{}}
    #1
  \end{tabularx}
}
\makeatother
\usepackage{mathtools}

\makeatletter
\def\BState{\State\hskip-\ALG@thistlm}
\makeatother

\usepackage[compact]{titlesec}
\titlespacing{\section}{0pt}{3ex}{2ex}
\titlespacing{\subsection}{0pt}{2ex}{1ex}
\titlespacing{\subsubsection}{0pt}{0.5ex}{0ex}

\newtheorem{theorem}{Theorem}[section]

\newtheorem{corollary}[theorem]{Corollary}

\newtheorem{lemma}[theorem]{Lemma}
\newtheorem{claim}[theorem]{Claim}
\newtheorem{hypothesis}{Hypothesis}

\newtheorem{problem}[theorem]{Problem}

\makeatletter
\let\c@fconjecture\c@conjecture
\makeatother

\makeatletter
\let\c@fconj\c@conj
\makeatother

\def \eps {\varepsilon}

\def \poly {\mathop{\rm poly}} %

\def\tO{\tilde{O}}

\newcommand{\OO}{\tilde{O}}%

\newcommand{\APSP}{\mbox{\sf Real-APSP}}
\newcommand{\IntAPSP}{\mbox{\sf Int-APSP}}
\newcommand{\ExactTri}{\mbox{\sf Real-Exact-Tri}}
\newcommand{\AEExactTri}{\mbox{\sf Real-AE-Exact-Tri}}
\newcommand{\AEExactTriCount}{\mbox{\sf Real-\#AE-Exact-Tri}}
\newcommand{\AENegTriCount}{\mbox{\sf Real-\#AE-Neg-Tri}}
\newcommand{\IntExactTri}{\mbox{\sf Int-Exact-Tri}}
\newcommand{\IntAEExactTri}{\mbox{\sf Int-AE-Exact-Tri}}
\newcommand{\IntAEExactTriCount}{\mbox{\sf Int-\#AE-Exact-Tri}}
\newcommand{\IntAENegTriCount}{\mbox{\sf Int-\#AE-Neg-Tri}}
\newcommand{\ThreeSUM}{\mbox{\sf Real-3SUM}}
\newcommand{\IntThreeSUM}{\mbox{\sf Int-3SUM}}
\newcommand{\AllThreeSUM}{\mbox{\sf Real-All-Nums-3SUM}}
\newcommand{\IntAllThreeSUM}{\mbox{\sf Int-All-Nums-3SUM}}
\newcommand{\OddSUM}{\mbox{\sf Real-$(2k+1)$SUM}}
\newcommand{\OV}{\mbox{\sf OV}}
\newcommand{\kOV}[1]{\mbox{\sf #1-OV}}

\newcommand{\AESparseTri}{\mbox{\sf AE-Sparse-Tri}}
\newcommand{\AESparseTriCount}{\mbox{\sf \#AE-Sparse-Tri}}
\newcommand{\AEMonoTri}{\mbox{\sf AE-Mono-Tri}}
\newcommand{\AEMonoTriCount}{\mbox{\sf \#AE-Mono-Tri}}

\newcommand{\MinPlus}{\mbox{\sf Real-$(\min,+)$-Product}}
\newcommand{\IntMinPlus}{\mbox{\sf Int-$(\min,+)$-Product}}

\newcommand{\DistinctEq}{\mbox{\sf $(\mbox{distinct},=)$-Product}}
\newcommand{\ModeEq}{\mbox{\sf $(\mbox{mode},=)$-Product}}
\newcommand{\DistinctPlus}{\mbox{\sf $(\mbox{distinct},+)$-Product}}

\newcommand{\ColorBMM}{\mbox{\sf Colorful-BMM}}
\newcommand{\AEColorSparseTri}{\mbox{\sf AE-Colorful-Sparse-Tri}}

\newcommand{\ModeBMM}{\mbox{\sf Mode-BMM}}

\newcommand{\CNFSAT}{\mbox{\sf CNF-SAT}}

\newcommand{\TClong}{\mbox{\sf Triangle-Collection*}}
\newcommand{\TCollong}{\mbox{\sf Triangle-Collection}}
\newcommand{\TCsslong}{\mbox{\sf Triangle-Collection**}}

\newcommand{\ACPTCblong}{\mbox{\sf ACP-Triangle-Collection}}

\newcommand{\TC}{\mbox{\sf Tri-Co*}}
\newcommand{\TCol}{\mbox{\sf Tri-Co}}
\newcommand{\TCss}{\mbox{\sf Tri-Co**}}
\newcommand{\TCsss}{\mbox{\sf Tri-Co$_{\textrm{light}}$}}

\newcommand{\ACPTCb}{\mbox{\sf ACP-Tri-Co}}
\newcommand{\ACPTC}{\mbox{\sf ACP-Tri-Co*}}
\newcommand{\ACPTCss}{\mbox{\sf ACP-Tri-Co**}}
\newcommand{\ACPTCsss}{\mbox{\sf ACP-Tri-Co$_{\textrm{light}}$}}

\newcommand{\SetDisj}{\mbox{\sf Set-Disjointness}}
\newcommand{\SetInter}{\mbox{\sf Set-Intersection}}
\newcommand{\col}{\mathop{\it color}}

\newcommand{\up}[1]{\left\lceil#1\right\rceil}
\newcommand{\nhat}{\hat{n}}

\DeclareMathOperator*{\argmin}{arg\,min}

\title{Hardness for Triangle Problems under\\ Even More Believable Hypotheses:\\ Reductions from Real APSP, Real 3SUM, and OV}

\author{Timothy M. Chan\thanks{Supported by NSF Grant CCF-1814026.}\\UIUC\\tmc@illinois.edu \and Virginia {Vassilevska Williams}\thanks{Supported by an NSF CAREER Award, NSF Grants CCF-1528078, CCF-1514339 and CCF-1909429, a BSF Grant BSF:2012338, a Google Research Fellowship and a Sloan Research Fellowship.}\\MIT\\virgi@mit.edu \and Yinzhan Xu\thanks{Partially supported by NSF Grant CCF-1528078.}\\MIT\\xyzhan@mit.edu}
\date{}

\begin{document}
\pagenumbering{gobble} 

\maketitle
\begin{abstract}
The $3$SUM hypothesis, the All-Pairs Shortest Paths (APSP) hypothesis and the Strong Exponential Time Hypothesis are the three main hypotheses in the area of fine-grained complexity. 
So far, within the area, the first two hypotheses have mainly been about integer inputs in the Word RAM model of computation. 
The ``Real APSP" and ``Real $3$SUM" hypotheses, which assert that the APSP and $3$SUM hypotheses hold for real-valued inputs in a reasonable version of the Real RAM model, are even more believable than their integer counterparts.  

Under the very believable hypothesis that at least one of the Integer $3$SUM hypothesis, Integer APSP hypothesis or SETH is true, Abboud, Vassilevska W. and Yu [STOC 2015] showed that a problem called Triangle Collection requires $n^{3-o(1)}$ time on an $n$-node graph. 

The main result of this paper is 
a nontrivial lower bound for a slight generalization of Triangle Collection, called All-Color-Pairs Triangle Collection, under the even more believable hypothesis that at least one of the Real $3$SUM, the Real APSP, and the Orthogonal Vector (OV) hypotheses is true. Combined with slight modifications of prior reductions from Triangle Collection, we obtain polynomial conditional lower bounds for problems such as the (static) ST-Max Flow problem and dynamic versions of Max Flow, Single-Source Reachability Count, and Counting Strongly Connected Components, now under the new weaker hypothesis. 

Our main result is built on the following two lines of reductions. 
\begin{itemize}
    \item Real APSP and Real $3$SUM hardness for the All-Edges Sparse Triangle problem. Prior reductions only worked from the integer variants of these problems.
    \item Real APSP and OV hardness for a variant of the Boolean Matrix Multiplication problem. 
\end{itemize}

Along the way we show that Triangle Collection is equivalent to a simpler restricted version of the problem, simplifying prior work.
Our techniques also have other interesting implications, such as a super-linear lower bound of Integer All-Numbers $3$SUM based on the Real $3$SUM hypothesis, and a tight
lower bound for a string matching problem based on the OV hypothesis. 

\end{abstract}

\newpage
\pagenumbering{arabic}  
\section{Introduction}

Fine-grained complexity is an active area of study that gives a problem-centric approach to complexity. Its major goal is to prove relationships and equivalences between problems whose best known running times have not been improved in decades. Through a variety of techniques and sophisticated reductions many  problems from a huge variety of domains and with potentially vastly different running time complexities are now known to be related via {\em fine-grained reductions} (see e.g. the survey \cite{virgisurvey}). 

As a consequence of the known reductions, the hardness of most of the studied problems in fine-grained complexity can be based on the presumed hardness of three key problems: the {\sf 3SUM} problem, the All-Pairs Shortest Paths ({\sf APSP}) problem and {\sf CNF-SAT}\@. Their associated hardness hypotheses below are all defined for the Word RAM model of computation with $O(\log n)$-bit words:

\begin{itemize}
\item {\bf The (Integer) $3$SUM hypothesis.} There is no algorithm that can check whether a list of $n$ integers from $\pm[n^c]$
for some constant $c$ contains three integers that sum up to zero in $O(n^{2-\eps})$ time for $\eps > 0$.%
\footnote{
Throughout this paper, $[N]$ denotes $\{0,1,\ldots,N-1\}$ and $\pm[N]$ denotes $\{-(N-1),\ldots,N-1\}$.
}

\item {\bf The (Integer) APSP hypothesis.} There is no algorithm that can solve the {\sf APSP} problem in an $n$-node graph whose edge weights are from $\pm[n^c]$
for some constant $c$ in $O(n^{3-\eps})$ time for $\eps > 0$.

\item {\bf The Strong Exponential Time Hypothesis (SETH).} For every $\eps > 0$, there exists an integer $k$ such that $k$-{\sf SAT} on $n$ variables cannot be solved in $O(2^{(1-\eps) n})$ time. An equivalent formulation states that there is no $O(2^{(1-\eps)n})$ time algorithm for {\sf CNF-SAT} with $n$ variables and $O(n)$ clauses for any $\eps>0$.
\end{itemize}

All three hardness hypotheses were studied before Fine-Grained Complexity even got its name.
The complexity of {\sf 3SUM} was first used as a basis for hardness in the computational geometry community by Gajentaan and Overmars~\cite{gajentaan1995class}. The complexity of {\sf APSP} has been used as a basis of hardness in the graph algorithms community at least since the early 2000s (e.g.\ \cite{RodittyZ04}). SETH was first studied in 1999 by Impagliazzo and Paturi \cite{ip01}, though the name ``SETH'' was first given later in \cite{cip13}. 
Together, the three hypotheses have been very influential, giving very strong lower bounds for many problems.

The $3$SUM hypothesis is now known to imply tight hardness results for many geometric problems (e.g.\  \cite{erickson1995lower,geo2, geo3, geo4, geo5, geo6, geo7, geo8, geo9}), and also for some non-algebraic problems (e.g. \cite{patrascu2010towards,AbboudWW14,lincoln2020monochromatic,abboud2014popular,kopelowitz2016higher}); some convolution problems~\cite{lincoln2020monochromatic,patrascu2010towards} are known to be equivalent to {\sf 3SUM}.
The APSP hypothesis is known to imply tight hardness results for many problems (e.g. \cite{patrascu2010towards, abboud2014popular,BringmannGMW20}), and many problems are also known to be fine-grained equivalent to {\sf APSP} (\cite{focsyj,abboud2014centrality}). SETH is now known to imply an enormous number of lower bounds both for problems in exponential time (e.g.\  \cite{exact1, exact2, exact3, exact4}), and in polynomial time (e.g.\  \cite{AlmanW15,Williams05,BackursRSWW18} and many more); the hardness of a small number of problems is also  known to be equivalent to SETH \cite{CyganDLMNOPSW16}. See \cite{virgisurvey} for more known implications.

There are no known direct relationships between the three hypotheses, and there is some evidence (e.g.~\cite{carmosino2016nondeterministic}) that reducing between them might be difficult.  
As we do not really know which, if any, of these hypotheses actually hold, it is important to consider weaker hypotheses that still give meaningful hardness results.

A natural hardness hypothesis considered by Abboud, Vassilevska W. and Yu~\cite{abboud2018matching} is the following:

\begin{hypothesis}
\label{conj:conj1}
At least one of SETH, the (Integer) $3$SUM or the (Integer) APSP hypothesis is true.
\end{hypothesis}

Under Hypothesis~\ref{conj:conj1}, Abboud, Vassilevska W. and Yu proved polynomial lower bounds for a variant of maximum flow and several problems in dynamic graph algorithms: dynamically maintaining the maximum flow in a graph, the number of nodes reachable from a fixed source in a graph ({\sf \#SSR}), the number of strongly connected components in a directed graph ({\sf \#SCC}),  and more.

Dahlgaard \cite{dahlgaard2016hardness} showed that computing the diameter of an unweighted graph with $n$ nodes and $m$ edges requires $n^{1-o(1)}\sqrt{m}$ time under Hypothesis~\ref{conj:conj1}. The reductions of both \cite{abboud2018matching} and \cite{dahlgaard2016hardness} utilized a problem called \TCollong,\ which we will abbreviate as \TCol. 

In the \TCol\ problem, given a node-colored graph with $n$ nodes, one is asked whether it is true that for all triples of distinct colors $(a, b, c)$, there exists a triangle whose nodes have these colors. A key step in the above reductions proved in \cite{abboud2018matching} is the following tight hardness result for \TCol.

\begin{theorem}[\cite{abboud2018matching}]
\label{thm:intro:abboud}
Assuming Hypothesis~\ref{conj:conj1}, \TCol\ requires $n^{3-o(1)}$ time.  
\end{theorem}

The main question that inspires this work is the following:
\begin{center}
\emph{Is there a natural hypothesis that is weaker than Hypothesis~\ref{conj:conj1} and implies similar hardness results? }
\end{center}

As mentioned, recent conditional lower bound results based on {\sf APSP} and {\sf 3SUM}
(especially since P{\u{a}}tra{\c{s}}cu's seminal paper~\cite{patrascu2010towards}) 
typically assumed the \emph{integer} variants of these hypotheses.
However, algorithms for  {\sf APSP} and {\sf 3SUM} are more often designed for
the \emph{real}-valued versions of the problems (this includes not only the traditional cubic or quadratic time algorithms, but the celebrated,  slightly subcubic or subquadratic algorithms of
Williams~\cite{Williams14a} or Gr{\o}nlund and Pettie~\cite{gronlund2014}, as we will review shortly).
This discrepancy between the literature on lower and upper bounds raises another intriguing question:

\begin{quote}
\centering
\em Do known conditional lower bounds derived from the integer versions of the
APSP and 3SUM hypotheses hold for the real versions of these hypotheses?
\end{quote}

Our work will give a positive answer to this second question for a plethora of known conditional lower bound results, and will hence provide an answer to the first question as well, since the real versions of the hypotheses are weaker/more believable.

\paragraph{Remarks on models of computation.}
Within fine-grained complexity it is standard to work in the Word RAM model of computation with $O(\log n)$-bit words. 
As we will consider variants of {\sf APSP} and {\sf 3SUM} with real-valued inputs, we will need to work in the Real RAM, a standard model in computational geometry.

The Real RAM (see e.g. Section 6 in the full version of ~\cite{erickson-real})
supports unit cost comparisons and arithmetic operations (addition, subtraction, multiplication, division) on real numbers, unit cost casting integers into reals, in addition to the standard unit cost operations supported by an $O(\log n)$-bit Word RAM\@.
No conversions from real numbers to integers are allowed, and randomization only happens by taking random $O(\log n)$-bit integers, not random reals.

Without further restrictions, the Real RAM can be unrealistically powerful.
However, it is not difficult to define a ``reasonable'' restricted Real RAM model 
for which our reductions still work, such that any algorithm for real-valued inputs in such a model can be converted into an algorithm in the word RAM for integer-valued inputs, running in roughly the same  time. See Appendix~\ref{sec:model} for
a detailed discussion, and several natural ways to define such a reasonable Real RAM model.  Thus, the real versions of the hardness hypotheses under such a model are indeed even more believable than the integer versions.

\paragraph{Real $3$SUM hypothesis.} 
Historically, the early papers on the 3SUM hypothesis from 
computational geometry were concerned with the real instead of the integer case.%
\footnote{
Technically, the original paper by Gajentaan and Overmars~\cite{gajentaan1995class} stated the 3SUM hypothesis for integer inputs, but they %
assumed a Real RAM model of computation, as in most work in computational geometry. %
} 

Let \ThreeSUM\ refer to the version of {\sf 3SUM} for real numbers (in contrast to \IntThreeSUM, which refers to the integer version).
In its original form, the Real $3$SUM hypothesis stated that there is no $o(n^2)$ time algorithm that solves \ThreeSUM\ in the Real RAM model, and some evidence was provided by Erickson~\cite{erickson1995lower},
who proved quadratic lower bounds for algorithms that are allowed to use only restricted forms of real comparisons  (testing the signs of linear functions involving just 3 input reals).
Gr{\o}nlund and Pettie refuted this hypothesis by giving an $O(n^2(\log\log n)^{2/3}/(\log n)^{2/3})$ time deterministic algorithm and an $O(n^2 (\log \log n)^2/ \log n)$ time randomized algorithm for \ThreeSUM~\cite{gronlund2014}. This was subsequently  improved by Freund \cite{freund2017improved}, Gold and Sharir \cite{gold2017improved}, and Chan \cite{chan3sum}, reaching an $n^2 (\log \log n)^{O(1)}/ \log^2 n$ deterministic running time.  Gr\o nlund and Pettie obtained their initial breakthrough by first proving a truly subquadratic upper bound, near $n^{3/2}$, on the linear-decision-tree complexity of \ThreeSUM\ (using comparisons of two sums of pairs of input reals).
More recently, Kane, Lovett, and Moran~\cite{KaneLM19} in another breakthrough improved the linear-decision-tree complexity upper bound to $\OO(n)$,\footnote{In this work, we use $\OO$ to suppress poly-logarithmic factors.} although their approach has not yet led to improved algorithms.
A truly subquadratic time algorithm remains open.

The modern Real $3$SUM hypothesis states that no $O(n^{2-\varepsilon})$  time  algorithm exists for \ThreeSUM{} with real-valued inputs, in a reasonable Real RAM model, for any $\varepsilon > 0$. 

The integer version of {\sf 3SUM} has gained more attention
after a groundbreaking result by P{\u{a}}tra{\c{s}}cu~\cite{patrascu2010towards} which showed that the Integer $3$SUM hypothesis implies lower bounds for many dynamic problems that may not even have numbers in their inputs. Notably, his reduction uses a hashing technique which doesn't quite apply to \ThreeSUM. Thus, these hardness results and many subsequent results~\cite{kopelowitz2016higher, abboud2014popular} following P{\u{a}}tra{\c{s}}cu's paper were not known to be true under the Real $3$SUM hypothesis. Thus, an intriguing question is whether these problems are also hard under the Real $3$SUM hypothesis.

\ThreeSUM\ could conceivably be harder than \IntThreeSUM. 
Baran, Demaine and 
P{\u{a}}tra{\c{s}}cu gave an  $n^2 (\log \log n)^{O(1)} / \log^2 n$ time algorithm for \IntThreeSUM\ as early as 2005 \cite{baran2005subquadratic}. However, \ThreeSUM\ did not have an algorithm with the same asymptotic running time (up to $(\log \log n)^{O(1)}$ factors) until more than ten years later \cite{chan3sum}. Also, many techniques that are useful for \IntThreeSUM\ stop working for \ThreeSUM, for example, the hashing technique used by Baran, Demaine and 
P{\u{a}}tra{\c{s}}cu.
Therefore, the Real $3$SUM hypothesis is arguably weaker than the Integer $3$SUM hypothesis.

\paragraph{Real APSP hypothesis. } 
Similarly, one  could naturally consider the APSP problem where the edge weights of the input graph are real numbers. In fact, 
historically, the APSP problem was first studied when the edge weights are real numbers \cite{floyd1962algorithm}.
Fredman~\cite{fredman1976new} gave the first slightly subcubic time algorithm for \APSP, running in $n^3 (\log\log n)^{O(1)}/(\log n)^{1/3}$ time.  Fredman obtained his result by first proving a truly subcubic upper bound, near $n^{5/2}$, on the linear-decision-tree complexity of \APSP\ (using what is now known as ``Fredman's trick'' that later inspired Gr{\o}nlund and Pettie's work on \ThreeSUM\@).
After a series of further improved poly-logarithmic speedups, Williams~\cite{Williams18} developed the current fastest algorithm for \APSP, with running time $n^3/2^{\Omega(\sqrt{\log n})}$ (which was subsequently derandomized by Chan and Williams~\cite{ChanW21}).
A truly subcubic algorithm remains open.

The Real APSP hypothesis states that no $O(n^{3-\eps})$ time algorithm exists for \APSP\ in graphs with real-valued weights for any $\eps > 0$ in a reasonable  Real RAM model. 

Although the current upper bounds for \APSP\ and \IntAPSP\ are the same, \IntAPSP\ could conceivably be easier than \APSP.  For example, in Williams' APSP paper \cite{Williams18}, he described a ``relatively short argument'' for an $n^3/2^{\Omega(\log^\delta n)}$ algorithm in the integer case as a warm-up, which did not immediately generalize to the real case (which required further ideas).

It is known \cite{Fischer71} that \IntAPSP\ (resp. \APSP) is equivalent to the seemingly simpler problem \IntMinPlus\ (resp. \MinPlus), where one is given two $n \times n$ integer (resp.\ real) matrices $A, B$, and is asked to compute an $n \times n$ matrix $C$ where $C[i,j] = \min_{k \in [n]} (A[i,k]+B[k,j])$. Thus, {\sf $(\min, +)$-Product} has been the main focus of algorithms and reductions for {\sf APSP}.

\paragraph{The Orthogonal Vectors (OV) hypothesis.} In the \OV\ problem, given a set of $n$ Boolean  vectors in $d$-dimension for some $d = \omega(\log n)$, one needs to determine if there are two vectors $u, v$ such that $\bigvee_{i=1}^d (u[i] \wedge v[i])$ is false. The OV hypothesis (OVH) states that no $O(n^{2-\eps})$  time algorithm exists for \OV\ for any $\eps > 0$. 
Williams~\cite{Williams05} showed that SETH implies OVH. In fact, a large fraction of the conditional lower bounds based on SETH actually use the \OV\ problem as an intermediate problem (see the survey \cite{virgisurvey} for an overview of problems that are hard under OVH and SETH). 

In the \kOV{$k$}\ problem we are given $k$ sets of Boolean vectors of small dimension and are asked if there exists a $k$-tuple of vectors, one from each set, that is orthogonal, i.e. there is no coordinate in which all $k$ vectors have $1$s. 
Williams' argument~\cite{Williams05} in fact implies a fine-grained reduction from \CNFSAT\ with $n$ variables and $O(n)$ clauses to \kOV{$k$}\ for any integer $k\geq 2$, so that under SETH, \kOV{$k$}\ requires $n^{k-o(1)}$ time. 

It is easy to reduce \kOV{$k$}\ for any $k>2$ to \OV, however a reduction in the reverse direction has been elusive. It is quite possible that, say \kOV{$1000$}\ has an $O(n^{999.99})$ time algorithm, yet \OV\ still requires $n^{2-o(1)}$ time. Thus basing hardness on \OV\ is arguably better than basing hardness on \kOV{$k$}\ for $k>2$, and basing hardness on any fixed \kOV{$k$}\ is better than basing hardness on SETH, as all one needs to do to refute SETH is to refute the presumed hardness of \kOV{$k$}\ for some $k$.

\subsection{Our results}
We begin by defining our new hypothesis and will then outline our reductions. 

Combining OVH, the  Real $3$SUM hypothesis and the Real APSP hypothesis, we consider the following very weak hypothesis. 
\begin{hypothesis}
\label{conj:conj2}
At least one of OVH, the Real $3$SUM hypothesis or the Real APSP hypothesis is true.
\end{hypothesis}

\subsubsection{Main result: hardness for Triangle Collection}
It is unclear whether any of the Real $3$SUM hypothesis, the Real APSP hypothesis, and the OV hypothesis (let alone Hypothesis~\ref{conj:conj2}) implies nontrivial lower bounds for the \TCol\ problem. Abboud, Vassilevska W. and Yu's prior reductions \cite{abboud2018matching}  from \IntThreeSUM\ and \IntAPSP\ to \TCol\ used hashing tricks that are not applicable to real inputs; furthermore, their reduction %
from \CNFSAT\ to \TCol\ does not go through \OV, but was from the stronger \kOV{$3$}\ problem.

The proof of Theorem~\ref{thm:intro:abboud} in \cite{abboud2018matching} 
 uses an intermediate problem, \IntExactTri, between \IntThreeSUM\ / \IntAPSP\ and \TCol.
In \IntExactTri, one is given an $n$-node graph with edge weights in $\pm[n^c]$ for some constant $c$ and one needs to determine whether the graph contains a triangle whose edge weights sum up to zero. 
One can define \ExactTri\ analogously by replacing edge weights with real numbers. 
The Integer (resp. Real) Exact-Triangle hypothesis states that no $O(n^{3-\eps})$ time algorithm for $\eps > 0$ exists for \IntExactTri\ (resp.\ \ExactTri) in the Word RAM model (resp.\ a reasonable Real RAM model), and it is well-known that the Integer $3$SUM hypothesis and the Integer APSP hypothesis imply the Integer Exact-Triangle hypothesis \cite{VWfindingcountingj}. 

If one tries to mimic the approach of going through \IntExactTri\ in the real case, one would first reduce \APSP\ and \ThreeSUM\ to \ExactTri, and then reduce \ExactTri\ to \TCol. However, it was unclear whether \APSP\ reduces to \ExactTri, since the previous reduction for the integer case relies on fixing the results bit by bit  \cite{VWfindingcountingj}.
The tight reduction from \IntThreeSUM\ to \IntExactTri\ also doesn't seem to apply since it relies on P{\u{a}}tra{\c{s}}cu's hashing technique \cite{patrascu2010towards}. (However, a previous non-tight reduction from  \IntThreeSUM\ to \IntExactTri\ by Vassilevska W. and Williams~\cite{VWfindingcountingj} does generalize to a reduction from \ThreeSUM\ to \ExactTri; see Appendix~\ref{sec:3sum2exactT} for more details.) 

Secondly, even if one manages to reduce \APSP\ to \ExactTri, one still faces an obstacle in reducing further to \TCol: the proof in \cite{abboud2018matching}  relies on transforming the integer edge weights into vectors of tiny integers, which isn't possible if the edge weights are real numbers.

The reduction from SETH to \TCol\ in \cite{abboud2018matching}  implicitly goes through the \kOV{$3$}\ problem. 
It is very natural to relate \kOV{$3$}\ to \TCol, 
since \kOV{$3$}\ asks whether all triples of vectors are not orthogonal, and \TCol\ asks whether all triples of colors have a triangle. It is then sufficient to embed \kOV{$3$}\ into \TCol\ so that a triple of vectors are not orthogonal if and only if their corresponding colors have a triangle. However, \OV\ and \TCol\ are seemingly conceptually  different (\OV\ is about pairs, whereas \TCol\ is about triples) and thus it is not quite clear whether one can reduce \OV\ to \TCol.

We consider a natural generalization of the \TCol\ problem: \ACPTCb\ (``ACP'' stands for ``All-Color-Pairs''). The input of \ACPTCb\ is the same as the input of \TCol, while \ACPTCb\ requires the algorithm to output for every
pair of distinct colors $(a, b)$, whether there exists a triangle with colors $(a, b, c)$ for every $c$ different from $a$ and $b$. 

As our main result, we give a nontrivial fine-grained lower bound for \ACPTCb\ based on the very weak Hypothesis~\ref{conj:conj2}.

\begin{restatable}{theorem}{mainThm}
\label{thm:intro:ACP-TC}
Assuming Hypothesis~\ref{conj:conj2}, \ACPTCb\  
requires $n^{2+\delta-o(1)}$ time for some $\delta > 0$, in a reasonable Real RAM model. 
\end{restatable}

Our reductions prove the above theorem for $\delta = 0.25$. 
Combining Theorem~\ref{thm:intro:ACP-TC} with the reductions by Abboud, Vassilevska W. and Yu \cite{abboud2018matching}, we obtain the following corollary. In the following, the {\sf \#SS-Sub-Conn}  problem  asks to maintain the number of nodes reachable from a fixed source in an undirected graph under node updates, and the {\sf ST-Max-Flow}  problem  asks to compute the maximum flow in a graph for every pair of vertices $(s, t) \in S \times T$ for two given subsets of nodes $S, T$.

\begin{corollary}\label{cor:dyn}
Assuming Hypothesis~\ref{conj:conj2},  there exists a constant $\delta > 0$ such that any fully dynamic algorithm for {\sf \#SSR}, {\sf \#SCC}, {\sf \#SS-Sub-Conn}, and {\sf Max-Flow}
requires either amortized $n^{\delta-o(1)}$ update or query times, or
$n^{2+\delta-o(1)}$ preprocessing time.

Also, assuming Hypothesis~\ref{conj:conj2}, {\sf ST-Max-Flow}
on a network with $n$ nodes and $O(n)$ edges requires $n^{1+\delta-o(1)}$ time for some $\delta > 0$, even when $|S| = |T| = \sqrt{n}$. 
\end{corollary}

We obtain Theorem~\ref{thm:intro:ACP-TC} via two conceptually different webs of reductions, together with equivalence results between variants of Triangle Collection. 

\subsubsection{Real APSP and Real \texorpdfstring{$3$}{3}SUM hardness via All-Edges Sparse Triangle} In the All-Edges Sparse Triangle problem, which we will abbreviate as \AESparseTri,
one is given a graph with $m$ edges, and is asked whether each edge in the graph is in a triangle. This problem has an $O(m^{2\omega / (\omega + 1)})$ time algorithm \cite{AlonYZ97} where $\omega$ is the square matrix multiplication exponent. This running time is $O(m^{1.41})$
with the current bound $\omega < 2.373$ \cite{vstoc12, legallmult, alman2021refined}, and $\tO(m^{4/3})$ if $\omega = 2$.

P{\u{a}}tra{\c{s}}cu \cite{patrascu2010towards} showed that this problem requires $m^{4/3 - o(1)}$ time assuming the integer $3$SUM hypothesis. Kopelowitz, Pettie, and Porat \cite{kopelowitz2016higher} extended P{\u{a}}tra{\c{s}}cu's result by showing conditional lower bounds for the \SetDisj\ problem and the \SetInter\ problem, which generalize
\AESparseTri. Recently, Vassilevska W. and Xu \cite{williamsxumono} showed a reduction from 
\IntExactTri\ to \AESparseTri. Combined with the known reductions from \IntAPSP\ and \IntThreeSUM\ to \IntExactTri~\cite{VWfindingcountingj,focsyj}, we get that \AESparseTri\ requires $m^{4/3-o(1)}$ time assuming either the Integer $3$SUM hypothesis or the Integer APSP hypothesis. 

These reductions fail under the Real APSP hypothesis or the Real $3$SUM hypothesis. P{\u{a}}tra{\c{s}}cu's and Kopelowitz, Pettie, and Porat's reductions rely heavily on hashing, which does not seem to apply for \ThreeSUM. Vassilevska W. and Xu's reduction uses \IntExactTri\ as an intermediate problem. As discussed earlier, it is unclear how to reduce \APSP\ to \ExactTri. Even if such a reduction is possible, one still needs to replace a hashing trick used by Vassilevska W. and Xu with some other technique that works for real numbers. 

We overcome these difficulties by designing conceptually very different reductions from before. On a very high level, instead of hashing, we use techniques inspired by ``Fredman's trick'' \cite{fredman1976new}. Using our new techniques, we obtain the following theorem. The resulting reduction also appears simpler than previous reductions, partly because we don't need to design and analyze any hash functions.  For example, the entire proof of our reduction from
\APSP\ to \AESparseTri\ fits in under two pages, and is included at the end of the introduction in Section~\ref{sec:proof_in_intro}.  The reader is invited to browse through Vassilevska W. and Xu's
longer, more complicated proof~\cite{williamsxumono} for a comparison.

\begin{theorem}
\label{thm:intro:AEsparse}
\AESparseTri\ on a graph with $m$ edges requires
\begin{itemize}
    \item $m^{4/3-o(1)}$ time assuming the Real APSP hypothesis;
    \item $m^{5/4-o(1)}$ time assuming the Real Exact-Triangle hypothesis; 
    \item $m^{6/5-o(1)}$ time assuming the Real $3$SUM hypothesis.
\end{itemize}
\end{theorem}

Theorem~\ref{thm:intro:AEsparse} immediately implies Real APSP and Real $3$SUM hardness for a large list of problems that were shown to be Integer APSP and Integer $3$SUM hard via \AESparseTri, such as dynamic reachability, dynamic shortest paths and Pagh’s problem \cite{patrascu2010towards, abboud2014popular}.

We obtain higher conditional lower bounds if we consider the counting version of \AESparseTri, \AESparseTriCount, where one is given a graph with $m$ edges, and is asked to output the number of triangles each edge is in. Note that the $O(m^{2\omega / (\omega + 1)})$ time algorithm by Alon, Yuster, and Zwick~\cite{AlonYZ97} still works for \AESparseTriCount. Therefore, the following theorem is tight if $\omega = 2$. 

\begin{theorem}
\label{thm:intro:AEsparseC}
\AESparseTriCount\ on a graph with $m$ edges requires $m^{4/3-o(1)}$ time if at least one of the Real APSP hypothesis,  the Real Exact-Triangle hypothesis or the Real $3$SUM hypothesis is true. 
\end{theorem}

We actually obtain slightly stronger results than Theorem~\ref{thm:intro:AEsparse}: we can reduce each of \APSP, \ExactTri\ and \ThreeSUM\ to some number of  instances of \AESparseTri.
Consequently, we obtain Real APSP and Real $3$SUM hardness for a problem called the \AEMonoTri\ (i.e.\ All-Edges Monochromatic Triangles, defined in Section~\ref{sec:prelim}),  
by combining a known reduction by Lincoln, Polak, and Vassilevska W.~\cite{lincoln2020monochromatic} from multiple instances of \AESparseTri\ to \AEMonoTri.

Finally, we will reduce \AEMonoTri\ to \ACPTCb, thus proving the Real APSP and Real $3$SUM hardness in Theorem~\ref{thm:intro:ACP-TC}.

\subsubsection{Real APSP and OV hardness via Colorful Boolean Matrix Multiplication}

As a key problem in our second line of reductions, we define a natural generalization of the Boolean Matrix Multiplication problem, \ColorBMM\@. In the \ColorBMM\  problem, we are given an $n\times n$ Boolean matrix $A$ and an
$n\times n$ Boolean matrix $B$ and a mapping $\col:[n]\rightarrow \Gamma$.  For each $i\in [n]$ and $j\in [n]$,
we want to decide whether $\{\col(k): A[i,k]\wedge B[k,j],\ k\in [n]\} = \Gamma$. In other words, for every pair of $i, j$, we want to determine whether the witnesses cover all the colors. 

We show that \APSP\ can be reduced to the \ColorBMM\  problem. Our reduction uses an idea from Williams' algorithm for \APSP\  \cite{Williams18}: to compute the Min-Plus product of two real matrices, it suffices to compute several logical ANDs of ORs. We then show that these logical operations can be naturally reduced to \ColorBMM\@. 
Since \OV\ also has a similar formulation of logical ANDs of ORs~\cite{abboud2014more}, we similarly reduce \OV\ to \ColorBMM\@. We obtain the following theorem using this idea. 

\begin{theorem}
\ColorBMM\ between two $n \times n$ matrices requires 
\begin{itemize}
    \item $n^{2.25 - o(1)}$ time assuming the Real APSP hypothesis;
    \item $n^{3-o(1)}$ time assuming OVH. 
\end{itemize}
\end{theorem}

We also obtain \APSP\ and \OV\ hardness for several natural matrix product problems, such as \DistinctEq\ and \DistinctPlus\ (see Section~\ref{sec:final} for their definitions).

We will then reduce \ColorBMM\ to \ACPTCb,\ as detailed in the following. 

\subsubsection{Equivalence between variants of Triangle Collection}

We also show some equivalence results between variants of \TCol. These equivalences will be useful when we further reduce the previous two lines of reductions to \ACPTCb\ to finish the proof of Theorem~\ref{thm:intro:ACP-TC}.

\TClong\ (\TC\ for short) as defined by \cite{abboud2018matching} is a ``restricted'' version of \TCol\ whose definition has two parameters $t$ and $p$ along with  other details about the structure of the input graphs. All previous reductions from \TCol\ \cite{abboud2018matching, dahlgaard2016hardness} are actually from the \TC\ problem with small parameters $t, p \le n^{o(1)}$ (or $t, p \le n^\eps$ for every $\eps > 0$). 

We consider a conceptually much simpler variant of \TCol,\  which we call \TCsss. The input and output of \TCsss\ are the same as those of \TCol;\ however, we use a parameter which denotes an upper bound for the number of nodes that can share the same color (``light'' means that the colors are light, i.e.\ have few vertices each). 

We show the following equivalence between \TC\ and \TCsss. 

\begin{theorem}
\label{thm:intro:TC***}
({\sf ACP-}) \TCsss\ with parameter $n^{o(1)}$ and ({\sf ACP-}) \TC\ with $t, p \le n^{o(1)}$ are equivalent up to $n^{o(1)}$ factors. 
\end{theorem}

Therefore, in order to reduce problems to \TC\ and thus to other problems known to be reducible from \TC, it suffices to reduce them to the much cleaner problem \TCsss.

We obtain our main result, Theorem~\ref{thm:intro:ACP-TC}, by combining our previous reductions with Theorem~\ref{thm:intro:TC***}. 

First, we obtain \APSP\ and \ThreeSUM\ hardness of  \ACPTC\  
by reducing \AEMonoTri\ to \ACPTCsss. 

Second, the \ColorBMM\ problem easily reduces to the \ACPTCb\ problem, establishing the \OV\ hardness of \ACPTCb\  and yielding another route of reduction from the \APSP\  problem. However, \ColorBMM\ doesn't seem to reduce to the more restricted problem \ACPTC. By unrolling our reduction from \OV\ to \ACPTCb, we show that \OV\ actually reduces to the original \TC\ problem.

\begin{theorem}
\TCsss\ with parameter $n^{o(1)}$ (and thus \TC\ with parameters $n^{o(1)}$ and \TCol) requires $n^{3-o(1)}$ time assuming OVH. 
\end{theorem}

Using Theorem~\ref{thm:intro:TC***} we are also able to prove the following surprising result: 
\begin{theorem}
\label{thm:intro:TC*}
If \TC\ with parameters $t, p \le n^\eps$ for some $\eps > 0$ has a truly subcubic time algorithm, then so does \TCol. 
\end{theorem}
Theorem~\ref{thm:intro:TC*} establishes a subcubic equivalence between \TCol\  and \TC\ with parameters $n^\eps$, and in fact implies that all the known hardness results so far that were proven from \TC\ also hold from \TCol\  itself.

\subsubsection{Other Reductions}

Using our reductions and techniques, we also obtain the following list of interesting applications.

\paragraph{A hard colorful version of AE-Sparse-Triangle.} We give a \textit{tight} conditional lower bound under Hypothesis~\ref{conj:conj2} for the parameterized time complexity of a natural variant of \AESparseTri. Specifically, in the
\AEColorSparseTri\ problem, we are  given a graph $G=(V,E)$ with $m$ edges
and a mapping $\col:V\rightarrow\Gamma$.  For each edge $uv$, we want to decide whether
$\{\col(w): uwv\mbox{ is a triangle}\} = \Gamma$. When the degeneracy of the graph is $m^\alpha$, we can clearly solve the problem in $\OO(m^{1+\alpha})$ time by enumerating all triangles in the graph \cite{chiba1985}. We show that, under Hypothesis~\ref{conj:conj2}, for any constant $0 < \alpha \le 1/5$, no algorithm can solve
\AEColorSparseTri\ in a graph with $m$ edges and degeneracy $O(m^\alpha)$ in $\OO(m^{1+\alpha-\eps})$ time for $\eps > 0$.

\paragraph{Real-to-integer reductions.} It is an intriguing question whether we can base the hardness of a problem with integer inputs on the hardness of the same problem but with real inputs. For instance, it would be extremely interesting if one could show that the Real $3$SUM hypothesis implies the Integer $3$SUM hypothesis. We partially answer this question by showing two conditional lower bounds of this nature. 

First, if \IntAllThreeSUM\ (a variant of Integer $3$SUM where one needs to output whether each input number is in a $3$SUM solution) can be solved in $\OO(n^{6/5 - \varepsilon})$ time for $\varepsilon > 0$, then \AllThreeSUM\ can be solved in truly subquadratic time, falsifying the Real $3$SUM hypothesis. 

Second, if \IntAEExactTri\ (the ``All-Edges'' variant of \IntExactTri) can be solved in $\OO(n^{7/3-\varepsilon})$ time for $\varepsilon > 0$, then \AEExactTri\ can be solved in truly subcubic time. We also show variants of this result such as an analogous result for the counting versions of these problems. 

\paragraph{An application to string matching.} 

In the pattern-to-text Hamming distance problem, we are given a text string $T = t_1\cdots t_N$ and a pattern string $P = p_1\cdots p_M$ in $\Sigma^*$ with $M\le N$, and we want to compute for every $i=0, \ldots, N-M$, the Hamming distance between $P$ and $t_{i+1}\cdots t_{i+M}$, which is defined as $M-|\{j:p_j=t_{i+j}\}|$. The current best algorithm in terms of $N$ runs in $\OO(N^{3/2})$ time \cite{matching1} while unfortunately there isn't a matching conditional lower bound under standard hypotheses (though there is an unpublished non-matching $N^{\omega/2-o(1)}$ time conditional lower bound based on the presumed hardness of Boolean Matrix Multiplication that has been attributed to Indyk, see e.g. \cite{GawrychowskiU18}).

We consider a similar string matching problem which we call pattern-to-text distinct Hamming similarity. The input to pattern-to-text distinct Hamming similarity is the same as the input to pattern-to-text Hamming distance, but for each $i = 0, \ldots, N-M$, we need to output $|\{p_j:p_j=t_{i+j}\}|$ instead. The $\OO(N^{3/2})$ time algorithm for pattern-to-text Hamming distance can be easily adapted to an $\OO(N^{3/2})$ time algorithm for pattern-to-text distinct Hamming similarity. Using our reduction from \OV\ to \ColorBMM, we show a matching $N^{3/2-o(1)}$ lower bound for pattern-to-text distinct Hamming similarity based on OVH. 

\paragraph{Real APSP hardness of Set-Disjointness and Set-Intersection.}
\SetDisj\ and \SetInter\ are two generalized versions of \AESparseTri\ (see Section~\ref{sec:prelim} for their formal definitions). Kopelowitz, Pettie, and Porat~\cite{kopelowitz2016higher} showed Integer $3$SUM hardness of these two problems, and used them as intermediate steps for showing $3$SUM hardness of many graph problems, such as {\sf Triangle Enumeration} and {\sf  Maximum Cardinality Matching}. Later, Vassilevska W. and Xu \cite{williamsxumono} showed reductions from \IntExactTri\ to these two problems, and thus obtained Integer APSP hardness for \SetDisj\ and \SetInter. By generalizing the techniques used in our reduction from \APSP\ to \AESparseTri, we obtain Real APSP hardness for \SetDisj\ and \SetInter. Therefore, all the hardness results shown by Kopelowitz, Pettie, and Porat~\cite{kopelowitz2016higher} now also have Real APSP hardness. 

\subsection{An Illustration of Our Techniques: Reduction from \APSP{} to \AESparseTri{}}
\label{sec:proof_in_intro}

In this subsection, we include our complete reduction from \APSP{} to \AESparseTri{} to exemplify how our techniques are different from, and simpler than, the previous hashing techniques~\cite{williamsxumono}. %

Our main new insight is simple: we observe that
Fredman's beautiful method for \APSP~\cite{fredman1976new}, which yielded an $\tO(n^{5/2})$-depth decision tree but not a truly subcubic time algorithm, can actually be converted to an 
efficient algorithm when given an oracle to \AESparseTri, if we combine the method with a
standard randomized search trick and an interesting use of ``dyadic intervals'' (essentially corresponding to a one-dimensional ``range tree''~\cite{BergCKO08}).

We first recall the well-known fact~\cite{Fischer71} that \APSP{} is equivalent to computing the $(\min,+)$-product of
two $n\times n$ real-valued matrices $A$ and $B$, defined as the matrix $C$ with $C[i,j] = \min_{k} (A[i,k]+B[k,j])$. 
This problem in turn reduces to $O(n/d)$ instances of computing the $(\min,+)$-product of
an $n\times d$ matrix and $d\times n$
matrix.

We define the following intermediate problem, which we show is equivalent to the original problem by a random sampling trick:

\newcommand{\APSPVar}{\mbox{\sf Real-$(\min,+)$-Product-Variant}}

\begin{problem}[\APSPVar]\label{prob:apsp2}
We are given an $n\times d$ real matrix $A$ and 
a $d\times n$ real matrix $B$, where $d\le n$.
For each $i,j\in [n]$, we are also given an index $k_{ij}\in [d]$.

For each $i,j\in [n]$, we want to find an index $k'_{ij}\in [d]$ (if it exists) satisfying
\begin{equation}\label{eqn:apsp}
A[i,k'_{ij}]+B[k'_{ij},j] < A[i,k_{ij}]+B[k_{ij},j]. 
\end{equation}
\end{problem}

\begin{restatable}{lemma}{lemAPSPone}
\label{lem:apsp1}
\MinPlus\ of an $n\times d$ real matrix $A$ and 
a $d\times n$ real matrix $B$ reduces to $\OO(1)$ calls of an oracle for
\APSPVar\
using Las Vegas randomization.
\end{restatable}

\begin{proof}
We compute the $(\min,+)$-product of $A$ and $B$ as follows:
take a random subset $R\subseteq [d]$ of size $d/2$.
For each $i,j\in [n]$, first compute
\[ k_{ij}^{(R)} = \argmin_{k\in R} (A[i,k]+B[k,j]). \]
This can be done recursively by $(\min,+)$-multiplying an
$n\times (d/2)$ and a $(d/2)\times n$ matrix.

Next, initialize $k_{ij}=k_{ij}^{(R)}$.  Invoke the oracle for \APSPVar\ to find some $k_{ij}'$ satisfying (\ref{eqn:apsp}) for each $i,j\in [n]$.   If $k_{ij}'$ exists, reset $k_{ij}=k_{ij}'$.  Now, repeat.  When $k_{ij}'$ is nonexistent for all $i,j\in [n]$, we can
stop, since we would have $k_{ij}=\argmin_{k\in[d]} (A[i,k]+B[k,j])$ for all $i,j\in[n]$.

To bound the number of iterations, observe that for each $i,j\in [n]$,
the number of indices $k'\in[d]$ satisfying 
$A[i,k']+B[k',j] < A[i,k_{ij}^{(R)}]+B[k_{ij}^{(R)},j]$
is $O(\log n)$ w.h.p.%
\footnote{``w.h.p.'' is short for ``with high probability'', i.e., with probability $1-O(1/n^c)$ for an arbitrarily large constant $c$.}\  
(since in a set of $d$ values, at most $O(\log n)$ values are smaller than the minimum
of a random subset of size $d/2$ w.h.p.).  Thus, $O(\log n)$ iterations suffice w.h.p.

As the recursion has $O(\log d)$ depth, the total number of
oracle calls is $O(\log d \log n)$ w.h.p.  
\end{proof}

We now present our main reduction from \APSPVar\ to
\AESparseTri, which is inspired by ``Fredman's trick''%
~\cite{fredman1976new}---namely, the obvious but crucial observation that
$a'+b' < a+b$ is equivalent to $a'-a < b-b'$.

\begin{lemma}\label{lem:apsp2}
\APSPVar\ reduces to $O(d)$  instances of
 \AESparseTri\ on graphs with $\OO(n^2/d+dn)$ edges.
\end{lemma}
\begin{proof}
To solve \APSPVar, we first sort the following list of $O(d^2n)$ elements in $O(d^2n\log n)$ time:
\[ L\ =\ \{A[i,k'] - A[i,k] : i\in [n],\, k,k'\in [d]\} \ \cup\ 
     \{B[k,j] - B[k',j] : j\in [n],\, k,k'\in [d]\}. \]

Fix $k\in[d]$.  Let
$P_k=\{(i,j)\in [n]^2: k_{ij}=k\}$.
Divide $P_k$ into $\up{\frac{|P_k|}{n^2/d}}$ subsets
of size $O(n^2/d)$.
Fix one such subset $P\subseteq P_k$.
We solve \AESparseTri\ on the following tripartite graph $G_{k,P}$:
\begin{enumerate}
\item[0.] The left nodes are $\{x[i]: i\in [n]\}$, the middle nodes are $\{y[k',I]: k'\in [d],\ \mbox{$I$ is a dyadic interval}\}$, and
the right nodes are $\{z[j]: j\in [n]\}$.  Here, a \emph{dyadic} interval
refers to an interval of the form $[2^s t, 2^s(t+1))\subset [0,4d^2n)$ for nonnegative
integers $s$ and $t$ (here, $4d^2n-1$ is an upper bound for the smallest power of $2$ that is larger than or equal to $2d^2n$).  We don't explicitly enumerate all these nodes, but generate a node when we need to add an edge to it. 
\item For each $(i,j)\in P$, create an edge $x[i]\,z[j]$.
\item For each $i\in [n]$, $k'\in [d]$, and each dyadic interval $I$,
create an edge $x[i]\, y[k',I]$ if the rank of $A[i,k']-A[i,k]$ in $L$ is in the left half of $I$.
\item For each $j\in [n]$, $k'\in [d]$, and each dyadic interval $I$,
     create an edge $y[k',I]\, z[j]$ if the rank of $B[k,j]-B[k',j]$ in $L$ is in the right half of $I$.
\end{enumerate}

Step~1 creates $O(n^2/d)$ edges.  Steps 2--3 create $O(dn\log n)$ edges, since any fixed value lies in $O(\log n)$ dyadic intervals.
Thus, the graph $G_{k,P}$ has $\OO(n^2/d + dn)$ edges.
The total number of graphs is $\sum_{k\in [d]}\up{\frac{|P_k|}{n^2/d}}=O(d)$.  

For any given $(i,j) \in [n]^2$, take $k=k_{ij}$ and $P$ to be the subset of $P_k$ containing $(i,j)$.  Finding a $k'$ with $A[i,k']+B[k',j] <  A[i,k_{ij}]+B[k_{ij},j]$ is equivalent to finding a $k'$ with
$A[i,k']-A[i,k] < B[k,j]-B[k',j]$, which is equivalent to finding
a triangle
$x[i]\,y[k',I]\,z[j]$ in the graph $G_{k,P}$.  
Here, we are using the following fact: for any two numbers $a,b\in [2d^2n]$, we have $a<b$ iff there is a (unique) dyadic interval $I$ such that $a$ lies on the left half of $I$ and $b$ lies on the right half of $I$.
Thus, the answers to
\APSPVar\ can be deduced from the answers
to the \AESparseTri\ instances.
Note that we have assumed an equivalent version of the \AESparseTri\ problem with ``witnesses'', i.e., that
outputs a triangle through each edge whenever such a triangle exists~\cite{duraj2020equivalences}.  
\end{proof}

A near $m^{4/3}$ lower bound for \AESparseTri\ immediately follows, assuming the \APSP\ hypothesis; this matches known upper bounds if $\omega=2$.  The new proof not only strengthens the previous proof by Vassilevska W. and Xu~\cite{williamsxumono}, which reduces from integer APSP, but it is also simpler (not requiring more complicated hashing arguments).

\begin{restatable}{theorem}{thmAPSP}
\label{thm:apsp}
If \AESparseTri\ could be solved in $\OO(m^{4/3-\eps})$ time, then \APSP\ could be solved in $\OO(n^{3-3\eps/2})$ time using
Las Vegas randomization.
\end{restatable}
\begin{proof}
By combining the above two lemmas, if \AESparseTri\ could be solved in $T(m)$ time, then
the $(\min,+)$-product of an $n\times d$
and a $d\times n$ real matrix could be computed in $\OO(d\cdot T(n^2/d + dn))$ 
time.  The $(\min,+)$-product of two $n\times n$
real matrices, and thus \APSP, reduce to $n/d$ instances of
such rectangular products and could then be computed in $\OO(n\cdot T(n^2/d + dn))$ time.  By choosing $d=\sqrt{n}$, the time bound becomes $\OO(n\cdot T(n^{3/2}))=\OO(n^{3-3\eps/2})$ if $T(m)=\OO(m^{4/3-\eps})$.
\end{proof}

Our reduction from \ThreeSUM\ to \AESparseTri\ (see Section~\ref{sec:3sum}) is based on a similar insight: we observe that Gr\o nlund and Pettie's elegant method (based on Fredman's work), which yielded an $\tO(n^{3/2})$-depth decision tree for \ThreeSUM, can also be converted to an efficient algorithm 
when given an oracle to \AESparseTri.  Since Gr\o nlund and Pettie's method is a bit cleverer, this reduction is technically more
challenging
(though it is still comparable in simplicity with P{\u{a}}tra{\c{s}}cu's original reduction from \IntThreeSUM~\cite{patrascu2010towards}), and because of these extra complications, the resulting conditional bounds are not tight.  Nevertheless, it is remarkable that nontrivial conditional lower bounds can be obtained at all, and that we do get tight bounds for a certain range
of degeneracy values, and tight bounds for reductions to the counting variant of \AESparseTri, when $\omega=2$.
(For applications to some of the dynamic graph problems, the earlier polynomial lower bounds weren't tight anyways.)

Our reduction from \APSP\ to \ColorBMM\ (see Section~\ref{sec:colorBMM}) is even simpler (under one page long), and is inspired by ideas from Williams' \APSP\ algorithm~\cite{Williams14a}, also based on Fredman's trick.
Our reduction from \OV\ to \ColorBMM\ is more straightforward, but we view the main innovation here to be the introduction of the
\ColorBMM\ problem itself, which we hope will find further applications.  Sometimes, the key to conditional lower bound proofs lies in formulating the right intermediate subproblems.  Our combined reduction from \OV\ to \ACPTCb\ via \ColorBMM\ is essentially as simple as
 Abboud, Vassilevska W. and Yu's original reduction from \kOV{3} to \TCol~\cite{abboud2018matching}, but adds more understanding of the \TCol{} and related problems as it reveals that a weaker hypothesis is sufficient
  to yield conditional lower bounds for all the dynamic graph problems in Corollary~\ref{cor:dyn}.

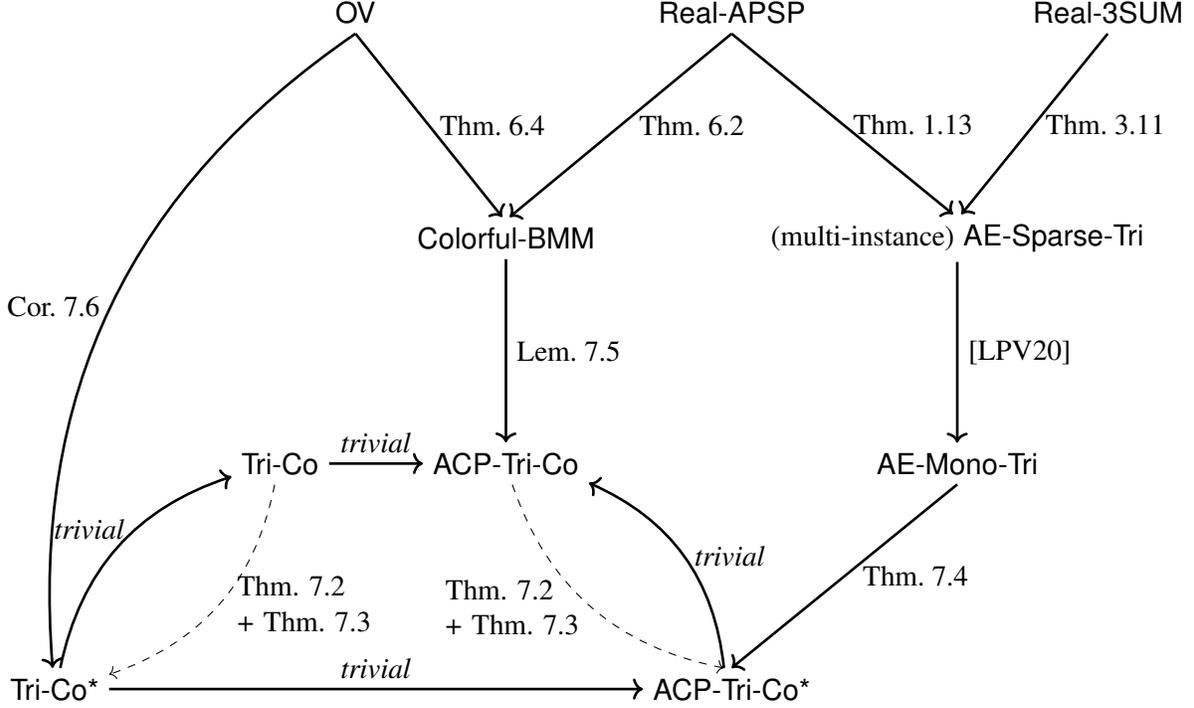
\begin{figure}[ht]
    \centering
    \begin{tikzpicture}
    
        \pgfmathsetmacro{\dx}{0}
        \pgfmathsetmacro{\dy}{0}
        
        \node at(\dx + 5, \dy)  [anchor=center] (3sum){\ThreeSUM};
        \node at(\dx, \dy)  [anchor=center] (apsp){\APSP};
        \node at(\dx - 5, \dy)  [anchor=center] (OV){\OV};
        
        \node at(\dx + 3, \dy-3)  [anchor=center] (AESparseT){(multi-instance) \AESparseTri};
        
        \node at(\dx + 3, \dy-6)  [anchor=center] (monoT){\AEMonoTri};
	
	    \node at(\dx - 3, \dy-3)  [anchor=center] (colorBMM){\ColorBMM};
	    
	     \node at(\dx , \dy - 9)  [anchor=center] (ACPTC1){\ACPTC};
	     
	      \node at(\dx-3 , \dy - 6)  [anchor=center] (ACPTC){\ACPTCb};
		
		 \node at(\dx-9 , \dy - 9)  [anchor=center] (TC1){\TC};
		 
		   \node at(\dx-6, \dy - 6)  [anchor=center] (TC){\TCol};

		  \draw[->,line width=1pt] (3sum.south) to[]  node[right] {Thm.~\ref{thm:3sum}} (AESparseT.80);
		
	 \draw[->,line width=1pt] (apsp.south) to[]  node[right] {Thm.~\ref{thm:apsp}} (AESparseT.100);
	
		 \draw[->,line width=1pt] (AESparseT.south) to[]  node[right] {\cite{lincoln2020monochromatic}} (monoT.north);
		 
		 \draw[->,line width=1pt] (monoT.south) to[]  node[right, xshift=0.1cm] {Thm.~\ref{thm:mono:ACPTC}} (ACPTC1.north);

		 \draw[->,line width=1pt, bend right] (ACPTC1.110) to[]  node[right] {\textit{trivial}} (ACPTC);
		 
		 \draw[->,dashed, bend right] (ACPTC) to[]  node[left, text width=1.8cm]  {Thm.~\ref{thm:TCLC:TC**} \ + Thm.~\ref{thm:TC:TCLC}} (ACPTC1.110);

	   \draw[->,line width=1pt] (TC1) to[]  node[above] {\textit{trivial}} (ACPTC1);

	   \draw[->,line width=1pt] (OV.south) to[]  node[right] {Thm.~\ref{thm:colorbmm:ov}} (colorBMM.100);
	   
	   \draw[->,line width=1pt] (apsp.south) to[]  node[right, xshift=0.1cm] {Thm.~\ref{thm:distincteq}} (colorBMM.80);
	   
	   \draw[->,line width=1pt] (colorBMM.south) to[]  node[right] {Lem.~\ref{lem:colorBMM:ACPTC}} (ACPTC);
	   
	   \draw[->,line width=1pt, bend right] (OV.south) to[]  node[left] {Cor.~\ref{cor:OV:TC}} (TC1);
	   
	   \draw[->,line width=1pt, bend left] (TC1) to[]  node[pos=0.6, left] {\textit{trivial}} (TC);
	   
	   \draw[->,line width=1pt] (TC) to[]  node[above] {\textit{trivial}} (ACPTC);
	   
	   \draw[->,dashed, bend left] (TC) to[]  node[right, xshift=0.1cm,text width=1.8cm] {Thm.~\ref{thm:TCLC:TC**} \ + Thm.~\ref{thm:TC:TCLC}} (TC1);
    \end{tikzpicture}
    \caption{The main reductions in this paper. The solid arrows represent usual fine-grained reductions. The two dashed arrows represent sub-cubic reductions that only hold when \TC\ and \ACPTC\ have parameters $n^\eps$ for $\eps > 0$. }
    \label{fig:main}
\end{figure}

\begin{table}[ht]
\small
\centering
\begin{tabular}{| c | cc | c | c|}
\hline
Problems &  \multicolumn{2}{c|}{Lower Bounds }  & Hypotheses & References \\ 
\hline

\multirow{5}*{\AESparseTri} & $m^{4/3-o(1)}$ & $\dag$ & \IntThreeSUM & \cite{patrascu2010towards}\\
\cline{2-5}
 & $m^{4/3-o(1)}$ & $\dag$ & \IntAPSP, \IntExactTri & \cite{williamsxumono}\\
 \cline{2-5}
 & $m^{6/5-o(1)}$ && \ThreeSUM & Thm.~\ref{thm:3sum}\\
 \cline{2-5}
 & $m^{5/4-o(1)}$ && \ExactTri & Thm.~\ref{thm:exacttri}\\
 \cline{2-5}
 & $m^{4/3-o(1)}$ & $\dag$ & \APSP & Thm.~\ref{thm:apsp}\\
 
\hline

\multirow{2}*{\AESparseTriCount} & $m^{4/3-o(1)}$ & $\dag$ & \ThreeSUM & Thm.~\ref{thm:3sum:count} \\
\cline{2-5}
& $m^{4/3-o(1)}$ & $\dag$ & \ExactTri &  Thm.~\ref{thm:exacttri:count}\\

\hline

\multirow{5}*{\AEMonoTri} & $n^{5/2-o(1)}$ & $\dag$ & \IntThreeSUM & \cite{lincoln2020monochromatic}\\
\cline{2-5}
& $n^{5/2-o(1)}$ & $\dag$ & \IntAPSP & \cite{williamsxumono} \\
\cline{2-5}
& $n^{9/4-o(1)}$ && \ThreeSUM & Thm.~\ref{thm:3sum:monotrinoncount} \\
\cline{2-5}
& $n^{7/3-o(1)}$ && \ExactTri & Thm.~\ref{thm:exacttri:monotrinoncount} \\
\cline{2-5}
& $n^{5/2-o(1)}$ & $\dag$ & \APSP & Thm.~\ref{thm:apsp:monotri} \\

\hline

\multirow{2}*{\AEMonoTriCount} & $n^{5/2-o(1)}$ & $\dag$ & \ExactTri & Thm.~\ref{thm:exacttri:monotricount}\\
\cline{2-5}
& $n^{5/2-o(1)}$ & $\dag$ & \ThreeSUM & Thm.~\ref{thm:3sum:monotri}\\

\hline

\multirow{2}*{\ColorBMM} & $n^{9/4-o(1)}$ && \APSP & Thm.~\ref{thm:distincteq}\\
\cline{2-5}
& $n^{3-o(1)}$ & $\ddag$ & \OV & Thm.~\ref{thm:colorbmm:ov}\\

\hline

\TCol & $n^{3-o(1)}$ & $\ddag$ & \IntAPSP, \IntThreeSUM, {\sf SETH} & \cite{abboud2018matching}\\

\hline

\multirow{3}*{\ACPTCb} & $n^{9/4-o(1)}$ && \ThreeSUM & Thm.~\ref{thm:3sum:monotrinoncount} + Thm.~\ref{thm:mono:ACPTC} + Thm.~\ref{thm:TCLC:TC**}\\
\cline{2-5}
 & $n^{5/2-o(1)}$ && \APSP & Thm.~\ref{thm:apsp:monotri} + Thm.~\ref{thm:mono:ACPTC} + Thm.~\ref{thm:TCLC:TC**}\\
\cline{2-5}
& $n^{3-o(1)}$ & $\ddag$ & \OV & Cor.~\ref{cor:OV:TC}\\
\hline

\IntAllThreeSUM & $n^{6/5-o(1)}$ && \ThreeSUM & Cor.~\ref{cor:real-3sum:int-3sum}\\
\hline 

\IntAEExactTri & $n^{7/3-o(1)}$ && \AEExactTri & Cor.~\ref{cor:real-exactT:int-exactT}\\
\hline 

\begin{tabular}{@{}c@{}}pattern-to-text distinct \\ Hamming similarity\end{tabular} & $n^{3/2-o(1)}$ & $\ddag$ & \OV & Thm.~\ref{thm:OV:ptt-Hamming}\\
\hline

\end{tabular}
\caption{Summary of some of our results and previous results. Ignoring $n^{o(1)}$ factors, entries marked $\dag$ match known upper bounds assuming $\omega=2$; entries marked $\ddag$ match known upper bounds without the assumption.}\label{tbl:main}
\end{table}

\subsection{Paper Organization}

In Section~\ref{sec:prelim}, we define necessary notations. In Section~\ref{sec:sparsetri}, we show hardness of \AESparseTri\  under the Real $3$SUM, Real APSP and the Real Exact-Triangle hypotheses. In Section~\ref{sec:sparsetri:count}, we show hardness of \AESparseTriCount\ under the Real $3$SUM and the Real Exact-Triangle hypotheses. From the results in Sections~\ref{sec:sparsetri}--\ref{sec:sparsetri:count}, we show hardness of \AEMonoTri\ (and its counting version) under the Real $3$SUM, Real APSP and the Real Exact-Triangle hypotheses in Section~\ref{sec:mono}. In Section~\ref{sec:colorBMM}, we show hardness of \ColorBMM\  based on the Real APSP and the OV hypotheses. We prove our main theorem in Section~\ref{sec:TC}, by showing equivalence between variants of \TCol\ problems and reducing problems in Section~\ref{sec:sparsetri} and Section~\ref{sec:colorBMM} to variants of \TCol. Finally, in Sections~\ref{sec:other}--\ref{sec:final}, we describe further applications of
our techniques and results.
Figure~\ref{fig:main} depicts the main reductions in the paper, and Table~\ref{tbl:main} summarizes our main lower bound results.

\section{Preliminaries}
\label{sec:prelim}
In this section, we give definitions and abbreviations of terminologies and problems we will consider throughout the paper. We group related definitions together for easier navigation.

Some problems require to output whether each edge in a given graph is in some triangle with certain property. We use the prefix \mbox{\sf AE-} (short for ``All-Edges-'') to denote that we need to output a Boolean value for each edge indicating whether each edge is in such a triangle. For these problems, we can instead use the prefix \mbox{\sf \#AE-}  to indicate that we need to count the number of such triangles involving each edge. 

\subsection{Fine-Grained Reductions}
We use the notion of fine-grained reduction \cite{focsyj,virgisurvey} which is as follows. Let $A$ and $B$ be problems, and let $a(n)$ and $b(n)$ be running time functions where $n$ is the input size or a suitable measure related to the input size such as the number of nodes in a graph.

We say that $A$ is $(a,b)$-fine-grained reducible to $B$ if for every $\eps>0$ there is a $\delta>0$ and an $O(a(n)^{1-\delta})$ time algorithm that solves size (or size measure) $n$ instances of $A$ making calls to an oracle for problem $B$, so that the sizes (or size measures) of the instances of $B$ in the oracle calls are $n_1,n_2,\ldots,n_k$ and $\sum_{i=1}^k (b(n_i))^{1-\eps} \leq a(n)^{1-\delta}$.

If $A$ is $(a,b)$-fine-grained reducible to $B$, and if there is an $O(b(n)^{1-\eps})$ time algorithm for $B$ for some $\eps>0$, then there is also an $O(a(n)^{1-\delta})$ time algorithm for some $\delta>0$.

All the reductions in the paper will be fine-grained, and hence we will often omit ``fine-grained'' when we say reduction.

\subsection{Hard Problems}

In this section, we define the problems that we will consider as our hardness sources. In case a problem \mbox{\sf P} has numbers in the input, we will define the generic version of problem \mbox{\sf P} and use \mbox{\sf Int-P} to denote the version where the input numbers are integers from $\pm[n^c]$ for some sufficiently large constant $c$, and use \mbox{\sf Real-P} to denote the version where the input numbers are reals. 

\begin{problem}[\mbox{\sf 3SUM}]
Given three sets  $A, B, C$ of numbers of size $n$, determine whether there exist $3$ numbers $a \in A, b \in B, c \in C$ such that $a+b+c=0$. 
\end{problem}

\begin{problem}[\mbox{\sf All-Nums-3SUM}]
Given three sets  $A, B, C$ of numbers of size $n$, for each $c \in C$, determine whether there exist  $a \in A, b \in B$ such that $a+b+c=0$. 
\end{problem}

It is known that \IntThreeSUM\ (resp. \ThreeSUM) and \IntAllThreeSUM\ (resp. \AllThreeSUM) are subquadratically equivalent \cite{focsyj}. 
Note that a ``one-set'' version of {\sf 3SUM} (instead of the above ``three-sets'' version) has also been used in the literature, but they are known to be equivalent.
Often, it is more convenient to consider the version of the problem where the third set is negated, in which case we are seeking a triple $(a,b,c)$ with $a+b=c$.

\begin{problem}[\mbox{\sf APSP}]
Given an $n$-node directed graph whose edge weights are given as numbers and which has no negative cycles, compute the shortest path distances from $u$ to $v$ for every pair of nodes $u$ and $v$ in the graph.
\end{problem}

\begin{problem}[\mbox{\sf $(\min,+)$-Product}]
Given an $n \times d$ matrix $A$ and a $d \times n$ matrix $B$ whose entries are given as numbers, compute an $n \times n$ matrix $C$ such that $C[i,j] = \min_{k \in [d]} (A[i,k]+B[k,j])$. 
\end{problem}

\begin{problem}[\mbox{\sf AE-Negative-Triangle (AE-Neg-Tri)}]
Given an $n$-node graph whose edge weights are given as numbers, for each edge determine whether there is a triangle containing that edge whose edge weights sum up to negative. 
\end{problem}

It is known \cite{Fischer71,focsyj} that \APSP, {\sf Real-AE-Neg-Tri} and \MinPlus\ between  $n \times n$ real-valued matrices are all subcubically equivalent; similarly, \IntAPSP, \mbox{\sf Int-AE-Negative-Triangle}, and \IntMinPlus\ between  $n \times n$ integer-valued matrices are all subcubically equivalent as well.

\begin{problem}[\mbox{\sf AE-Exact-Triangle (AE-Exact-Tri)}]
Given an $n$-node graph whose edge weights are given as numbers, for each edge determine whether there is a triangle containing that edge whose edge weights sum up to zero. 
\end{problem}

It is known that the Integer $3$SUM and the Integer APSP hypotheses imply that \IntAEExactTri\ requires $n^{3-o(1)}$ time \cite{VWfindingcountingj}. However, such tight reductions aren't known for their real variants. However, we can adapt one previous non-tight reduction from \IntThreeSUM\ to  \IntAEExactTri\ \cite{VWfindingcountingj} to the real case, which implies \AEExactTri\ requires $n^{2.5-o(1)}$ time assuming the Real $3$SUM hypothesis (see more details in Appendix~\ref{sec:3sum2exactT}). 

\begin{problem}[\OV]
Given a set of $n$ Boolean vectors in $f$ dimensions, determine whether the set contains two vectors that are orthogonal. 
\end{problem}

\subsection{Triangle Collection and Variants}

We define a tripartite version of \TCollong, which is slightly different from the original definition of Abbound, Vassilevska W. and Yu~\cite{abboud2018matching}. In Section~\ref{sec:TC} we will show these two definitions are equivalent. 

For brevity, we will give short names for \TCollong\ and \TClong, \TCol\ and \TC\ respectively. Similar abbreviations will also be given for the All-Color-Pairs versions.

\begin{restatable}[\TCollong\  (\TCol)]{problem}{TriangleCollection}
Given a tripartite graph $G=(V,E)$ on partitions $A,B,C$ such that the colors of the nodes in $A$ are from a set $K_A$, the colors of the nodes in $B$ are from a set $K_B$ and the colors of the nodes in $C$ are from a set $K_C$, where $K_A\cap K_B\cap K_C=\emptyset$, and one needs to determine whether for all triples of colors $a\in K_A,b\in K_B,c\in K_C$ there exists some triangle $x,y,z\in V$ such that $\col(x)=a,\col(y)=b,\col(z)=c$. 
\end{restatable}

\begin{problem}[\ACPTCblong\ (\ACPTCb)]
Given the input to a \TCol\ instance, one needs to determine for every $a \in K_A, b \in K_B$, whether for all $c \in K_C$ there exists some triangle $x, y, z \in V$ such that $\col(x) = a, \col(y) = b, \col(z) = c$. 
\end{problem}

All other variants of \TCol\ will have the same output as \TCol, so they also naturally have an All-Color-Pairs (ACP) variant. We will only define their normal versions for conciseness. 

\begin{restatable}[\TClong\ (\TC)]{problem}{TriangleCollectionStar}
An instance of \TC\ is a restricted instance of \TCol. For  parameters $p$ and $t$ it is a node-colored graph $G$ which is a disjoint union of graphs $G_1,\ldots, G_t$. $G$ (and hence all the $G_i$s) is tripartite on partitions $A,B,C$. The aforementioned value $p$ is an upper bound on the number of nodes of any particular color in any $G_i$.

The node colors are from $[3] \times [n]$. For every $t$, the nodes of $G_t$ are:
\begin{itemize}
\item Nodes in $A$ of the form $(a,t)$ of color $(1, a)$. Note that this means that each $G_t$ has nodes of distinct colors in $A$.

\item Nodes in $B$ of the form $(b,t,j)$ of color $(2,b)$, where $j\leq p$. For each $(a,t)\in A$ and every color $(2, b)$, there is at most one node $(b,t,j)$ in $B$ that $(a,t)$ has an edge to. 

\item Nodes in $C$ of the form $(c,t,j)$ of color $(3, c)$ for $j\leq p$. For each $(a,t)\in A$ and every color $(3, c)$, there is at most one node $(c,t,j)$ in $C$ that $(a,t)$ has an edge to.

\end{itemize}
The last two bullets mean that in each $G_t$ all the neighbors of a node in $A$ have distinct colors. There is no restriction on the edges between nodes in $B$ and $C$ (beyond that the graphs $G_i$ are disjoint).

An algorithm for \TC\ needs to output whether for all triples $(a, b, c)$, there is a triangle with node colors $(1, a), (2, b), (3, c)$. 
\end{restatable}

Here is an even simpler version of \TC\ that we will call \TCsslong\ (\TCss\ for short).

\begin{restatable}[\TCsslong\ (\TCss)]{problem}{TriangleCollectionStarStar}
An instance of \TCss\ is a restricted instance of \TCol.\ For an integer parameter $t$, it is a disjoint union of graphs $G_1,\ldots,G_t$, each of which contains at most one node of every color.
\end{restatable}

Clearly, \TCss\ is a special case of \TC. 
We now define a third version of \TCol\ , \TCsss, that has \TC\ as a special case. The ``light'' part of the name stands for ``Light Colors'', meaning that each colors has few nodes.

\begin{restatable}[\TCsss]{problem}{TriangleCollectionStarStarStar}
An instance of \TCsss\ is a restricted instance of \TCol.\ For an integer parameter $p$, it is a graph that has at most $p$ nodes of any fixed color.
\end{restatable}

\subsection{Other Problems}

\begin{problem}[{\sf AE-Sparse-Triangle} (\AESparseTri)]
Given a graph with $m$ edges, for each edge determine whether it is in a triangle. 
\end{problem}

\begin{problem}[{\sf AE-Monochromatic-Triangle} (\AEMonoTri)]
Given a graph on $n$ nodes where each edge has a color, for each edge determine whether it is in a triangle whose three edges share the same color.
\end{problem}

In some of our reductions, it is convenient to assume that in \AESparseTri, for each edge where the answer is yes, a witness (a triangle through the edge) must also be provided.  This version is equivalent (up to poly-logarithmic factors) by standard random sampling techniques for witness finding \cite{AlonGMN92,seidel1995} (e.g., see the proof of Lemma~\ref{lem:exacttri23}),
which requires only Las Vegas randomization. The same is true for \AEMonoTri\ as well.

\begin{problem}[\SetDisj]
Given a universe $U$, a collection of sets $\mathcal{F} \subseteq 2^U$, and $q$ queries of the form $(F_1, F_2) \in \mathcal{F} \times \mathcal{F}$, an algorithm needs to compute whether $F_1 \cap F_2 = \emptyset$ for every query. 
\end{problem}

\begin{problem}[\SetInter]
Given a universe $U$, a collection of sets $\mathcal{F} \subseteq 2^U$, and $q$ queries of the form $(F_1, F_2) \in \mathcal{F} \times \mathcal{F}$, an algorithm needs to output $T$ elements in the $q$ intersections $F_1 \cap F_2$ for a parameter $T$ (an element in multiple intersections is counted multiple times). 
\end{problem}

\begin{problem}[{\sf Colorful Boolean Matrix Multiplication} (\ColorBMM)]
Given an $n\times n$ Boolean matrix $A$ and an
$n\times n$ Boolean matrix $B$ and a mapping $\col:[n]\rightarrow \Gamma$, for each $i\in [n]$ and $j\in [n]$, decide whether $\{\col(k): A[i,k]\wedge B[k,j],\ k\in [n]\} = \Gamma$.
\end{problem}

\begin{problem}[{\sf AE-Colorful-Sparse-Triangle} (\AEColorSparseTri)]
Given a graph $G=(V,E)$ with $m$ edges
and a mapping $\col:V\rightarrow\Gamma$.  For each edge $uv$, we want to decide whether
$\{\col(w): uwv\mbox{ is a triangle}\} = \Gamma$.
\end{problem}

\section{Hardness of All-Edges Sparse Triangle}\label{sec:sparsetri}

In this section, we show conditional lower bounds of \AESparseTri\ 
based on the conjectured hardness of \APSP, \ExactTri, or \ThreeSUM.

\subsection{\APSP\ \texorpdfstring{$\rightarrow$}{rightarrow} \AESparseTri}\label{sec:apsp}

Recall that we proved Theorem~\ref{thm:apsp} in Section~\ref{sec:proof_in_intro}:
\thmAPSP*

More generally, we can obtain a near $mD$ conditional lower bound in terms of the degeneracy $D$, if $D\ll m^{1/3}$, which again matches known upper bounds (in fact, regardless of the value of $\omega$).

\begin{theorem}\label{thm:apsp:deg}
If \AESparseTri\ for graphs with $m$ edges and degeneracy $\OO(m^\alpha)$ could be solved in $\OO(m^{1+\alpha-\eps})$ time for some constant $\alpha\le 1/3$, then \APSP\ could be solved in $\OO(n^{3-2\eps/(1+\alpha)})$ time using
Las Vegas randomization.
\end{theorem}
\begin{proof}
First, observe that the graph $G_{k,P}$ in Lemma~\ref{lem:apsp2}'s proof can be modified to
have degeneracy $\OO(d)$: Whenever a left node $x$ has $\Delta_{x}> d$ neighbors
among the right nodes,
we split $x$ into $\up{\frac{\Delta_{x}}{d}}$ copies, where each copy is linked to up to $d$ neighbors among the right nodes.
Each copy is also linked to all of the original $O(d\log n)$ neighbors among the middle nodes.  The number of left nodes increases to $\sum_x 
\up{\frac{\Delta_{x}}{d}}$, and so the number of edges increases by
$\OO(d\sum_x 
\up{\frac{\Delta_{x}}{d}})=\OO(\sum_x \Delta_x + dn)=\OO(n^2/d + dn)$.  Now, each left node has $\OO(d)$ neighbors among the middle and right nodes.  Each right node has $\OO(d)$ neighbors among the middle nodes.  It follows that
the degeneracy of the modified graph is $\OO(d)$.

To prove the theorem, choose $d$ so that $d=(n^2/d)^\alpha$,
i.e., $d = n^{2\alpha/(1+\alpha)}$.  Since $\alpha\le 1/3$, we 
have $d\le \sqrt{n}$ and so $dn\le n^2/d$.  
If \AESparseTri\ with $m$ edges and degeneracy $D$ could be solved in $T(m,D)$ time, then
\APSP\ could be solved in time
$\OO(n\cdot T(n^2/d,\,d))= \OO(n^{3-2\eps/(1+\alpha)})$ if $T(m)=\OO(m^{1+\alpha-\eps})$.
\hspace*{\fill}
\end{proof}

\subsection{\ExactTri\ \texorpdfstring{$\rightarrow$}{rightarrow} \AESparseTri}\label{sec:exacttri}

Next, we adapt our proof to reduce from \ExactTri.  As a matrix problem,
\ExactTri\ (or \AEExactTri) reduces to the following (after negating the third matrix): given three $n\times n$ real matrices $A,B,C$, decide for each $i,j\in [n]$ whether there exists $k\in[n]$ with $C[i,j]=A[i,k]+B[k,j]$.  We will solve a slightly stronger problem of finding the predecessor and successor of $C[i,j]$ among
$\{A[i,k]+B[k,j]: k\in [n]\}$. 
Here, if $C[i,j]$ is in this set, we will deviate from convention and define the predecessor of $C[i,j]$ to be itself.
The following is a version of the problem for rectangular matrices:

\newcommand{\ExactTriOne}{\mbox{\sf Real-Exact-Tri-Variant$_1$}} %
\newcommand{\ExactTriTwo}{\mbox{\sf Real-Exact-Tri-Variant$_2$}} %
\newcommand{\ExactTriThree}{\mbox{\sf Real-Exact-Tri-Variant$_3$}} %

\begin{problem}[\ExactTriOne]\label{prob:exacttri1}
We are given an $n\times d$ real matrix $A$,
a $d\times n$ real matrix $B$, and an $n\times n$ matrix~$C$, where $d\le n$.

For each $i,j\in [n]$, we want to find 
the predecessor and the successor among $\{A[i,k]+B[k,j]: k\in [d]\}$.
\end{problem}

Again, we introduce an intermediate problem, which the original problem reduces to:

\begin{problem}[\ExactTriTwo]\label{prob:exacttri2}
We are given an $n\times d$ real matrix $A$ and 
a $d\times n$ real matrix $B$, where $d\le n$.
For each $i,j\in [n]$, we are also given two indices $k_{ij}^-,k_{ij}^+\in [d]$.

For each $i,j\in [n]$, we want to find an index $k'_{ij}\in [d]$ (if it exists) satisfying
\begin{equation}\label{eqn:exacttri}
A[i,k_{ij}^-]+B[k_{ij}^-,j] < A[i,k'_{ij}]+B[k'_{ij},j] < A[i,k_{ij}^+]+B[k_{ij}^+,j]. 
\end{equation}
\end{problem}

\begin{lemma}\label{lem:exacttri1}
\ExactTriOne\ reduces to $\OO(1)$ calls to an oracle for
\ExactTriTwo\ using Las Vegas randomization.
\end{lemma}
\begin{proof}
We solve \ExactTriOne\ as follows:
Take a random subset $R\subseteq [d]$ of size $d/2$.
For each $i,j\in [n]$, first compute 
indices $k_{ij}^{-(R)},k_{ij}^{+(R)}\in [d]$ such that
$A[i,k_{ij}^{-(R)}]+B[k_{ij}^{-(R)},j]$ and 
$A[i,k_{ij}^{+(R)}]+B[k_{ij}^{+(R)},j]$ are the predecessor and successor (respectively) of the value $C[i,j]$ among the elements in $\{A[i,k]+B[k,j]: k\in R\}$.
This computation can be done recursively.

Next, initialize $k_{ij}^-=k_{ij}^{-(R)}$ and
$k_{ij}^+=k_{ij}^{+(R)}$.  Invoke the oracle for \ExactTriTwo\ to find some $k_{ij}'$ satisfying (\ref{eqn:exacttri}) for each $i,j\in [n]$.   
If $k_{ij}'$ exists and
$A[i,k_{ij}']+B[k_{ij}',j] \le C[i,j]$, reset $k_{ij}^-=k_{ij}'$.  
If $k_{ij}'$ exists and
$A[i,k_{ij}']+B[k_{ij}',j] > C[i,j]$,
reset $k_{ij}^+=k_{ij}'$. 
Now, repeat.  When $k_{ij}'$ is nonexistent for all $i,j\in[n]$, we can stop, 
since we would know that $k_{ij}^-$ and $k_{ij}^+$ are the indices
defining the predecessor and successor of $C[i,j]$ in $\{A[i,k]+B[k,j]: k\in [d]\}$ for all $i,j\in[n]$.

To bound the number of iterations, observe that for each $i,j\in [n]$,
the number of indices $k'\in[d]$ satisfying 
$A[i,k_{ij}^{-(R)}]+B[k_{ij}^{-(R)},j] < A[i,k']+B[k',j] < A[i,k_{ij}^{+(R)}]+B[k_{ij}^{+(R)},j]$
is $O(\log n)$ w.h.p.\ (since in a set of $d$ values, there are at most $O(\log n)$ values between two
consecutive elements in a random subset of size $d/2$ w.h.p.).  Thus, $O(\log n)$ iterations suffice w.h.p.

As the recursion has $O(\log d)$ depth, the total number of
oracle calls is $O(\log d \log n)$ w.h.p.
\end{proof}

We now present our main reduction from \ExactTriTwo\ to \AESparseTri, which is similar to
our earlier reduction.  The details look a bit more involved, because of the need to work with pairs of indices $(k^-,k^+)$ instead of a single index $k$.  This also causes some loss of efficiency.
Nevertheless, it can still lead to a tight conditional lower bound for \AESparseTri\ in the case when the degeneracy is sufficiently small ($\ll m^{1/4}$).

\begin{lemma}\label{lem:exacttri2}
\ExactTriTwo\ reduces to 
$O(d^2)$  instances of \AESparseTri\ on $n$-node graphs with 
$\OO(n^2/d^2+dn)$ edges and degeneracy $\OO(d)$.
\end{lemma}
\begin{proof}
To solve Problem~\ref{prob:exacttri2}, first
sort the following list of $O(d^2n)$ elements in $O(d^2n\log n)$ time:
\[ L\ =\ \{A[i,k'] - A[i,k] : i\in [n],\,k,k'\in [d]\} \ \cup\ 
     \{B[k,j] - B[k',j] : j\in [n],\,k,k'\in [d]\}. \]

Fix $k^-,k^+\in[d]$.  Let
$P_{k^-,k^+}=\{(i,j)\in [n]^2: (k_{ij}^-,k_{ij}^+)=(k^-,k^+)\}$.
Divide $P_{k^-,k^+}$ into $\up{\frac{|P_{k^-,k^+}|}{n^2/d^2}}$ subsets
of size $O(n^2/d^2)$.
Fix one such subset $P\subseteq P_{k^-,k^+}$.
We solve \AESparseTri\ on the following tripartite graph $G_{k^-,k^+,P}$:
\begin{enumerate}
\item[0.] The left nodes are $\{x[i]: i\in [n]\}$, the middle nodes are $\{y[k,I^-,I^+]: k\in [d],\ \mbox{$I^-$ and $I^+$
are}$ $\mbox{dyadic intervals}\}$, and
the right nodes are $\{z[j]: j\in [n]\}$. 
\item For each $(i,j)\in P$, create an edge $x[i]\,z[j]$.
\item For each $i\in [n]$, $k'\in [d]$, and dyadic intervals $I^-$ and $I^+$,
create an edge $x[i]\, y[k',I^-,I^+]$ if 
the rank of $A[i,k^-]-A[i,k']$ in $L$ is in the left half of $I^-$ and
the rank of $A[i,k']-A[i,k^+]$ in $L$ is in the left half of $I^+$.
\item For each $j\in [n]$, $k'\in [d]$, and  dyadic intervals $I^-$ and $I^+$,
create an edge $y[k',I^-,I^+]\, z[j]$ if 
the rank of $B[k',j]-B[k^-,j]$ in $L$ is in the right half of $I^-$ and
the rank of $B[k^+,j]-B[k',j]$ in $L$ is in the right half of $I^+$.
\end{enumerate}

Step~1 creates $O(n^2/d^2)$ edges.  Steps 2--3 create $O(dn\log^2 n)$ edges, since any fixed value lies in $O(\log n)$ dyadic intervals.
Thus, the graph $G_{k^-,k^+,P}$ has $\OO(n^2/d^2 + dn)$ edges.
The total number of graphs is $\sum_{k=1}^{d^2}\up{\frac{|P_{k^-,k^+}|}{n^2/d^2}}=O(d^2)$.  

For any given $(i,j) \in [n]^2$, take $(k^-,k^+)=(k_{ij}^-,k_{ij}^+)$ and $P$ to be the subset of $P_{k^-,k^+}$ containing $(i,j)$.  Finding a $k'$ with $A[i,k_{ij}^-]+B[k_{ij}^-,j] < A[i,k']+B[k',j] <  A[i,k_{ij}^+]+B[k_{ij}^+,j]$ is equivalent to finding a $k'$ with
$A[i,k^-]-A[i,k'] < B[k',j]-B[k^-,j]$ and
$A[i,k']-A[i,k^+] < B[k^+,j]-B[k',j]$, which is equivalent to finding
a triangle
$x[i]\,y[k',I^-,I^+]\,z[j]$ in the graph $G_{k^-,k^+,P}$.  Thus, the answers to
\ExactTriTwo\ can be deduced from the answers
to the \AESparseTri\ instances.

Finally, to ensure that each graph $G_{k^-,k^+,P}$ has degeneracy $\OO(d)$, we modify the graph by splitting nodes in the same
way as in the proof of Theorem~\ref{thm:apsp:deg}.
The number of edges increases by
$\OO(d\sum_x 
\up{\frac{\Delta_{x}}{d}})=\OO(\sum_x \Delta_x + dn)=\OO(n^2/d^2 + dn)$.
\end{proof}

\begin{theorem}
\label{thm:exacttri}
If \AESparseTri\ could be solved in $\OO(m^{5/4-\eps})$ time, then \ExactTri\ (and in fact, \AEExactTri) could be solved in $\OO(n^{3-4\eps/3})$ time  using Las Vegas randomization.

More generally, if \AESparseTri\ with $m$ edges and degeneracy $\OO(m^{\alpha})$ could be solved in $\OO(m^{1+\alpha-\eps})$ time
for some constant $\alpha\le 1/4$, then \ExactTri\ (and in fact, \AEExactTri) could be solved in $\OO(n^{3-2\eps/(1+2\alpha)})$ time  using Las Vegas randomization.
\end{theorem}
\begin{proof}
By combining the above two lemmas, if \AESparseTri\ with $m$ edges and degeneracy $D$ could be solved in $T(m,D)$ time, then
Problem~\ref{prob:exacttri1} could be solved
 in $\OO(d^2\cdot T(n^2/d^2 + dn,\,d))$ 
time.  \ExactTri\ (in fact, \AEExactTri) reduces to $O(n/d)$ instances
of \ExactTriOne\ and could then be solved in 
$\OO(dn\cdot T(n^2/d^2 + dn,\,d))$ 
time.  
Choose $d$ so that $d=(n^2/d^2)^\alpha$, i.e., $d=n^{2\alpha/(1+2\alpha)}$.  Since $\alpha\le 1/4$, we have $d\le n^{1/3}$
and so $dn\le n^2/d^2$.  The time bound becomes
 $\OO(dn\cdot T(n^2/d^2,d))=\OO(n^{3-2\eps/(1+2\alpha)})$ if $T(m,m^\alpha)=\OO(m^{1+\alpha-\eps})$.
\end{proof}

\subsection{\ThreeSUM\ \texorpdfstring{$\rightarrow$}{rightarrow} \AESparseTri}\label{sec:3sum}

We next adapt our proof to reduce from \ThreeSUM.
We will more generally consider an asymmetric version of
\ThreeSUM\ with three sets $A$, $B$, and $C$
of sizes $n$, $n$, and $\nhat$ respectively with $\nhat\le n$ (the
standard version has $\nhat=n$).  Sort $A$ and $B$ and divide the sorted lists of $A$ and $B$ into sublists $A_1,\ldots,A_{n/d}$ and $B_1,\ldots,B_{n/d}$ of size $d$, for a given parameter $d\le n$.  Let $A[i,k]$ denote the $k$-th element of $A_i$ for each $i\in [n/d]$ and $k\in [d]$, and $B[j,\ell]$ denote the $\ell$-th element of $B_j$ for each $j\in [n/d]$ and $\ell\in [d]$.  We will consider a slightly stronger
problem: for each $c\in C$, find the predecessor and successor of $c$
among $A+B$.  
(If $c$ is in this set, we will deviate from convention and  define the predecessor of $c$ to be itself.)
By a known observation (which was used in Gr{\o}nlund and Pettie's \ThreeSUM\ algorithm~\cite{gronlund2014} and in subsequent algorithms~\cite{chan3sum}), for each $c\in C$, it suffices to search for $c$ in $A_i+B_j$ for $O(n/d)$ pairs $(i,j)$
(since these $(i,j)$ pairs form a ``staircase'' in the $[n/d]\times [n/d]$ grid).  Thus, \ThreeSUM\ (or \AllThreeSUM) reduces to the following problem (which was explicitly formulated, for example, in Chan's paper on \ThreeSUM~\cite{chan3sum}):

\newcommand{\ThreeSUMOne}{\mbox{\sf Real-3SUM-Variant$_1$}} %
\newcommand{\ThreeSUMTwo}{\mbox{\sf Real-3SUM-Variant$_2$}} %
\newcommand{\ThreeSUMThree}{\mbox{\sf Real-3SUM-Variant$_3$}} %

\begin{problem}[\ThreeSUMOne]\label{prob:3sum1}
We are given an $(n/d)\times d_A$ real matrix $A$ and an
$(n/d)\times d_B$ real matrix $B$ with $d_A,d_B\le d$.
For each $i,j\in [n/d]$, we are also given a set $C_{ij}$ of real
numbers with $\sum_{i,j}|C_{ij}|=O(\nhat n/d)$.

For each $i,j\in [n/d]$ and each $c\in C_{ij}$, we want to find the predecessor and successor 
of the value $c$ among the elements in
$\{A[i,k]+B[j,\ell]: k\in [d_A],\, \ell\in [d_B]\}$.
\end{problem}

We again introduce an intermediate problem, which the original problem reduces to via random sampling:

\begin{problem}[\ThreeSUMTwo]\label{prob:3sum2}
We are given an $(n/d)\times d$ real matrix $A$ and 
an $(n/d)\times d$ real matrix $B$.
For each $i,j\in [n/d]$, we are also given a set $Q_{ij}\subseteq [d]^4$ of quadruples with $\sum_{i,j}|Q_{ij}|=O(\nhat n/d)$.

For each $i,j\in [n/d]$ and each $q=(k^-,\ell^-,k^+,\ell^+)\in Q_{ij}$, we want to find indices $k'_{ijq},\ell'_{ijq}\in [d]$ (if they exist) satisfying
\begin{equation}\label{eqn:3sum}
A[i,k^-]+B[j,\ell^-] < A[i,k'_{ijq}]+B[j,\ell'_{ijq}] < A[i,k^+]+B[j,\ell^+]. 
\end{equation}
\end{problem}

\begin{lemma}\label{lem:3sum1}
\ThreeSUMOne\ reduces to $\OO(1)$ calls to an oracle for
\ThreeSUMTwo\ using Las Vegas randomization.
\end{lemma}
\begin{proof}
We solve \ThreeSUMOne\ as follows:
Take a random subset $R\subseteq [d_A]$ of size $d_A/2$.
For each $i,j\in [n/d]$ and each $c\in C_{ij}$, first compute 
indices $k_{ijc}^{-(R)},\ell_{ijc}^{-(R)},k_{ijc}^{+(R)},\ell_{ijc}^{+(R)}\in [d]$ such that
$A[i,k_{ijc}^{-(R)}]+B[j,\ell_{ijc}^{-(R)}]$ and 
$A[i,k_{ijc}^{+(R)}]+B[j,\ell_{ijc}^{+(R)}]$ are the predecessor and successor (respectively) of the value $c$ among the elements in $\{A[i,k]+B[j,\ell]: k\in R,\,\ell\in [d_B]\}$.
This computation can be done recursively.

Next, initialize $(k_{ijc}^-,\ell_{ijc}^-)=(k_{ijc}^{-(R)},\ell_{ijc}^{-(R)})$ and
$(k_{ijc}^+,\ell_{ijc}^+)=(k_{ijc}^{+(R)},\ell_{ijc}^{+(R)})$.  Invoke the oracle for \ThreeSUMTwo\ to find some 
$(k_{ijc}',\ell_{ijc}')\in [d_A]\times [d_B]$ satisfying (\ref{eqn:3sum}) for each $i,j\in [n]$ and $c\in C_{ij}$.   If $(k_{ijc}',\ell_{ijc}')$ exists and $A[i,k_{ijc}']+B[j,\ell_{ijc}']\le c$,
find the index $\ell\in [d_B]$
such that $A[i,k_{ijc}']+B[j,\ell]$ is the predecessor of the value $c$
among the elements in $\{A[i,k_{ijc}']+B[j,\ell]: \ell\in [d_B]\}$.
This index can be found in $O(\log d)$ time by binary search,
assuming that each row of $B$ has been sorted (which requires only
$O(n\log n)$ preprocessing time).
Reset $(k_{ijc}^-,\ell_{ijc}^-) =(k_{ijc}',\ell)$.  
If $(k_{ijc}',\ell_{ijc}')$ exists and $A[i,k_{ijc}']+B[j,\ell_{ijc}']> c$, we proceed similarly, replacing ``predecessor'' with ``successor'' and $-$ superscripts with $+$.
Now, repeat.  When $(k_{ijc}',\ell_{ijc}')$ is nonexistent for all $i,j\in[n]$, we can stop.

To bound the number of iterations, observe that for each $i,j\in [n]$
and each $c\in C_{ij}$,
the number of indices $k'\in[d_A]$ such that the predecessor
of $c$ among $\{A[i,k']+B[j,\ell]:\ell\in [d_B]\}$
is greater than $A[i,k_{ijc}^{-(R)}]+B[\ell,k_{ijc}^{-(R)}]$
is $O(\log n)$ w.h.p (since in a set of $d_A$ values, there are 
at most $O(\log n)$ values greater than the maximum of a random subset of size $d_A/2$).  A similar statement, replacing ``predecessor'' with ``successor'', $-$ superscripts with $+$, and ``greater'' with ``less'', holds.  Thus, $O(\log n)$ iterations suffice w.h.p.

As the recursion has $O(\log d_A)$ depth, the total number of
oracle calls is $O(\log d \log n)$ w.h.p., and the extra cost
of the binary searches is $O((\nhat n/d)\log^2 d\log n)$, which is negligible.
\end{proof}

We now present our main reduction from \ThreeSUMTwo\ to
\AESparseTri, which is inspired by
Gr{\o}nlund and Pettie's work~\cite{gronlund2014} and is also based
on Fredman's trick.  Although we now need to work with even more indices, the basic idea is similar to our earlier reductions.  The extra complications cause further loss of efficiency
(but we still obtain a tight conditional lower bound for \AESparseTri, albeit for a more restricted range of degeneracy $\ll m^{1/5}$).
The whole proof is still simple (comparable to the original reductions from \IntThreeSUM\ by P{\u{a}}tra{\c{s}}cu~\cite{patrascu2010towards} or Kopelowitz, Pettie and Porat~\cite{kopelowitz2016higher}, but completely bypassing hashing arguments).

\begin{lemma}\label{lem:3sum2}
\ThreeSUMTwo\ reduces to one instance of
\AESparseTri\ on a graph with $\OO(\nhat n/d + d^2n)$ edges
and degeneracy $\OO(d)$.
\end{lemma}
\begin{proof}
To solve \ThreeSUMTwo, first
sort the following list of $O(dn)$ elements in $O(dn\log n)$ time:
\[ L\ =\ \{A[i,k'] - A[i,k] : i\in [n/d], \,k,k'\in [d]\} \ \cup\ 
     \{B[j,\ell] - B[j,\ell'] : j\in [n/d], \,\ell,\ell'\in [d]\}. \]

We solve \AESparseTri\ on the following tripartite graph $G$:
\begin{enumerate}
\item[0.] The left nodes are $\{x[i,k^-,k^+]: i\in [n/d], \,k^-,k^+\in[d]\}$, the middle nodes are $\{y[I^-,I^+]: \mbox{$I^-$ and $I^+$ are dyadic intervals}\}$, and
the right nodes are $\{z[j,\ell^-,\ell^+]: j\in [n/d], \,\ell^-,\ell^+\in [d]\}$. 
\item For each $i,j\in [n/d]$ and $(k^-,\ell^-,k^+,\ell^+)\in Q_{ij}$, create an edge $x[i,k^-,k^+]\,z[j,\ell^-,\ell^+]$.
\item For each $i\in [n/d]$, $k^-,k^+,k'\in [d]$, and dyadic intervals $I^-$ and $I^+$,
create an edge $x[i,k^-,k^+]\, y[I^-,I^+]$ if 
the rank of $A[i,k^-]-A[i,k']$ in $L$ is in the left half of $I^-$ and
the rank of $A[i,k']-A[i,k^+]$ in $L$ is in the left half of $I^+$.
\item For each $j\in [n/d]$, $\ell^-,\ell^+,\ell'\in [d]$, and  dyadic intervals $I^-$ and $I^+$,
create an edge $y[I^-,I^+]\, z[j,\ell^-,\ell^+]$ if 
the rank of $B[j,\ell']-B[j,\ell^-]$ in $L$ is in the right half of $I^-$ and
the rank of $B[j,\ell^+]-B[j,\ell']$ in $L$ is in the right half of $I^+$.
\end{enumerate}

Step~1 creates $\sum_{i,j}|Q_{ij}|=O(\nhat n/d)$ edges.  Steps 2--3 create $O((n/d)d^3\log^2 n)$ edges, since any fixed value lies in $O(\log n)$ dyadic intervals.
Thus, the graph $G$ has $\OO(\nhat n/d + d^2n)$ edges.

For any given $(i,j) \in [n]^2$ and $(k^-,\ell^-,k^+,\ell^+)\in Q_{ij}$, finding a $(k',\ell')$ with
$A[i,k^-]+B[j,\ell^-] < A[i,k']+B[j,\ell'] < A[i,k^+]+B[j,\ell^+]$
is equivalent to finding a $(k',\ell')$ with
$A[i,k^-]-A[i,k'] < B[j,\ell']-B[j,\ell^-]$ and
$A[i,k']-A[i,k^+] < B[j,\ell^+]-B[j,\ell']$, which is equivalent to finding
a triangle
$x[i,k^-,k^+]\,y[I^-,I^+]\,z[j,\ell^-,\ell^+]$ in the graph $G$.  Thus, the answers to \ThreeSUMTwo\ can be deduced from the answers
to the \AESparseTri\ instances.

Finally, to ensure that the graph $G$ has degeneracy $\OO(d)$, we modify
the graph by splitting nodes in the same way as in the last paragraph of Theorem~\ref{thm:apsp:deg}'s proof. The number of edges increases by 
$\OO(d\sum_x 
\up{\frac{\Delta_{x}}{d}})=\OO(\sum_x \Delta_x + (n/d)d^2\cdot d)=\OO(\nhat n/d + d^2n)$.
\end{proof}

\begin{theorem}\label{thm:3sum}
If \AESparseTri\ with $m$ edges could be solved in $\OO(m^{6/5-\eps})$ time, then \ThreeSUM\ (and in fact, \AllThreeSUM) could be solved in $\OO(n^{2-5\eps/3})$ time  using Las Vegas randomization.

More generally, if \AESparseTri\ with $m$ edges and degeneracy $\OO(m^{\alpha})$ could be solved in $\OO(m^{1+\alpha-\eps})$ time
for some constant $\alpha\le \beta/(3+2\beta)$, then \ThreeSUM\ 
for three sets of sizes $n$, $n$, and $\nhat=n^\beta\ (\beta\le 1)$ could be solved in $\OO(n^{1+\beta-\eps(1+\beta)/(1+\alpha)})$ time  using Las Vegas randomization.
\end{theorem}
\begin{proof}
By combining the above two lemmas, if \AESparseTri\ could be solved in $T(m,D)$ time, then
\ThreeSUMOne\ could be solved
 in $\OO(T(\nhat n/d + d^2n,\,d))$ 
time.  
\ThreeSUM\ (more generally, \AllThreeSUM) reduces to an instance of \ThreeSUMOne, with $\nhat=n^\beta$.   (For the original symmetric version of \ThreeSUM, $\beta=1$.)
Choose $d$ so that $d=(n^{1+\beta}/d)^\alpha$, i.e., $d=n^{(1+\beta)\alpha/(1+\alpha)}$.  Since $\alpha\le \beta/(2+3\beta)$, we have $d\le n^{\beta/3}$
and so $d^2n\le n^{1+\beta}/d$.
The time bound becomes $\OO(T(n^{1+\beta}/d,d))=\OO(n^{1+\beta-\eps(1+\beta)/(1+\alpha)})$ if $T(m,m^\alpha)=\OO(m^{1+\alpha-\eps})$.  
\end{proof}

\section{Hardness of All-Edges Sparse Triangle Counting}\label{sec:sparsetri:count}

For the counting variant of \AESparseTri, we show that our reductions from \ExactTri\ and \ThreeSUM\ can be made more efficient, yielding near $m^{4/3}$ conditional lower bounds to match
the reduction from \APSP.  Note that the $O(m^{2\omega / (\omega + 1)})$ time algorithm by Alon, Yuster and Zwick~\cite{AlonYZ97} still works for \AESparseTriCount, so these lower bounds are tight assuming $\omega = 2$. 

\subsection{\ExactTri\ \texorpdfstring{$\rightarrow$}{rightarrow} \AESparseTriCount}\label{sec:exacttri:count}

To adapt the \ExactTri\ $\rightarrow$ \AESparseTri\ proof from Section~\ref{sec:exacttri}, we introduce a counting version of the
intermediate problem, which the original problem reduces to:

\begin{problem}[\ExactTriThree]\label{prob:exacttri3}
We are given an $n\times d$ real matrix $A$ and 
a $d\times n$ real matrix $B$, where $d\le n$.
For each $i,j\in [n]$, we are also given an index $k_{ij}\in [d]$.

For each $i,j\in [n]$, we want to count the number of indices $k'\in [d]$ satisfying
\begin{equation}\label{eqn:exacttri3}
A[i,k']+B[k',j] < A[i,k_{ij}]+B[k_{ij},j].
\end{equation}
(Note that we can solve the analogous problem with $>$ by negation of all matrix entries, and thus 
with $\le$ as well.)
\end{problem}

\begin{lemma}\label{lem:exacttri23}
\ExactTriTwo\ reduces to
$\OO(1)$ calls to an oracle for \ExactTriThree\ using Las Vegas randomization.
\end{lemma}
\begin{proof}
First consider a weaker \emph{decision} version of
\ExactTriTwo\ where
we just want to decide the existence of $k_{ij}'$
for each $i,j\in[n]$.  This can
be solved by calling the oracle for \ExactTriThree\
twice, first with $k_{ij}$ set to $k_{ij}^+$, then with $k_{ij}$
set to $k_{ij}^-$ (the latter with $\le$ instead of $<$).
For each $i,j$, the answer is yes iff the two counts are different.

Next, we reduce \ExactTriTwo\ to the decision version,
by standard random sampling techniques for witness finding~\cite{AlonGMN92, seidel1995}, which we briefly sketch:
If we know that $k_{ij}'$ is unique, for each $t\in [\log d]$, we can determine the $t$-th bit
of $k_{ij}'$ by solving the decision problem 
after restricting $k_{ij}'$ to lie in the subset
$\{k\in [d]: \mbox{the $t$-th bit of $k$ is 1}\}$.
To ensure uniqueness, we restrict $k_{ij}'$ to lie in 
$O(\log n)$ random subsets of $[d]$ of size $2^s$ for each $s\in [\log d]$; then $k_{ij}'$ would be unique w.h.p.\ in one of the subsets.  

The overall number of calls to
the oracle for \ExactTriThree\ is polylogarithmic.
(Note that the above requires only Las Vegas randomization, since
fo each $i,j$, we know definitively whether $k_{ij}'$ exists, 
and if so, we can repeat until success, since we can easily verify
whether a given $k_{ij}'$ works.)
\end{proof}

The advantage of \ExactTriThree\ is that we are back to working with single indices instead of pairs of indices, so we recover
the same result as the reduction from \APSP:

\begin{lemma}
\ExactTriThree\ reduces to $O(d)$  instances
of  \AESparseTriCount\ on graphs with $\OO(n^2/d+dn)$ edges
and degeneracy $\OO(d)$.
\end{lemma}
\begin{proof}
We can reuse the same construction of the graphs $G_{k,P}$ from the proof of Lemma~\ref{lem:apsp2}.  (The degeneracy bound follows as in
the proof of Theorem~\ref{thm:apsp:deg}.)

For any given $(i,j) \in [n]^2$, take $k=k_{ij}$ and $P$ to be the subset of $P_k$ containing $(i,j)$.
Counting the number of indices $k'$ with $A[i,k']+B[k',j] <  A[i,k_{ij}]+B[k_{ij},j]$ is equivalent to counting the
number of indices $k'$ with
$A[i,k']-A[i,k] < B[k,j]-B[k',j]$, which is equivalent to 
counting the number of triangles of the form
$x[i]\,y[k',I]\,z[j]$ in the graph $G_{k,P}$.
\end{proof}

\begin{theorem}\label{thm:exacttri:count}
If \AESparseTriCount\ could be solved in $\OO(m^{4/3-\eps})$ time, then \ExactTri\ (and in fact, \AEExactTriCount\ or \AENegTriCount) could be solved in $\OO(n^{3-3\eps/2})$ time using Las Vegas randomization.

More generally, if \AESparseTriCount\ for graphs with $m$ edges and degeneracy $\OO(m^\alpha)$ could be solved in $\OO(m^{1+\alpha-\eps})$ time for some constant $\alpha\le 1/3$, then \ExactTri\ (and in fact, \AEExactTriCount\ or \AENegTriCount) could be solved in $\OO(n^{3-2\eps/(1+\alpha)})$ time using
Las Vegas randomization.
\end{theorem}
\begin{proof}
By combining Lemma~\ref{lem:exacttri1} and the above two lemmas, if \AESparseTri\ with $m$ edges and degeneracy $D$ could be solved in $T(m,D)$ time, then \ExactTriOne\ could be solved
 in $\OO(d\cdot T(n^2/d + dn,\,d))$ 
time.  \ExactTri\ reduces to $O(n/d)$ instances
of \ExactTriOne\ (note that \AEExactTriCount\ and \AENegTriCount\ also reduce 
to $O(n/d)$ instances of \ExactTriOne\ and \ExactTriThree), and could then be solved in 
$\OO(n\cdot T(n^2/d + dn,\, d))$ 
time.  
Choose $d$ so that $d=(n^2/d)^\alpha$,
i.e., $d = n^{2\alpha/(1+\alpha)}$.  Since $\alpha\le 1/3$, we 
have $d\le \sqrt{n}$ and so $dn\le n^2/d$.  
The time bound becomes
$\OO(n\cdot T(n^2/d,\,d))= \OO(n^{3-2\eps/(1+\alpha)})$ if $T(m,m^\alpha)=\OO(m^{1+\alpha-\eps})$.
\end{proof}

\subsection{\ThreeSUM\ \texorpdfstring{$\rightarrow$}{rightarrow} \AESparseTriCount}
\label{sec:3sum:count}

We similarly adapt the proof for \ThreeSUM\ from Section~\ref{sec:3sum}, by introducing a counting version of the
intermediate problem:

\begin{problem}[\ThreeSUMThree]\label{prob:3sum3}
We are given an $(n/d)\times d$ real matrix $A$ and 
an $(n/d)\times d$ real matrix $B$.
For each $i,j\in [n/d]$, we are also given a set $Q_{ij}\subseteq [d]^2$ of pairs with $\sum_{i,j}|Q_{ij}|=O(\nhat n/d)$.

For each $i,j\in [n/d]$ and each $q=(k,\ell)\in Q_{ij}$, we want to
count the number of $(k',\ell')\in [d]^2$ satisfying
\begin{equation}\label{eqn:3sum3}
A[i,k']+B[j,\ell'] < A[i,k]+B[j,\ell].
\end{equation}
\end{problem}

\begin{lemma}\label{lem:3sum23}
\ThreeSUMTwo\ reduces to $\OO(1)$ calls to an oracle for \ThreeSUMThree\ using Las Vegas randomization.
\end{lemma}
\begin{proof}
First consider a weaker \emph{decision} version of \ThreeSUMTwo\ where
we just want to decide the existence of $(k_{ijq}',\ell_{ijq}')$
for each $i,j\in[n]$ and each $q\in Q_{ij}$.  This can
be solved by calling the oracle for \ThreeSUMThree\
twice, first with $Q_{ij}$ replaced by $\{(k^+,\ell^+) : (k^-,\ell^-,k^+,\ell^+)\in Q_{ij}\}$, then with $Q_{ij}$ replaced by $\{(k^-,\ell^-) : (k^-,\ell^-,k^+,\ell^+)\in Q_{ij}\}$ (the latter with $\le$ instead of $<$).
For each $i,j$ and $q\in Q_{ij}$, the answer is yes iff the corresponding two counts are different.

Next, we reduce \ThreeSUMTwo\ to the decision version,
by standard techniques for witness finding~\cite{AlonGMN92, seidel1995}, as we have
already sketched in Lemma~\ref{lem:exacttri23}'s proof: 
If we know that $k_{ijq}'$ is unique, for each $t\in [\log d]$, we can determine the $t$-th bit
of $k_{ijq}'$ by solving the decision problem 
after restricting $k_{ijq}'$ to lie in the subset
$\{k\in [d]: \mbox{the $t$-th bit of $k$ is 1}\}$.
Knowing $k_{ijq}'$, we can find a corresponding $\ell_{ijq}'$
in $O(\log d)$ time by binary search, assuming that each row of $B$
is sorted (which requires only $O(n\log n)$ preprocessing time).
To ensure uniqueness, we restrict $k_{ijq}'$ to lie in 
$O(\log n)$ random subsets of $[d]$ of size $2^s$ for each $s\in [\log d]$; then $k_{ij}'$ would be unique w.h.p.\ in one of the subsets.  

The overall number of calls to
the oracle for \ThreeSUMThree\ is polylogarithmic.
\end{proof}

The graph construction from Lemma~\ref{lem:3sum2} can now be simplified, since we are working with fewer indices:

\begin{lemma}\label{lem:3sum3}
\ThreeSUMThree\ reduces to one instance
of \AESparseTriCount\ on a graph with $\OO(\nhat n/d + dn)$ edges
and degeneracy $\OO(d)$.
\end{lemma}
\begin{proof}
To solve \ThreeSUMThree, first
sort the following list of $O(dn)$ elements in $O(dn\log n)$ time:
\[ L\ =\ \{A[i,k'] - A[i,k] : i\in [n/d], \,k,k'\in [d]\} \ \cup\ 
     \{B[j,\ell] - B[j,\ell'] : j\in [n/d], \,\ell,\ell'\in [d]\}. \]

We solve \AESparseTriCount\ on the following tripartite graph $G$:
\begin{enumerate}
\item[0.] The left nodes are $\{x[i,k]: i\in [n/d], \,k\in[d]\}$, the middle nodes are $\{y[I]: \mbox{$I$ is a dyadic}$ $\mbox{interval}\}$, and
the right nodes are $\{z[j,\ell]: j\in [n/d], \,\ell\in [d]\}$. 
\item For each $i,j\in [n/d]$ and $(k,\ell)\in Q_{ij}$, create an edge $x[i,k]\,z[j,\ell]$.
\item For each $i\in [n/d]$, $k,k'\in [d]$, and each dyadic interval $I$,
create an edge $x[i,k]\, y[I]$ if 
the rank of $A[i,k']-A[i,k]$ in $L$ is in the left half of $I$.
\item For each $j\in [n/d]$, $\ell,\ell'\in [d]$, and each dyadic interval $I$,
create an edge $y[I]\, z[j,\ell]$ if 
the rank of $B[j,\ell]-B[j,\ell']$ in $L$ is in the right half of $I$.
\end{enumerate}

Step~1 creates $\sum_{i,j}|Q_{ij}|=O(\nhat n/d)$ edges.  Steps 2--3 create $O((n/d)d^2\log n)$ edges, since any fixed value lies in $O(\log n)$ dyadic intervals.
Thus, the graph $G$ has $\OO(\nhat n/d + dn)$ edges.

For any given $(i,j) \in [n]^2$ and $(k,\ell)\in Q_{ij}$, counting the number of $(k',\ell')$ with
$A[i,k']+B[j,\ell'] < A[i,k]+B[j,\ell]$
is equivalent to counting the number of $(k',\ell')$ with
$A[i,k']-A[i,k] < B[j,\ell]-B[j,\ell']$, which is equivalent to counting
the number of triangles of the form
$x[i,k]\,y[I]\,z[j,\ell]$ in the graph $G$.  Thus, the answers to \ThreeSUMThree\ can be deduced from the answers
to the \AESparseTriCount\ instances.

The degeneracy bound follows as before.
\end{proof}

\begin{theorem}\label{thm:3sum:count}
If \AESparseTriCount\ could be solved in $\OO(m^{4/3-\eps})$ time, then \ThreeSUM\ (and in fact, \AllThreeSUM) could be solved in $\OO(n^{2-3\eps/2})$ time using Las Vegas randomization.

More generally, if \AESparseTriCount\ with $m$ edges and degeneracy $\OO(m^{\alpha})$ could be solved in $\OO(m^{1+\alpha-\eps})$ time
for some constant $\alpha\le 1/3$, then \ThreeSUM\ (and in fact, \AllThreeSUM) could be solved in $\OO(n^{2-2\eps/(1+\alpha)})$ time  using Las Vegas randomization.
\end{theorem}
\begin{proof}
By combining Lemma~\ref{lem:3sum1} and the above two lemmas, if \AESparseTriCount\ with $m$ edges and degeneracy $D$ could be solved in $T(m,D)$ time, then
\ThreeSUMOne\ could be solved
 in $\OO(T(\nhat n/d + dn,\,d))$ 
time.  \ThreeSUM\ (in fact, \AllThreeSUM) reduces to an instance of \ThreeSUMOne\ with $\nhat =n$.  
Choose $d$ so that $d=(n^2/d)^\alpha$,
i.e., $d = n^{2\alpha/(1+\alpha)}$.  Since $\alpha\le 1/3$, we 
have $d\le \sqrt{n}$ and so $dn\le n^2/d$.  
The time bound becomes
$\OO(T(n^2/d,\,d))= \OO(n^{2-2\eps/(1+\alpha)})$ if $T(m,m^\alpha)=\OO(m^{1+\alpha-\eps})$.
\end{proof}

\section{Hardness of
All-Edges Monochromatic Triangle}
\label{sec:mono}

All the reductions to \AESparseTri\ (or \AESparseTriCount) in Sections \ref{sec:sparsetri}--\ref{sec:sparsetri:count}  also imply
reductions to \AEMonoTri\ (or \AEMonoTriCount).  This is because
of a result by Lincoln, Polak and Vassilevska W.~\cite{lincoln2020monochromatic} which reduces multiple instances of \AESparseTri\ to \AEMonoTri\ (the idea is to simply
overlay the multiple input graphs into one edge-colored graph,
after randomly permuting the nodes of each input graph):

\begin{lemma}\label{lem:monotri}
Any $n^2/m$  instances of \AESparseTri\ (resp.\ \AESparseTriCount) on graphs with $n$ nodes and $m$ edges reduce to $\OO(1)$
instances of \AEMonoTri\ (resp.\ \AEMonoTriCount) on graphs with $n$ nodes (and $\OO(n^2/m)$ colors) using Las Vegas randomization.
\end{lemma}

\subsection{\APSP\ \texorpdfstring{$\rightarrow$}{rightarrow} \AEMonoTri}

For example, we can combine our \APSP\ $\rightarrow$ \AESparseTri\ reduction (Section~\ref{sec:apsp}) with Lemma~\ref{lem:monotri} to obtain the following theorem.  The near $n^{5/2}$ conditional lower bound below for \AEMonoTri\ matches known upper bounds~\cite{VassilevskaWY06} if $\omega=2$;
the second bound in terms of the number of colors is also tight if $\omega=2$.

\begin{theorem}\label{thm:apsp:monotri}
If \AEMonoTri\ could be solved in $\OO(n^{5/2-3\eps/2})$ time, then \APSP\ could be solved in $\OO(n^{3-\eps})$ time using
Las Vegas randomization.

More generally, if \AEMonoTri\ with $\OO(n^\alpha)$ colors
could be solved in $\OO(n^{2+\alpha-\eps})$ time for some
constant $\alpha\le (1-\eps)/2$, then \APSP\ could be solved in $\OO(n^{3-\eps})$ time using
Las Vegas randomization.
\end{theorem}
\begin{proof}
The graph $G_{k,P}$ from Lemma~\ref{lem:apsp2}'s proof may have a large
number of middle nodes.  We first observe that all $O(d)$ graphs can be modified to have $\OO(n+d^2n^\eps)$ nodes, after spending $O(dn^{2-\eps})$ time, by using a ``high vs.\ low degree'' trick:
For each middle node $y$ of degree at most $n^{1-\eps}/d$,
we enumerate all its neighbors $x$ among the left nodes
and all its neighbors $z$ among the right nodes, and for each
such pair $xz$ that is an edge, we record that its answer for \AESparseTri\ is yes.
We can then remove all such $y$'s. All this takes $O((dn\log n)\cdot 
(n^{1-\eps}/d))=\OO(n^{2-\eps})$ time per graph $G_{k,P}$.
The remaining middle nodes
have degree at least $n^{1-\eps}/d$, and so the number of middle nodes
is $O(\frac{dn\log n}{n^{1-\eps}/d})=\OO(d^2n^\eps)$.

Now, as in Theorem~\ref{thm:apsp}'s proof, \APSP\ reduces
to $\OO(1)$ rounds of $O(n)$  instances of \AESparseTri\ on graphs
with $\OO(n^2/d+dn)$ edges and $\OO(n+d^2n^\eps)$ nodes, where the reduction runs in $O((n/d)\cdot dn^{2-\eps})=O(n^{3-\eps})$ time.
Choose $d=n^\alpha$.  Since $\alpha\le (1-\eps)/2$, we have $dn\le n^2/d$
and $d^2n^\eps\le n$, so these graphs have $\OO(n^2/d)$ edges
and $\OO(n)$ nodes.
By Lemma~\ref{lem:monotri}, these $\OO(n)$ \AESparseTri\ instances
reduce to $\OO(n/d)=\OO(n^{1-\alpha})$ \AEMonoTri\ instances on graphs with $\OO(n)$ nodes and $\OO(d) = \OO(n^\alpha)$ colors. Thus, \AEMonoTriCount\ has an $\OO(n^{2+\alpha - \eps})$ time algorithm with $\OO(n^\alpha)$ colors, then \APSP\ has an $\OO(n^{1-\alpha} \cdot n^{2+\alpha - \eps}) = \OO(n^{3-\eps})$ time algorithm.
\end{proof}

\subsection{\ExactTri\ \texorpdfstring{$\rightarrow$}{rightarrow} \AEMonoTri\ (and \AEMonoTriCount)}

We can similarly combine our \ExactTri\ $\rightarrow$ \AESparseTriCount\ reduction (Section~\ref{sec:exacttri:count}) with Lemma~\ref{lem:monotri} (the proof is basically the same):

\begin{theorem}\label{thm:exacttri:monotricount}
If \AEMonoTriCount\ could be solved in $\OO(n^{5/2-3\eps/2})$ time, then \ExactTri\ (and in fact, \AEExactTriCount) could be solved in $\OO(n^{3-\eps})$ time using
Las Vegas randomization.

More generally, if \AEMonoTriCount\ with $\OO(n^\alpha)$ colors
could be solved in $\OO(n^{2+\alpha-\eps})$ time for some
constant $\alpha\le (1-\eps)/2$, then \ExactTri\ (and in fact, \AEExactTriCount) could be solved in $\OO(n^{3-\eps})$ time using
Las Vegas randomization.
\end{theorem}

We can also adapt the \ExactTri\ $\rightarrow$ \AESparseTri\ reduction:

\begin{theorem}\label{thm:exacttri:monotrinoncount}
If \AEMonoTri\ could be solved in $\OO(n^{7/3-4\eps/3})$ time, then \ExactTri\ could be solved in $\OO(n^{3-\eps})$ time using
Las Vegas randomization.

More generally, if \AEMonoTri\ with $\OO(n^{2\alpha})$ colors
could be solved in $\OO(n^{2+\alpha-\eps})$ time for some
constant $\alpha\le (1-\eps)/3$, then \ExactTri\ could be solved in $\OO(n^{3-\eps})$ time using
Las Vegas randomization.
\end{theorem}
\begin{proof}
The proof is similar. We first observe that the $O(d^2)$ graphs from Lemma~\ref{lem:exacttri2} can be modified to have $\OO(n + d^3 n^\eps)$ nodes after spending $O(dn^{2-\eps})$ time via the ``high vs.\ low degree'' trick, by handling the middle nodes with degree at most $n^{1-\eps}/d^2$ by brute-force. 

Now \ExactTri\ reduces to $\OO(nd)$ instances of \AESparseTri\ on graphs with $\OO(n^2/d^2 + dn)$ edges and $\OO(n + d^3 n^\eps)$ nodes, where the reduction runs in $O(n^{3-\eps})$ time. Choose $d = n^\alpha$. Since $\alpha \le (1-\eps) / 3$, we have $dn \le n^2 / d^2$ and $d^3 n^\eps \le n$, so these graphs have $\OO(n^2/d^2)$ edges and $\OO(n)$ nodes. By Lemma~\ref{lem:monotri}, these $\OO(nd)$ \AESparseTri\ instances
reduce to $\OO(n/d)=\OO(n^{1-\alpha})$ \AEMonoTri\ instances on graphs with $\OO(n)$ nodes and $\OO(n^{2\alpha})$ colors. If there exists an $\OO(n^{2+\alpha - \eps})$ time algorithm for \AEMonoTri\ with $\OO(n^{2\alpha})$ colors, then there exists an $\OO(n/d \cdot n^{2 + \alpha - \eps}) = \OO(n^{3-\eps})$ time algorithm for \AEMonoTri.
\end{proof}

\subsection{\ThreeSUM\ \texorpdfstring{$\rightarrow$}{rightarrow} \AEMonoTri\ (and \AEMonoTriCount)}

We can also combine our \ThreeSUM\ $\rightarrow$ \AESparseTriCount\ reduction (Section~\ref{sec:3sum:count}) with Lemma~\ref{lem:monotri}:

\begin{theorem}
\label{thm:3sum:monotri}
If \AEMonoTriCount\ could be solved in $\OO(n^{5/2-15\eps / 8})$ time, then \ThreeSUM\ could be solved in $\OO(n^{2-\eps})$ time using
Las Vegas randomization. 

More generally, if \AEMonoTriCount\ with $\OO(n^{\alpha})$ colors
could be solved in $\OO(n^{2+\alpha-\frac{2+\alpha}{2} \eps})$ time for some
constant $\alpha\le \frac{1-\eps}{2+\eps/2}$, then \ThreeSUM\ could be solved in $\OO(n^{2-\eps})$ time using
Las Vegas randomization.

\end{theorem}
\begin{proof}
We first observe that the graph $G$ from Lemma~\ref{lem:3sum3} (ignoring the node-splitting step to lower degeneracy)
can be modified to have $\OO(n+d^2n^\delta)$ nodes,
after spending $O(n^{2-\delta})$ time for any $\delta > 0$.  This follows from
the same ``high vs.\ low degree'' trick from the proof of Theorem~\ref{thm:apsp:monotri}. We set $\delta = \frac{2 + \alpha}{2} \eps$ in the rest of the proof. This way, $\alpha \le (1-\delta) / 2$ since $\alpha\le \frac{1-\eps}{2+\eps/2}$.

\ThreeSUM\ (with $\nhat=n$) thus reduces to $\OO(1)$ instances of
\AESparseTriCount\ on graphs with $\OO(n^2/d+dn)$ edges and $\OO(n+d^2n^\delta)$
nodes, plus $\OO(n^{2-\delta})$ work.  By choosing $d=n^{\alpha} \le n^{(1-\delta)/2}$, these graphs have
$\OO(n^{2-\alpha})$ edges and $\OO(n)$ nodes.

However, in order to apply Lemma~\ref{lem:monotri}, we need a sufficient number of independent
instances.  We use the well known fact (see e.g. \cite{baran2005subquadratic, kopelowitz2016higher, patrascu2010towards, lincoln2016deterministic}) that a \ThreeSUM\ instance of size $n$
reduces to $O((n/r)^2)$ independent \ThreeSUM\ instances each
of size $r$.  This way, we obtain $\OO(1)$ rounds of $O((n/r)^2)$
independent \AESparseTriCount\ instances each with $\OO(r^{2-\alpha})$ edges
and $\OO(r)$ nodes after $\OO((n/r)^2 \cdot r^{2-\delta}) = \OO(n^2/r^\delta)$ work. 
By choosing $r=n^{2/(2+\alpha)}$, the amount of extra work becomes $\OO(n^2/r^\delta) = \OO(n^{2-\eps})$, and each round has
$O(n^{2\alpha/(2+\alpha)})$ 
independent \AESparseTriCount\ instances each with $\OO(n^{(4-2\alpha)/(2+\alpha)})$ edges
and $\OO(n^{2/(2+\alpha)})$ nodes. 
By Lemma~\ref{lem:monotri}, this reduces to $\OO(1)$ instances
of \AEMonoTriCount\ on graphs with $\OO(n^{2/(2+\alpha)})$  nodes and $O(n^{2\alpha/(2+\alpha)})$  colors.  If
\AEMonoTriCount\ with $N^\alpha$ colors could be solved in $T(N)=\OO(N^{2+\alpha-\frac{2+\alpha}{2}\eps})$ time, the time
bound for \ThreeSUM\ is $\OO(T(n^{2/(2+\alpha)}))=\OO(n^{2-\eps})$. 

We obtain the first claim in the theorem statement by setting $\alpha = \frac{1-\eps}{2+\eps/2}$, and use $2+\alpha - \frac{2+\alpha}{2} \eps \ge 5/2-15\eps / 8$ for $\eps > 0$. 
\end{proof}

\begin{theorem}
\label{thm:3sum:monotrinoncount}
If \AEMonoTri\ could be solved in $\OO(n^{9/4-45\eps/32})$ time, then \ThreeSUM\ could be solved in $\OO(n^{2-\eps})$ time using
Las Vegas randomization. 

More generally, if \AEMonoTri\ with $\OO(n^{3\alpha/(1+\alpha)})$ colors
could be solved in $\OO(n^{2+\frac{\alpha}{1+\alpha}-\frac{2+3 \alpha}{2+2\alpha} \eps})$ time for some
constant $\alpha\le \frac{1-\eps}{3+3\eps/2}$, then \ThreeSUM\ could be solved in $\OO(n^{2-\eps})$ time using
Las Vegas randomization.

\end{theorem}
\begin{proof}
We first observe that the graph $G$ from Lemma~\ref{lem:3sum2} (ignoring the node-splitting step to lower degeneracy)
can be modified to have $\OO(nd+d^4n^\delta)$ nodes,
after spending $O(n^{2-\delta})$ time for any $\delta > 0$.  This follows from
the same ``high vs.\ low degree'' trick. In the rest of the proof, we will set $\delta = \frac{2+3\alpha}{2} \eps$. It is straightforwards to verify $\alpha \le (1-\delta) /3$  when $\alpha\le \frac{1-\eps}{3+3\eps/2}$.

Thus, \ThreeSUM\ (with $\nhat=n$) thus reduces to $\OO(1)$ instances of
\AESparseTri\ on graphs with $\OO(n^2/d+d^2n)$ edges and $\OO(nd+d^4n^\delta)$
nodes, plus $\OO(n^{2-\delta})$ work.  By choosing $d=n^{\alpha} \le n^{(1-\delta)/3}$, these graphs have
$\OO(n^{2-\alpha})$ edges and $\OO(n^{1+\alpha})$ nodes.

We again use the well known fact that a \ThreeSUM\ instance of size $n$
reduces to $O((n/r)^2)$ independent \ThreeSUM\ instances each
of size $r$.  This way, we obtain $\OO(1)$ rounds of $O((n/r)^2)$
independent \AESparseTri\ instances each with $\OO(r^{2-\alpha})$ edges
and $\OO(r^{1+\alpha})$ nodes after $\OO((n/r)^2 \cdot r^{2-\delta}) = \OO(n^2/r^\delta)$ work. We choose $r=n^{2/(2+3\alpha)}$, so the amount of extra work becomes $\OO(n^{2-\eps})$, and each round has
$O(n^{6\alpha/(2+3\alpha)})$ 
independent \AESparseTri\ instances each with $\OO(n^{(4-2\alpha)/(2+3\alpha)})$ edges
and $\OO(n^{(2+2\alpha)/(2+3\alpha)})$ nodes.
By Lemma~\ref{lem:monotri}, this reduces to $\OO(1)$ instances
of \AEMonoTri\ on graphs with $\OO(n^{(2+2\alpha)/(2+3\alpha)})$  nodes and $O(n^{6\alpha/(2+3\alpha)})$   colors.  If
\AEMonoTri\ with $N$ nodes and $N^{3\alpha/(1+\alpha)}$ colors could be solved in $T(N)=\OO(N^{2+\frac{\alpha}{1+\alpha}-\frac{2+3 \alpha}{2+2\alpha} \eps})$ time, the running time bound for \ThreeSUM\ is $\OO(T(n^{(2+2\alpha)/(2+3\alpha)}))=\OO(n^{2-\eps})$. 

To get the first claim in the theorem, we set $\alpha = \frac{1-\eps}{3+3\eps/2}$ and use $2+\frac{\alpha}{1+\alpha}-\frac{2+3 \alpha}{2+2\alpha} \eps \ge 9/4 - 45\eps / 32$ for $\eps > 0$. 
\end{proof}

\section{Hardness of Colorful BMM}
\label{sec:colorBMM}

In this section, we show that  \APSP\ and \OV\ can be reduced to \ColorBMM.

\subsection{\APSP\ \texorpdfstring{$\rightarrow$}{rightarrow} \ColorBMM}

\newcommand{\AAA}{{\cal A}}
\newcommand{\BBB}{{\cal B}}

We first present our reduction from \APSP\ to \ColorBMM, which is simple and is inspired by
Williams's \APSP\ algorithm~\cite{Williams18}
(and its derandomization by Chan and Williams~\cite{ChanW21}) using ANDs of ORs.

\begin{lemma}\label{lem:distincteq}
\MinPlus\ of an $n\times d$ real matrix $A$
and a $d\times n$ real matrix $B$ reduces to
one instance of \ColorBMM\ for an $n\times\OO(d^4n^\eps)$ and
an $\OO(d^4n^\eps)\times n$ Boolean matrix
with $O(d)$ colors (and $\OO(d^2n)$ nonzero input entries), after spending $\OO(dn^{2-\eps})$ time.
\end{lemma}
\begin{proof}

For simplicity, we assume that $A[i,k]+B[k,j]\neq A[i,k']+B[k',j]$
for all $i,j\in [n]$ and $k,k'\in [d]$.  This can be ensured, for example, by
adding $k\delta$ to $A[i,k]$ for an infinitesimally small $\delta>0$. Note that we don't need to explicitly store $A[i,k]+k\delta$. Instead, we can store a number $a+b\delta$ as a pair of numbers $(a, b)$. Every time we need to compare two numbers $(a_1, b_1)$ and $(a_2, b_2)$, we first compare $a_1$ and $a_2$ and only compare $b_1$ and $b_2$ if $a_1 = a_2$. This implementation only incurs a constant factor overhead. 

Fix $t\in[\log d]$.
We will describe how to
compute the $t$-th bit of $k_{ij}=\arg\min_{k\in [d]}
(A[i,k]+B[k,j])$.
Let $K_t=\{k\in [d]: \mbox{the $t$-th bit of $k$ is 1}\}$.
Note that the $t$-th bit of $k_{ij}$ is 1 iff
\[ %
 \bigwedge_{k\in [d]-K_t}\bigvee_{k'\in K_t} \big[A[i,k']+B[k',j]<A[i,k]+B[k,j]\big]\ =\ 
  \bigwedge_{k\in [d]-K_t}\bigvee_{k'\in K_t} \big[A[i,k']-A[i,k]<B[k,j]-B[k',j]\big].
\]

Thus, the answers can be determined by solving \ColorBMM\ on the following matrices $\AAA$ and~$\BBB$ and color mapping to $[d]-K_t$:
\begin{enumerate}
\item For each $i\in [n]$, $k\in [d]-K_t$, $k'\in K_t$,
and each dyadic interval $I$,
let $\AAA[i,(k,k',I)]=1$ iff the rank of $A[i,k']-A[i,k]$ lies in
the left half of the dyadic interval $I$.
\item For each $j\in [n]$, $k\in [d]-K_t$, $k'\in K_t$,
and each dyadic interval $I$,
let $\BBB[(k,k',I),j]=1$ iff the rank of $B[k,j]-B[k',j]$ lies in
the right half of the dyadic interval $I$.
\item Define $\col((k,k',I))=k$.
\end{enumerate}

However, the inner dimension (i.e., the number of columns of $\AAA$ or rows of $\BBB$)
is large.  We lower the inner dimension by using the ``high vs.\ low degree'' trick:  Label a triple $(k,k',I)$ ``high'' if the number of 
indices $i$ for which
$\AAA[i,(k,k',I)]=1$ exceeds $n^{1-\eps}/d^2$; otherwise, label it ``low''.  For each low triple $(k,k',I)$,
we enumerate all $(i,j)$ for which $\AAA[i,(k,k',I)]=1$ and
$\BBB[(k,k',I),j]=1$, and mark these pairs $(i,j)$
as ``bad''.  There are $O((d^2n\log n)\cdot n^{1-\eps}/d^2)=\OO(n^{2-\eps})$
bad pairs, and for each bad pair $(i,j)$, we can compute
its corresponding entry of the $(\min,+)$-product in $O(d)$ time by brute force.  This takes $\OO(dn^{2-\eps})$ time.  For the remaining good pairs $(i,j)$, it suffices to
keep only the high triples, and the number of high triples is $O(\frac{d^2n\log n}{n^{1-\eps}/d^2})=\OO(d^4n^\eps)$.
So, in the remaining \ColorBMM\ instance, $\AAA$ and $\BBB$ have dimensions $\OO(n)\times \OO(d^4n^\eps)$ and $\OO(d^4n^\eps)\times \OO(n)$.  
\end{proof}

\begin{theorem}\label{thm:distincteq}
If \ColorBMM\ for two $n\times n$ Boolean matrices could be solved in $\OO(n^{9/4-5\eps/4})$ time, then
\APSP\ could be solved in $\OO(n^{3-\eps})$ time.

More generally, if \ColorBMM\ for two $n\times n$ Boolean matrices with $O(n^\alpha)$ colors  
could be solved in $\OO(n^{2+\alpha-\eps})$ time for some constant $\alpha\le (1-\eps)/4$, then
\APSP\ could be solved in $\OO(n^{3-\eps})$ time.
\end{theorem}
\begin{proof}
The $(\min,+)$-product of two $n\times n$ real matrices,
and thus \APSP, reduces to $n/d$ rectangular $(\min,+)$-products between $n \times d$ matrices and $d \times n$ matrices.
Choose $d=n^\alpha$.  Since $\alpha\le (1-\eps)/4$, we have $d^4n^\eps\le n$. Then the theorem is immediately implied by Lemma~\ref{lem:distincteq}.
\end{proof}

\subsection{\OV\ \texorpdfstring{$\rightarrow$}{rightarrow} \ColorBMM}

We can similarly reduce \OV\ to \ColorBMM, since the former also reduces to computing expressions involving ANDs of ORs (as exploited in Abboud, Williams, and Yu's \OV\ algorithm~\cite{abboud2014more}).  In fact, the reduction from \OV\ is even simpler, and more efficient, than the reduction from \APSP.

\begin{lemma}\label{lem:colorbmm:ov}
One instance of \OV\ for $n$ Boolean vectors in $f$ dimensions
reduces in $O(ndf)$ time to
one instance of \ColorBMM\ for 
an $(n/d)\times (d^2f)$ and a $(d^2f)\times (n/d)$ Boolean
matrix with $d^2$ colors for any given $d\le n$.
\end{lemma}
\begin{proof}
Divide the input set $A$ and $B$ into groups $A_1,\ldots,A_{n/d}$
and $B_1,\ldots,B_{n/d}$ of $d$ Boolean vectors each.  Let $A[i,k,s]$ be the
$s$-th bit of the $k$-th vector in $A_i$ and let $B[j,\ell,s]$ be the
$s$-th bit of the $\ell$-th vector in $B_j$.

For each $i,j\in [n/d]$, there are no orthogonal pairs of vectors
in $A_i\times B_j$ iff
\[
 \bigwedge_{k,\ell\in [d]}\bigvee_{s\in [f]} (A[i,k,s]\wedge B[j,\ell,s]).
\]

Thus, the answer can be determined by solving \ColorBMM\ on the following matrices $\AAA$ and~$\BBB$:
\begin{enumerate}
\item For each $i\in [n/d]$, $k,\ell\in [d]$, and $s\in [f]$,
let $\AAA[i,(k,\ell,s))]=A[i,k,s]$.
\item For each $j\in [n/d]$, $k,\ell\in [d]$, and $s\in [f]$,
let $\BBB[(k,\ell,s),j]=B[j,\ell,s]$.
\item Define $\col((k,\ell,s))=(k,\ell)$.
\end{enumerate}
The inner dimension (i.e., the number of triples $(k,\ell,s)$) is $d^2f$.
\end{proof}

\begin{theorem}\label{thm:colorbmm:ov}
If \ColorBMM\ for two $N\times N$ Boolean matrices could be solved in $\OO(N^{3-\eps})$ time, then
\OV\ for $n$ Boolean vectors in $f$ dimensions could be solved in $\OO(f^{O(1)}n^{2-2\eps/3})$ time.

More generally, if \ColorBMM\ for two $N\times N$ Boolean matrices with $N^\alpha$ colors  
could be solved in $\OO(N^{2+\alpha-\eps})$ time for some constant $\alpha\le 1$, then
\OV\ for $n$ Boolean vectors in $f$ dimensions could be solved in $\OO(f^{O(1)}n^{2-2\eps/(2+\alpha)})$ time.
\end{theorem}
\begin{proof}
Choose $d$ so that $d^2 = (n/d)^\alpha$, i.e., $d=n^{\alpha/(2+\alpha)}$.  Since $\alpha\le 1$, we have $d \le n^{1/3}$ and so
$d^2\le n/d$. Thus, if there is an $\OO(N^{2+\alpha-\eps})$ time algorithm for \ColorBMM\ with $O(N^\alpha)$ colors, we can solve \OV\ in 
 $\OO((f^{O(1)}n/d)^{2+\alpha-\eps})=\OO(f^{O(1)}n^{2-2\eps/(2+\alpha)})$ time.
\end{proof}

\section{Triangle Collection and Triangle-Collection*}
\label{sec:TC}

We first show reductions between variants of \TCollong\ in this section. Their relationships are depicted in Figure~\ref{fig:TC}. 
\begin{figure}[ht]
    \centering
    \begin{tikzpicture}
    
        \node at(0, 0)  [anchor=center] (TC){\TCol};
        \node at(0, 4)  [anchor=center] (TC1){\TC};
        \node at(6, 4)  [anchor=center] (TC2){\TCss};
       \node at(6, 0)  [anchor=center] (TCl){\TCsss};
       
       \draw[->,line width=1pt] (TC1) to[]  node[left] {} (TC);
	   \draw[->,line width=1pt] (TC2) to[]  node[left] {} (TC1);
	   \draw[->,line width=1pt, bend right] (TCl) to[]  node[left] {} (TC);
	   \draw[->,line width=1pt] (TC1) to[]  node[left] {} (TCl);
	   
       \draw[->, line width=1pt] (TCl) to[]  node[right] {Thm.~\ref{thm:TCLC:TC**}} (TC2);
       \draw[->,dashed, line width=1pt, bend right] (TC) to[]  node[above] {Thm.~\ref{thm:TC:TCLC}} (TCl);

    \end{tikzpicture}
    \caption{The reductions for variants of \TCol. All unlabeled arrows follow from problem definitions. The reduction from \TCsss\ to \TCss\ has a $\poly(p)$ loss if the parameter is  $p$. The dashed arrow represents a sub-cubic reduction that only holds when \TCsss\ has parameter $n^\eps$ for $\eps > 0$. All reductions also hold for the All-Color-Pair variants. }
    \label{fig:TC}
\end{figure}
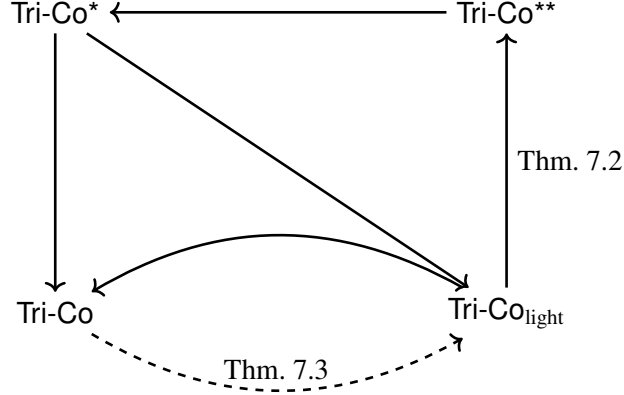

The original version of \TCollong\ (\TCol) as defined in \cite{abboud2018matching} is as follows: given a graph $G=(V,E)$ with node colors $\col:V\mapsto [K]$, determine whether for all triples of distinct colors $a,b,c\in K$ there exists some triangle $x,y,z\in V$ such that $\col(x)=a,\col(y)=b,\col(z)=c$.

We first give a simple proof that the original \TCol\ is equivalent to the tripartite version of the problem, whose definition we recall here:
\TriangleCollection*

Let us call this tripartite-\TCol\ for now. 
The goal of the lemma below is to show that \TCol\ is equivalent to its tripartite version. After the lemma, when we talk about \TCol\ we will mean tripartite-\TCol. 

\begin{lemma} The Original \TCol\ (resp. \ACPTCb) and tripartite-\TCol\ (resp. tripartite-\ACPTCb) are equivalent.\end{lemma}

\begin{proof}
Suppose we are given an instance $G=(V,E)$ with colors $\col:~V\mapsto K$ for a set $K$, of \TCol,\ and we want to know if for all triples of {\em distinct} colors there is a triangle of that color triple. We will add to $G$ the following tripartite graph. For every color $k\in K$, add three nodes $k,k',k''$, and add the edge $(k',k'')$. For every pair of colors $k,\ell\in K$, add edges $(k,\ell'),(k,\ell'')$. This new part of the graph has the property that every triple of colors, at least two of which are the same, has a triangle of that triple. The new part does not add any new triangles for distinct color triples. Let's call the new graph $G^*$.

Now we can assume that in our given graph $G^*=(V^*,E^*)$ with colors $\col:~V^*\mapsto K$,  we want to solve the problem of whether every color triple (not necessarily distinct) contains a triangle.

We now create an instance $G'$ of tripartite-\TCol\ by simply creating three copies of $V^*$, $V_1,V_2,V_2$ where $V_i=\{v_i~|~v\in V^*\}$.
We create three sets of disjoint color sets $K_1,K_2,K_3$ where $K_i=\{(k,i)~|~k\in K\}$, so that if for $v\in V^*$, $\col(v)=k$, then in $G'$, for each $i=1,2,3$, $\col(v_i)=(k,i)$.
For every edge $(u,v)\in E^*$ we add edges $(u_i,v_j)$ for all $i\neq j$, $i,j\in [3]$. 

Then for a triple of colors $(a,1)\in K_1, (b,2)\in K_2, (c,3)\in K_3$, there is a triangle $x_1\in V_1, y_2\in V_2, z_3\in V_3$ with those colors iff color triple $a,b,c\in K$ has a triangle $x,y,z$ in $G^*$. Thus solving tripartite-\TCol\ in $G'$ solves (non-distinct color triple) \TCol\ in $G^*$, and hence \TCol\ in $G$.

Now let us reduce tripartite-\TCol\ to \TCol.\ Given an instance $G=(V,E)$ of tripartite-\TCol\ with partitions $A,B,C$ and colors $\col:A \mapsto K_A$, $\col:B \mapsto K_B$, $\col:C \mapsto K_C$, we only need to make sure that every triple of distinct colors that has at least one pair of colors in $K_A$ or in $K_B$ or in $K_C$ contains a triangle. A way to do this was already shown in \cite{abboud2018matching}. There are several way to do this. For instance, for each choice $X\in \{A,B,C\}$, add a node $k^X$ for every color $k\in (K_A\cup K_B\cup K_C)\setminus K_X$ and add edges between all pairs of nodes $k^X$ and $\ell^X$.
This adds three cliques, and ensures that for every triple of distinct colors that are not from $K_A\times K_B\times K_C$, there is a triangle. At the same time, this does not change whether some triple of distinct colors from $K_A\times K_B\times K_C$ has a triple or not.

The equivalence between  \ACPTCb\ and  tripartite-\ACPTCb\ follows from the same reductions. 
\end{proof}

In light of the above equivalence, from now on when we refer to \TCol,\ we will mean its {\em tripartite} version, as defined above.

Recall the definition of \TC\ defined by Abboud, Vassilevska W. and Yu~\cite{abboud2018matching}. 

\TriangleCollectionStar*

Also recall the definitions of the variants \TCss\ and \TCsss. 

\TriangleCollectionStarStar*

\TriangleCollectionStarStarStar*

We now show that \TC,\ \TCss\ and \TCsss\ are equivalent up to $n^{o(1)}$ factors when their parameters are $n^{o(1)}$. It suffices to show that \TCsss\ reduces to \TCss,\ as \TCss\ is a special case of \TCsss,\ and the original problem \TC\ is sandwiched between them.

\begin{theorem}
\label{thm:TCLC:TC**}
An instance of \TCsss\ (resp. \ACPTCsss) with parameter $p$ and $n$ colors (and hence $\leq np$ nodes) can be reduced in $O(n^2p^3)$ time to an instance of \TCss\ (resp. \ACPTCss) on $O(np^3)$ nodes and $n$ colors, and parameter $t=p^3$.
\end{theorem}

\begin{proof}
Given an instance $G$ of \TCsss\ with parameter $p$, recall by the definition of \TCol\ $G$ has partitions $A,B,C$ and the colors of $A,B,C$ are disjoint. Let's call the nodes $(v,i)$ where $v$ is a color and $i\leq p$, meaning the $i$th node of color $v$.

Let's create an instance $G'$ of \TCss\ with parameter $p^3$ as follows.
For every node $(a,i)$ of $A$ create $p^2$ copies $(a,i,j,k)$ for all $j,k\in [p]$.
Similarly, for every node $(b,j)$ of $B$ create $p^2$ copies $(b,i,j,k)$ for all $i,k\in [p]$, and for every node $(c,k)$ of $C$ create $p^2$ copies $(c,i,j,k)$ for all $i,j\in [p]$.

$G'$ will be a disjoint union of $p^3$ graphs: one for each choice of a triple $i,j,k\in [p]$.

For fixed $i,j,k\in [p]$, add an edge between each node $(a,i,j,k)$ in $A$ and each node $(b,i,j,k)$ in $B$ whenever $(a,i)$ and $(b,j)$ had an edge in $G$. Similarly, add an edge between each node $(a,i,j,k)$ in $A$ and each node $(c,i,j,k)$ in $C$ whenever $(a,i)$ and $(c,k)$ had an edge in $G$. Finally, add an edge between each node $(b,i,j,k)$ in $B$ and each node $(c,i,j,k)$ in $C$ whenever $(b,j)$ and $(c,k)$ had an edge in $G$.

Notice that for any triple of colors $a,b,c$, there is a triangle $(a,i,j,k),(b,i,j,k),(c,i,j,k)$ for some $(i,j,k)$ if and only if for some $(i,j,k)$ $(a,i),(b,j),(c,k)$ was a triangle in $G$.

Thus, $G'$ is a YES instance of \TCol\ if and only if $G$ is a YES instance of \TCol.\ The number of colors is the same, and the parameter is now $p^3$.

The reduction from \ACPTCsss\ to \ACPTCss\ is essentially the same. 
\end{proof}

Now we show that \TCol\ reduces in truly subcubic time to \TCsss\ with parameter $n^\eps$ for any constant $\eps>0$. Thus, \TC\ with parameters $n^\eps$ is subcubically equivalent to \TCol.

\begin{theorem}
\label{thm:TC:TCLC}
For any $\eps> 0$, an instance of \TCol\ (resp. \ACPTC) on $n$ nodes can be reduced to an instance of \TCsss\ (resp. \ACPTCsss) on at most $n$ nodes and parameter $n^\eps$ in time $n^{3-\min\{\eps, 0.3\}+o(1)}$. 
\end{theorem}

\begin{proof}
Let $G$ be a given instance of \TCol\ on $n$ nodes. %

Let's look at the number of nodes that are colored $c$ for each $c$; we will call this the frequency of the color. Consider every color $c$ with frequency at least $n^\eps$. The number of such colors is at most $\leq n^{1-\eps}$.

For each such color $c$, let $n^p$ be the number of nodes of color $c$. Create two matrices $A$ and $B$. $A$ is $n\times n^p$ and $A[v,i]$ is $1$ if node $v$ has an edge to the $i$th node of color $c$. Similarly, $B$ is $n^p\times n$ and $B[i,v]=1$ if node $v$ has an edge to the $i$th node of color $c$.
Multiply $A$ by $B$, and for every pairs of nodes $u,v$ connected by an edge, increment a counter for the color triple $color(u),c,color(v)$. Thus in $n^{\omega(1, p, 1)+o(1)}$ time we get all color triples that include $c$ and have a triangle, where $\omega(1, p, 1)$ is the rectangular matrix multiplication exponent for multiplying an $n \times n^p$ matrix with an $n^p \times n$ matrix. 

Thus, in $\max_{\eps \le p \le 1} n^{\omega(1, p, 1)+1-p+o(1)}$ time we can handle all colors with frequency at least $n^\eps$. Note that $\omega(1, p, 1) \le \omega(1, \eps, 1) + (p - \eps)$ since we can first split a pair of an $n \times n^p$ matrix and an $n^p \times n$ matrix to pieces of size $n \times n^\eps$ and $n^\eps \times n$ each, and aggregate the results together. Therefore, we can bound the running time as $\max_{\eps \le p \le 1} n^{\omega(1, p, 1)+1-p} \le \max_{\eps \le p \le 1} n^{\omega(1, \eps, 1) + (p - \eps)+1-p} = n^{\omega(1, \eps, 1) + 1 -\eps}$. If $\eps < 0.3$, then $\omega(1, \eps, 1) = 2$ as shown in \cite{legallurr}, so the reduction runs in $n^{3-\eps + o(1)}$ time; if $\eps \ge 0.3$, then we use $\omega(1, \eps, 1) + 1 -\eps \le \omega(1, 0.3, 1) + 1 - 0.3 = 2.7$, so the reduction runs in $n^{2.7 + o(1)}$ time.

What remains are the color triples that have all their colors with frequency less than $n^\eps$. We take the subgraph of $G$ consisting only of nodes with infrequent colors. This is an instance of \TCsss\ with parameter $n^{\eps}$. 

Thus if \TCsss\ with parameter $n^\eps$ can be solved in $O(n^{3-\delta})$ time for some $\delta>0$, then we can solve \TCol\ in time $O(n^{3-\min\{\eps,0.3\}+o(1)}+n^{3-\delta})$.

The reduction between the All-Color-Pairs versions is essentially the same. 
\end{proof}

Recall our main theorem is the following. 
\mainThm*

In Section~\ref{sec:mono_to_ACPTC}, we show a reduction from \AEMonoTri\ to \ACPTC. Combined with the reduction from \APSP\ to \AEMonoTri\ from Theorem~\ref{thm:apsp:monotri} and the reduction from \ThreeSUM\ to \AEMonoTri\ from Theorem~\ref{thm:3sum:monotri}, we obtain the \APSP\ and \ThreeSUM\ hardness in the main theorem. In Section~\ref{sec:colorBMM_to_ACPTC}, we show a reduction from \ColorBMM\ to \ACPTC. By combining with the reduction from \OV\ to \ColorBMM\ and unrolling the whole reduction, we in fact obtain a reduction from \OV\ to \TC. This shows the \OV\ hardness of the main theorem, and thus will conclude the proof of the main theorem. 

\subsection{From All-Edges Monochromatic Triangle to All-Color-Pairs Triangle Collection}
\label{sec:mono_to_ACPTC}
Here we show that the \AEMonoTri\ problem can be reduced to the \ACPTC\ problem. By our results above, it suffices to reduce  \AEMonoTri\ to \ACPTCsss\ with parameter $n^{o(1)}$. 

Recall that in \AEMonoTri\ we are given an $n$ node graph $G=(V,E)$ with colors $\col:E\mapsto [n^2]$ on the edges and we want to know for every edge $(u,v)\in E$ if there exists a $w\in V$ so that $(u,v),(v,w),(w,u)\in E$ (i.e. they form a triangle) and $\col(u,v)=\col(v,w)=\col(w,u)$.

\begin{theorem}
\label{thm:mono:ACPTC}
An $n$ node instance of \AEMonoTri\ can be reduced in $O(n^2\log n)$ time to an $O(n\log n)$-node instance of \ACPTCsss\ with parameter $O(\log n)$.  
\end{theorem}

\begin{proof}
Let $G=(V,E)$ be the instance of \AEMonoTri\ with edge colors $\col :E\mapsto [n^2]$.
Without loss of generality, $G$ is tripartite with node parts $A,B,C$.

We will build an instance of \ACPTCsss\ $H$.
For every $t\in [2\log n]$ we will create a graph $G_t$ and $H$ will be the disjoint union of these $O(\log n)$ graphs.

For a fixed $t\in [2\log n]$, the nodes of $G_t$ are as follows. $G_t$ will be a tripartite graph on parts $A',B',C'$. For every node $z\in C$ of $G$, we add a node $z$ to $G_t$; $z$ has color $z$  (we associate the nodes of $G$ with the integers $[n]$). These nodes will be in part $C'$. For every node $v\in A\cup B$ of $G$ we
create two copies of $v$ in $G_t$: $v_0$ and $v_1$, both having color $v$. If $v$ was in $A$, $v_0,v_1$ are in the partition $A'$, and if $v$ was in $B$, then $v_0,v_1$ are in partition $B'$.

For every pair of nodes of $G$, $v\in A\cup B$ and $z\in C$, if the $t$th bit of the color $\col(v,z)$ is $b$, then add an edge between $v_b$ and $z$.
For every pair of nodes of $G$, $v\in A$ and $v'\in B$, add an edge between $v_p$ and $v'_{p'}$ for all choices of $p$ and $p'$ s.t. $p\neq p'$. In addition, add an edge between between $v_b$ and $v'_{b}$
if the $t$th bit of the color $\col(v,v')$ is not $b$.

Notice that for every $z\in C, v\in A, v'\in B$ (which are now a triple of colors), $G_t$ does not have a triangle of these colors iff the $t$th bit of the colors of $(z,v),(z,v'),(v,v')$ in $G$ match. Thus, there are no triangles colored $z,v,v'$ in $H=\cup_t G_t$ if and only if the colors of $(z,v),(z,v'),(v,v')$ in $G$ match for all choices of $t$, and hence if and only if the colors are exactly the same.

By construction, the number of nodes of each color is at most $O(\log n)$. 
Thus we have reduced \AEMonoTri\ to an $O(n\log n)$-node instance of \ACPTCsss.

\end{proof}

\subsection{\ColorBMM\ \texorpdfstring{$\rightarrow$}{rightarrow} \ACPTCb\ and \OV\ \texorpdfstring{$\rightarrow$}{rightarrow} \TC}
\label{sec:colorBMM_to_ACPTC}
We first give a simple reduction from \ColorBMM\ to \ACPTCb.

\begin{lemma}
\label{lem:colorBMM:ACPTC}
Any instance of \ColorBMM~ on $n\times n$ matrices can be reduced in $O(n^2)$ time to an instance of \ACPTCb~ on $n$ nodes. 

If the maximum number of $k$ of any given color in the \ColorBMM~ instance is $n^{o(1)}$, then \ColorBMM{} also reduces to \ACPTC\ with parameter $n^{o(1)}$.\label{lem:colorbmmtoacptc}
\end{lemma}

\begin{proof}
Let $A$ and $B$ be two $n\times n$ matrices that constitute an instance of \ColorBMM.
Let $\col: [n]\rightarrow [K]$ be the color function of the columns of $A$ and rows of $B$. Note that without loss of generality $K\leq O(n)$.

We will create an instance $G$ of \ACPTCb~ with $O(n)$ colors.

For every color $k\in [K]$ and every $t\in [n]$ such that $\col(t)=k$, create a node $k_{2,t}$ of color $k$.

For every $i\in [n]$, create a node $i_1$ of color $i_1$ and a node $i_3$ of color $i_3$, and add an edge between $i_1$ and $j_3 $ for all $i,j\in [n]$.

For every pair $i\in [n], t\in [n]$ with $\col(t)=k$, add an edge $(i_1,k_{2,t})$ if $A[i,t]=1$ and add an edge $(i_3,k_{2,t})$ if $B[t,i]=1$.

Notice that for a fixed triple $i\in [n],j\in [n],k\in [K]$, there is a triangle with colors $i_1,k,j_3$ if and only if there is some $t\in [n]$, $A[i,t]\cdot B[t,j]=1$. 

\ACPTCb~ asks for every pair of colors $i_1,j_3$ to compute whether there is some $k$ with no triangles of color triple $i_1,j_3,k$, conversely, whether for all colors $k$ there is a triangle with color triple $i_1,j_3,k$, i.e. whether for all $i,j,k$ there is some $t$ of color $k$ such that $A[i,t]\cdot B[t,j]=1$.

We get an instance of \ACPTCb~ with $O(n)$ nodes and $O(n)$ colors.

Now suppose that the number of $k$ in any given color of the \ColorBMM~ instance is at most $T\leq n^{o(1)}$. In this case, the number of nodes of each color in the above reduction is at most $n^{o(1)}$, so we actually get an instance of \ACPTCsss\ with parameter $n^{o(1)}$. By our reduction from \ACPTCsss\ with parameter $n^{o(1)}$ to \ACPTC, we get an instance of \ACPTC\ with $n^{1+o(1)}$ nodes and $O(n)$ colors. 
\end{proof}

We get as a corollary a reduction from \OV\ to \TC\  (not only to \ACPTCb).

\begin{corollary}
\label{cor:OV:TC}
For any $1 \le d\leq n$, \OV\ for $n$ Boolean vectors in $f$ dimensions reduces in $O(n^2/d^2 +ndf)$ time to an instance of \TCsss\  on $O(nf/d + d^2)$ nodes with parameter $f$. Thus if \OV\ requires $n^{2-o(1)}$ time for $f=n^{o(1)}$, then \TCsss\  on $N$ nodes with parameter $N^{o(1)}$ requires $N^{3-o(1)}$ time, then so does \TC\ with parameter $n^{o(1)}$. 
\end{corollary}

\begin{proof}
Consider the reduction from Lemma~\ref{lem:colorbmm:ov}
from \OV\ to \ColorBMM. 
For any $d\leq n$, it partitioned the input sets of vectors $A$ and $B$ into groups $\{A_i\}_{i\in [n/d]}$ and $\{B_i\}_{i\in [n/d]}$ of $d$ vectors each, letting $A[i,j,s]$ be the $s$th bit of the $k$th vector of $A_i$ and $B[j,\ell,s]$ be the $s$th bit of the $\ell$th vector of $B_j$.

It then created an $n\times d^2f$ matrix $\AAA$ with $\AAA[i,(k,\ell,s))]=A[i,k,s]$ and an $d^2f\times n$ matrix $\BBB$ with $\BBB[(k,\ell,s),j]=B[j,\ell,s]$. The color of $(k,\ell,s)$ was $(k,\ell)$. The number of columns of $\AAA$ of each color is thus at most $f$.

We can then apply our reduction from \ColorBMM~ to \ACPTCsss~ from the proof of Lemma~\ref{lem:colorbmmtoacptc} to show that \OV\ reduces to \ACPTCsss.

However, notice that in the above reduction from \OV\ to \ColorBMM, there are no pairs of orthogonal vectors iff for all $i,j\in [n/d]$ and $k,\ell\in [d]$, there exists some $s\in [f]$ such that $\AAA[i,(k,\ell,s))]=\BBB[(k,\ell,s),j]=1$. When we create the \ACPTCsss~ instance from the \ColorBMM~ instance in our reduction, it suffices to figure out whether for all triples of colors there is some triangle, which is exactly the \TCsss\ problem.
\end{proof}

\section{Other Reductions}
\label{sec:other}

\subsection{Hardness of \AEColorSparseTri}

\begin{theorem}\label{thm:colorsparsetri}
If \AEColorSparseTri\ with $m$ edges and degeneracy $O(m^\alpha)$
could be solved in $\OO(m^{1+\alpha-\eps})$ time for some constant $\alpha\le 1/2$, then
\OV\ for $n$ Boolean vectors in $f$ dimensions could be solved in $\OO(f^{O(1)}n^{2-2\eps/(2-\alpha)})$ time.
\end{theorem}
\begin{proof}
By Lemma~\ref{lem:colorbmm:ov},
\OV\ reduces to $O(n/d^3)$ instances of \ColorBMM\ for 
an $(n/d)\times (d^2f)$ and an $(d^2f)\times d^2$ Boolean
matrix, if $d\le n^{1/3}$.  Each such instance reduces to 
an instance of \AEColorSparseTri, by mapping the two given matrices into
a tripartite graph (the first matrix maps to
edges between the left and middle parts, the second matrix maps
to edges between the middle and right parts, and we have
a complete bipartite graph between the left and right parts).
Clearly, this tripartite graph has $O((n/d)d^2f)= O(dnf)$ edges and
degeneracy $O(d^2f)$.

Choose $d$ so that $d^2 = (dn)^\alpha$, i.e., $d=n^{\alpha/(2-\alpha)}$.  Since $\alpha\le 1/2$, we have $d \le n^{1/3}$.
The time bound for \OV\ is $\OO((n/d^3)\cdot (f^{O(1)}dn)^{1+\alpha-\eps})=\OO(f^{O(1)}n^{2-2\eps/(2-\alpha)})$.
\end{proof}

Since \AEColorSparseTri\ generalizes \AESparseTri\ (which corresponds to the case with just 1 color), we can combine with 
Theorems~\ref{thm:apsp:deg}, \ref{thm:3sum}, and \ref{thm:colorsparsetri} to conclude 
that the 
\AEColorSparseTri\ problem with $m$ edges and degeneracy $D$ 
has a lower bound near $mD$ for $D\ll m^{1/5}$ 
if any one of the Real APSP, Real 3SUM, or OV hypothesis is true.
Since one can enumerate all triangles in a graph with $m$ edges and degeneracy $D$ in $O(mD)$ time~\cite{chiba1985} and thus solve \AEColorSparseTri, we have thus obtained a \emph{tight} conditional lower bound (up to $n^{o(1)}$ factor) for this problem under Hypothesis~\ref{conj:conj2}, at least for a restricted range of $D$.

\subsection{Real-to-Integer Reductions}
\label{sec:real-to-int}

Since \AESparseTri\ with $m$ edges can be reduced back to \IntAllThreeSUM\ on
$m$ integers~\cite{jafargholi2016mathrm}, Theorem~\ref{thm:3sum} immediately implies:

\begin{corollary}
\label{cor:real-3sum:int-3sum}
If \IntAllThreeSUM\ could be solved in $\OO(n^{6/5-\eps})$ time,
then \AllThreeSUM\ could be solved in $\OO(n^{2-5\eps/3})$ time 
with Las Vegas randomization.
\end{corollary}

\newcommand{\IntAllConvThreeSUM}{\mbox{\sf Int-All-Nums-Convolution-3SUM}}
\newcommand{\AllConvThreeSUM}{\mbox{\sf Real-All-Nums-Convolution-3SUM}}
\newcommand{\ConvThreeSUM}{\mbox{\sf Real-Convolution-3SUM}}

Although the above bound is weak, such real-to-integer reductions may have other implications.
For example, it is known that for \IntAllThreeSUM\ reduces to \IntAllConvThreeSUM\footnote{
In {\sf Convolution-3SUM}, we are given three arrays $A,B,C$ of $n$ numbers, and want to decide the existence of indices $i$ and $j$ with $C[i]=A[j]+B[i-j]$.  In the ``All-Nums'' version, we want to decide, for each $i$, the existence of an index $j$ satisfying the same equation.
}~\cite{patrascu2010towards,kopelowitz2016higher,ChanHe}, without increasing the
running time except by polylogarithmic or sublogarithmic factors, using Las Vegas
randomization (these reductions were stated for the original non-``All'' versions).  
By the above corollary, if
\AllConvThreeSUM\ could be solved in $O(n^{6/5-\eps})$ time, then \AllThreeSUM\ could be solved in subquadratic time.  Despite the weakness of the bound, this is interesting, since
the question of whether \ThreeSUM\ could be reduced to \ConvThreeSUM\ has been raised in previous papers~\cite{kopelowitz2016higher,ChanHe}.

One application is conditional lower bounds for
the \emph{jumbled indexing problem} under the Real 3SUM hypothesis.  
Amir, Chan, Lewenstein, and Lewenstein~\cite{AmirCLL14} proved
lower bounds for the problem under the Integer 3SUM hypothesis.  By combining their
reduction from \IntAllThreeSUM\ to jumbled indexing and
our reduction from \AllThreeSUM\ to \IntAllThreeSUM, one can prove that assuming
the Real 3SUM hypothesis, 
no data structure for jumbled indexing can simultaneously
have $O(n^{1+\alpha_\sigma})$ preprocessing time and $O(n^{\beta_\sigma})$ query time
for some constants $\alpha_\sigma,\beta_\sigma>0$ when the alphabet size $\sigma$ is a sufficiently large constant. Though our bound is weaker than~\cite{AmirCLL14}'s original bound, it is based on a more believable hypothesis.

We can obtain similar real-to-integer reductions for {\sf Exact-Triangle}, by combining our previous reduction from \AEExactTri\ to \AEMonoTri, and the following reduction from \AEMonoTri\ to \IntAEExactTri:

\begin{lemma}
\label{lem:AEMonoT:IntExactT}
\AEMonoTri\ with $n^\alpha$ colors
reduces to \IntAEExactTri\ for integer edge weights in $\pm[n^{\alpha+o(1)}]$.
\end{lemma}
\begin{proof}
To produce the reduction, we use a result from additive combinatorics about Salem-Spencer set. A Salem-Spencer set is a set of numbers that do not contain any $3$-term arithmetic progressions. In other words, if a set $S$ is a  Salem-Spencer set, then for any $a, b, c \in S$, $a+c=2b$ if and only if $a=b=c$. It is known that $[N]$ contains a large Salem-Spencer set. 

\begin{quote}{\bf Fact.} \em\cite{behrend1946sets}
There exists a subset of $[N]$ that is a Salem-Spencer set of size $N/e^{O(\sqrt{\log N})} \ge N^{1-o(1)}$. Furthermore, we can find such a set efficiently, in $\OO(N)$ time. \end{quote}

First, we find a Salem-Spencer set $S$ of size $n^\alpha$ in $[n^{\alpha+o(1)}]$. Then we find an arbitrary injective mapping $f$ from the set of colors to elements in set $S$. Say $f(x)$ maps the color $x$ to a number in $S$. Without loss of generality, we assume the \AEMonoTri\ instance on graph $G$ with $n^\alpha$ colors is tripartite on partitions $I, J, K$. For every edge $e$ between $I$ and $K$ or between $J$ and $K$ with color $x$, we replace its color with an integer $f(x)$; for every edge $e$ between $I$ and $J$ with color $x$, we replace its color with $-2f(x)$. For any triangle in $G$ with edge colors $x_1, x_2, x_3$, it becomes a zero triangle if and only if $f(x_1) + f(x_2) - 2f(x_3) = 0$. Since $S$ is a Salem-Spencer set, it is equivalent to $f(x_1) = f(x_2) = f(x_3)$. Since $f$ is injective, it is further equivalent to $x_1=x_2=x_3$. Thus, a zero triangle in the new graph corresponds to a monochromatic triangle in the original graph. The other direction is more straightforward.  Therefore, by running \IntAEExactTri\ on the new graph, we can solve the original \AEMonoTri\ instance. 
\end{proof}

Since in Section~\ref{sec:mono} we showed super-quadratic lower bounds of \AEMonoTri\ based on the Real APSP hypothesis, the Real Exact-Triangle hypothesis or the Real 3SUM hypothesis, Theorem~\ref{thm:exacttri:monotrinoncount} implies super-quadratic lower bound of \IntAEExactTri\ based on any of the three hypotheses. In particular, we obtain the following corollary:

\begin{corollary}
\label{cor:real-exactT:int-exactT}
If \IntAEExactTri\  for integers in $\pm[n^{2\alpha}]$ could be solved in
$\OO(n^{2+\alpha-\eps})$ time for some constant $\alpha\le (1-\eps)/3$, then
\AEExactTri\  could be solved in $n^{3-\eps+o(1)}$ time.
\end{corollary}

A similar proof, together with Theorem~\ref{thm:exacttri:monotricount}, also implies the following
(note that \IntAEExactTriCount\ easily reduces to two instances of \IntAENegTriCount\ by subtracting counts). 

\begin{corollary}
If \IntAEExactTriCount\ (or \IntAENegTriCount) for integers in $\pm[n^\alpha]$ could be solved
in $\OO(n^{2+\alpha-\eps})$ time for some constant $\alpha\le (1-\eps)/2$, then
\AEExactTriCount\ (or \AENegTriCount) could be solved in $n^{3-\eps+o(1)}$ time.
\end{corollary}

Note that \AEExactTriCount\ (and \AENegTriCount) for integer edge weights in $\pm[n^\alpha]$ can be solved in $n^{2+\alpha+o(1)}$ time if $\omega=2$, by matrix multiplication on $\OO(n^\alpha)$-bit numbers.  Thus, the above result interestingly says that any improved algorithm for the small integer weight case would lead to an improved algorithm for the general real case (if $\omega=2$).

We remark that using Lemma~\ref{lem:AEMonoT:IntExactT}, it is possible to obtain another route of reduction from \AEMonoTri\ to \ACPTC. First, we use Lemma~\ref{lem:AEMonoT:IntExactT} to reduce \AEMonoTri\ to \IntAEExactTri, then adapt the reduction by Abboud, Vassilevska W. and Yu \cite{abboud2018matching} from \IntExactTri\ to \TC\ to get a reduction from \IntAEExactTri\ to \ACPTC. However, 
our proof of Theorem~\ref{thm:mono:ACPTC} is simpler and only introduces $\OO(1)$ overhead in contrast to the $n^{o(1)}$ overhead introduced by this reduction.

\subsection{An Application to String Matching}

Given two strings $a_1\cdots a_n$ and $b_1\cdots b_n$ in $\Sigma^*$,
define their \emph{Hamming similarity} to be $|\{i:a_i=b_i\}|$ (i.e., it is $n$ minus their Hamming distance).
Define their \emph{distinct Hamming similarity} to be
$|\{a_i: a_i=b_i\}|$.  
In other words, instead of counting the number of positions that matches, we count the number of distinct alphabet symbols $c\in\Sigma$ that are matched (where there exists $i$ with $a_i=b_i=c$).

We consider the following problem: given a text string $T = t_1\cdots t_N$
and a pattern string $P = p_1\cdots p_M$ in $\Sigma^*$ with $M\le N$,
compute the distinct Hamming similarity between $P$ and $t_{i+1}\cdots t_{i+M}$ for every $i=0,\ldots,N-M$.  Call this the \emph{pattern-to-text distinct Hamming similarity} problem.  The corresponding problem
for standard Hamming similarity/distance is well-studied, for which the current best  algorithms~\cite{matching1,matching2} have running time $\OO(N^{3/2})$ in terms of $N$, and
$\OO(|\Sigma| N)$ in terms of $N$ and $|\Sigma|$, but unfortunately matching lower bounds are currently not known under the standard conjectures.\footnote{
A conditional lower bound near $N^{3/2}$ for ``combinatorial algorithms'' was actually known, under the hypothesis that Boolean matrix multiplication requires near cubic time for combinatorial algorithms, by a reduction attributed to Indyk (see e.g. \cite{GawrychowskiU18}) (which we will actually adapt in our own reduction later).  However, known algorithms use Fast Fourier Transform (FFT), and there is no precise definition of combinatorial algorithms, especially if one wants to include FFT but forbid Strassen-like matrix multiplication methods.
For general non-combinatorial algorithms, Indyk's lower bound was only around $N^{\omega/2}$.
}
The known algorithms easily generalize to solve the pattern-to-text distinct Hamming similarity problem.
Using our \OV\ $\rightarrow$ \ColorBMM\  reduction, we are able to prove nearly matching lower bounds for the distinct version of pattern-to-text  Hamming similarity under the OV hypothesis.

\begin{lemma}
\ColorBMM\ for an $n\times df$ Boolean matrix $A$ and a $df\times n$ Boolean matrix $B$ with colors in $[d]$, such that there are $f$
indices with each color, reduces to the pattern-to-text distinct Hamming similarity
problem for two strings with length $O(n^2f)$ and alphabet size $O(d)$,
assuming that $d\le n$.
\end{lemma}
\begin{proof}
First we move the columns of $A$ and the rows of $B$ so that the color of each index $k$ is $k\bmod n$, if $k\bmod n$ is in $[d]$.  This can be accomplished by taking each $\ell\in [d]$, and mapping the smallest index with color $\ell$ to index $\ell$, the second smallest index with color $\ell$ to index $n+\ell$, the third to $2n+\ell$, etc.  All unused indices are colored ``$!$'', with the corresponding empty columns of $A$ and rows of $B$ set to false.
The inner dimension (the number of columns of $A$ or the rows of $B$) is now increased to $nf$.  

Our transformation of the matrix problem to a string problem is
a variant of a known reduction from Boolean matrix multiplication to the standard pattern-to-text Hamming distance problem, attributed to Indyk (see e.g. \cite{GawrychowskiU18}).

For the $i$-th row of $A$, we define a corresponding string 
$T_i=a'_{i,1}a'_{i,2}\cdots a'_{i,nf}$ with $a'_{i,k}=\col(k)$ if $A[i,k]$ is true, and $a'_{i,k}=\$$ otherwise.  Here, $\$$ is a new extra symbol.

For the $j$-th column of $B$, we similarly define a corresponding string 
$P_j=b'_{1,j}b'_{2,j}\cdots b'_{nf,j}$ with $b'_{k,j}=\col(k)$ if $B[k,j]$ is true, and $b'_{k,j}=\#$ otherwise.  Here, $\#$ is another new extra symbol.

We define the text string $T=\$^{n(nf+1)}T_1T_2\cdots T_n\$^{n(nf+1)}$,
and the pattern string $P=P_1\#P_2\#\cdots\#P_n$, and solve the
pattern-to-text distinct Hamming similarity problem on these two strings
of length $O(n^2f)$, with alphabet $[d]\cup\{!,\#,\$\}$.

To determine the $(i,j)$-th output entry for \ColorBMM,
consider the shift of the pattern string $P$ so that $P_j$ is aligned with
$T_i$.  Fix $\ell\in [d]$. All occurrences of the symbol $\ell$ in $T$ are at indices congruent to $\ell\pmod{n}$.  In the shifted copy of $P$,
all occurrences of $\ell$ in $P_j$ are also at indices congruent to $t\pmod n$, but the occurrences of $\ell$ in $P_{j\pm 1}$ are at indices
congruent to $(\ell\pm 1)\pmod n$ (because of the separator $\#$
between $P_j$ and $P_{j\pm 1})$, and more generally, the occurrences of $t$ in $P_{j\pm s}$ are at indices congruent to $\ell\pm s\pmod{n}$ for each $s$.
Thus, the only common occurrences of $\ell$ in $T$ and the shifted $P$ are in the $P_j$ portion of $P$, and a common occurrence exists
iff there exists $k\in [nf]$ with $\col(k)=\ell$ such that $A[i,k]\wedge B[k,j]$.
It follows that the $(i,j)$-th output entry for \ColorBMM\ is yes
iff the distinct Hamming similarity between $T$ and the shifted $P$ is $d$ (i.e., only the 3 extra symbols $!,\#,\$$ are not matched).
\end{proof}
 
\begin{theorem}
\label{thm:OV:ptt-Hamming}
If the pattern-to-text distinct Hamming similarity
problem for two strings with length $O(N)$ could be solved in
$O(N^{3/2-\eps})$ time,
then \OV\ for $n$ Boolean vectors in $f$ dimensions could be
solved in $O(f^{O(1)}n^{2-4\eps/3})$ time.

More generally, if the pattern-to-text distinct Hamming similarity
problem for two strings with length $O(N)$ and alphabet size $O(N^\alpha)$ could be solved in
$O(N^{1+\alpha-\eps})$ time for some constant $\alpha\le 1/2$,
then \OV\ for $n$ Boolean vectors in $f$ dimensions could be
solved in $O(f^{O(1)}n^{2-2\eps/(1+\alpha)})$ time.
\end{theorem}
\begin{proof}
Note that in the reduction from Lemma~\ref{lem:colorbmm:ov}, the number of indices per color is indeed bounded by $f$.
By the above lemma, it follows that if the pattern-to-text distinct Hamming similarity
problem with string length $N$ and alphabet size $\sigma$ could be solved in $T(N,\sigma)$ time, then the time bound for \OV\ is
$O(T((n/d)^2f,\,d^2))$, assuming that $d^2\le n/d$.

Choose $d$ so that $d^2 = (n/d)^{2\alpha}$, i.e., $d=n^{\alpha/(1+\alpha)}$.  Since $\alpha\le 1/2$, we have $d\le n^{1/3}$, and so $d^2\le n/d$.  The time bound becomes $O(T((n/d)^2f,\,d^2))=O(f^{O(1)}n^{2-2\eps/(1+\alpha)})$ if $T(N,N^\alpha)=O(N^{1+\alpha-\eps})$.
\end{proof}

\subsection{Hardness of Set-Disjointness and Set-Intersection from \APSP}
\label{sec:disjoint}
In this section, we show tight hardness of \SetDisj\ and \SetInter\  based on the Real APSP hypothesis.

Recall that in the \SetDisj\  problem, one is given a universe $U$, a collection of sets $\mathcal{F} \subseteq 2^U$, and $q$ queries of the form $(F_1, F_2) \in \mathcal{F} \times \mathcal{F}$ asking whether $F_1 \cap F_2 = \emptyset$. Tight hardness of \SetDisj\  was first shown by Kopelowitz, Pettie and Porat~\cite{kopelowitz2016higher} based on the Integer 3SUM hypothesis.  It was later strengthened by Vassilevska W. and Xu~\cite{williamsxumono} to be based on the Integer Exact-Triangle hypothesis. Now we prove the same conditional lower bound, but based on the Real APSP hypothesis. 

\begin{theorem}
\label{thm:setinter}
For any constant $0 < \gamma < 1$, no algorithm for \SetDisj\  where $|U| = O(N^{2-2\gamma})$, $|\mathcal{F}| = O(N)$, each set of $\mathcal{F}$ has size $O(N^{1-\gamma})$, and $q=O(N^{1+\gamma})$ can run in $O(N^{2-\epsilon})$ time for $\epsilon > 0$, under the Real APSP hypothesis.
\end{theorem}
\begin{proof}
It is almost an immediate corollary of Lemma~\ref{lem:apsp2}. Suppose  an $O(N^{2-\epsilon})$ time algorithm for \SetDisj\  exists. We apply the previous reduction in Lemma~\ref{lem:apsp1} and Lemma~\ref{lem:apsp2} 
with $n=N$ and $d = N^{1-\gamma}$ to get $\OO(d)$ instances of  \AESparseTri. 
By the Real APSP hypothesis, each instance requires $n^{2-o(1)}$ time. 
By a close inspection of that reduction, the number of left and right nodes is $\Theta(n)$, the number of edges between the left and right nodes is $O(n^2 / d)$, and each left and right node has $d$ neighbors in the middle part. For each middle node whose degree is at most $n^{1-\delta} / d$ for some $\delta > 0$, we can enumerate all pairs of its neighbors can then ignore the node afterwards. Thus, we can bound the number of middle nodes by  $O(n^\delta d^2)$ by paying $O(n^{2-\delta})$ time.
We can split the instance to $n^\delta$ instances so that the number of middle nodes is $O(d^2)$ in each instance. 
Let $\mathcal{F}$ be the union of the left and right nodes, $U$ be the set of middle nodes, and add the edges between these accordingly. Clearly, this is a valid input to the \SetDisj\  problem, and thus we can solve each instance in $O(n^{2-\epsilon})$ time, and all the $n^\delta$ instances in $O(n^{2-\epsilon + \delta})$ time. 

Thus, now we have an $O(n^{2-\epsilon + \delta} + n^{2-\delta})$ time algorithm for the \AESparseTri\  instance, which is $O(n^{2-\epsilon / 2})$ by setting $\delta = \epsilon / 2$. This contradicts with the Real APSP hypothesis. 
\end{proof}

We also show lower bound for the \SetInter\  problem based on the Real APSP hypothesis, matching the previous lower bound from the Integer Exact Triangle hypothesis. Recall that the input of \SetInter\  is the same as the input of \SetDisj\ : a universe $U$, a collection of sets $\mathcal{F} \subseteq 2^U$, and $q$ queries. Instead of outputting whether each pair of sets intersect, the \SetInter\  problem asks the algorithm to output a certain number of elements in the $q$ intersections. Intuitively, the \SetInter\  problem asks to list a certain number of triangles in a certain type of tripartite graph. 

\begin{theorem}
For any constant $0 < \gamma < 1$ and $\delta \ge 0$, no algorithm for \SetDisj\  where $|U| = O(N^{1+\delta-\gamma})$, $|\mathcal{F}| = O(\sqrt{N^{1+\delta+\gamma}})$, each set of $\mathcal{F}$ has size $O(N^{1-\gamma})$,  $q=O(N^{1+\gamma})$, and the algorithm is required to output $O(N^{2-\delta})$ elements from the intersections, can run in $O(N^{2-\epsilon})$ time for $\epsilon > 0$, under the the Real APSP hypothesis.
\end{theorem}
\begin{proof}
If $\delta > 2$, then the input size is already $\sqrt{N^{1+\delta+\gamma}}\cdot N^{1-\gamma} \ge N^2$, and the lower bound is trivially true. Thus, we may assume $\delta \le 2$.

The key intuition in the proof is to generalize Lemma~\ref{lem:apsp1} and Lemma~\ref{lem:apsp2}. In the proof of Lemma~\ref{lem:apsp1}, we take the random subset $R \subseteq [d]$ of size $d/2$. Suppose we keep each element with probability $1/t$ for some $1 \le t \le d$ instead, the $R$ has size $\OO(d/t)$ w.h.p., and so does the number of distinct $k^{(R)}_{i,j}$. 
Thus we reduce each instance of \MinPlus\ between an $n \times d$ real matrix and a $d \times n$ real matrix to $O(d/t)$ instances of some triangle problems, but in each instance of the triangle problem we need to enumerate $\OO(t)$ triangles per edge since the number of $k$ that beats $k^{(R)}_{i,j}$ is $\OO(t)$ w.h.p. 

On the other hand, suppose we keep each element in $[d]$ with probability $1-1/t$. Then for each pair of $(i, j)$, the probability that the best $k$ in $[d] \setminus R$ beats the best $k$ in $[d]$ is only $O(1/t)$, so we only need to enumerate fewer triangles. 

In the following, we will formalize the above intuitions.

\begin{claim}
\label{cl:set_inter1}
For any $2 \le t \le d$, \MinPlus\ of an $n \times d$ real matrix $A$ and a $d \times n$ real matrix $B$ (randomly) reduces to $\OO(d/t)$ instances of \SetInter\  where $|\mathcal{F}|=n$, each set of $\mathcal{F}$ has size $d$, $q=\OO(n^2t/d)$ and the number of elements to output is $\OO(n^2t^2/d)$. 
\end{claim}
\begin{proof}
We sample a random subset $R \subseteq [d]$ by keeping each element with probability $1/t$. Then we compute $k^{R}_{i,j}=\argmin_{k \in R} (A[i,k]+B[k,j])$ recursively. Similar to the proof of Lemma~\ref{lem:apsp2}, let $P_k = \{(i,j) \in [n^2]: k_{ij} = k\}$ for each $k \in R$, and we subdivide each $P_k$ to subsets of sizes $O(n^2t / d)$. 
W.h.p., $|R| = \OO(d/t)$, and the number of such subsets is $\OO(d/t)$. 

For each subset $P$, we create a triangle (listing) instance in the same way as Lemma~\ref{lem:apsp2}. Note that the number of edges between the left and right part is $|P| = O(n^2 t / d)$ and the number of middle neighbors of every left and right node is $d$. 
The difference is that, for each $(i, j) \in P$, we expect to see $\OO(t)$ distinct $k \in [d]$ such that $A_{i,k}+B_{k,j} < A_{i, k^{R}_{i,j}}+B_{k^{R}_{i,j}, j}$, i.e., the number of triangles involving edge $(x[i], z[j])$ is $\OO(t)$ with high probability. Thus, we need to output $\OO(t)$ triangles per edge in order to find the optimal $k$. In total, we need to output $\OO(n^2 t^2/d)$ triangles. 

It is then straightforward to reduce the triangle listing instances to \SetInter\ with the required parameters. 
\end{proof}
Note that the guarantees of Claim~\ref{cl:set_inter1} does not involve the size of the universe $U$. However, similar as before, we can overcome it via the ``high vs.\ low degree'' trick. 

Given $N, \gamma, \delta$, we can use Claim~\ref{cl:set_inter1} by setting $n=N^{\frac{\delta+\gamma+1}{2}}, d=N^{1-\gamma}$ and $t=N^{1-\delta-\gamma}$. We can  check that Claim~\ref{cl:set_inter1} gives the correct values of $|\mathcal{F}|$, sizes of each set, $q$, and the number of triangles to enumerate. The lower bound implied by Claim~\ref{cl:set_inter1} is $n^{2-o(1)}d / (d/t)=N^{2-o(1)}$. Note that we need $2 \le t \le d$, so it works for all ranges of $\gamma, \delta$ where $\delta+\gamma \le 1$. 

The following claim handles the remaining cases. 

\begin{claim}
\label{cl:set_inter2}
For any $2 \le t \le n^2/d$, \MinPlus\ of an $n \times d$ real matrix $A$ and a $d \times n$ real matrix $B$ (randomly) reduces to $\OO(dt)$ instances of \SetInter\  where $|\mathcal{F}|=n$, each set of $\mathcal{F}$ has size $d/t$, $q=\OO(n^2/d)$ and the number of elements to output is $\OO(n^2/td)$. 
\end{claim}
\begin{proof}
We sample a random subset $R \subseteq [d]$ by keeping each element with probability $1-1/t$. Then we compute $k^{R}_{i,j}=\argmin_{k \in R} (A[i,k]+B[k,j])$ recursively. Notably, the depth of the recursion is $\log_{1/(1-1/t)} n = \OO(t)$, so there will be an additional $\OO(t)$ factor on number of instances. 

Then similar to the proof of Lemma~\ref{lem:apsp2}, let $P_k = \{(i,j) \in [n^2]: k_{ij} = k\}$ for each $k \in R$, and we subdivide each $P_k$ to subsets of sizes $O(n^2 / d)$. The number of such subsets is $\OO(d)$. Since the depth of the recursion is $\OO(t)$, the total number of instances is $\OO(dt)$.

For each subset $P$, we create a triangle (listing) instance in the same way as Lemma~\ref{lem:apsp2}. 
One difference is that, for each $i$ or $j$, we only need to connect it to the middle nodes corresponding to some $k' \in [d] \setminus R$, which has size $\OO(d/t)$ with high probability. 

The other difference requires some analysis. For some pair $(i,j) \in [n] \times [n]$, we consider the number of times its $k^{R}_{i,j}$ gets improved over the \textit{whole algorithm}. For each iteration, suppose the current middle set is $R$, and we are sampling $R' \subseteq R$ by adding each element to $R'$ with probability $1-1/t$. Therefore, the probability that $k^R_{i, j} \not \in R'$ is $1/t$. It means that the best $k$ for the pair $(i, j)$ is updated in the iteration for $R$ with probability $1/t$. Since there are $\OO(t)$ iterations, the number of times $k^{R}_{i,j}$ gets improved is $\OO(1)$ w.h.p. by Chernoff bound, and each time it gets improved there are at most $\OO(1)$ distinct of $k \in R \setminus R'$ that can improve over $k^{R'}_{i,j}$ w.h.p. 
Because of this, the number of triangles that are required to be outputted over the whole course of the algorithm is only $\OO(n^2)$ with high probability. 

If some instance contains more than $\OO(n^2/td)$ triangles, we can split the instance to several instances so that we only need $\OO(n^2/td)$ triangles from each instance. One way to do this split is to order pairs in $P$ arbitrarily and use binary search to find the maximum prefix $P' \subseteq P$ such that the instance corresponding to $P'$ only has $\OO(n^2/td)$ triangles, and then remove $P'$ from $P$ and continue. 

It is then straightforward to reduce the triangle listing instances to \SetInter. 
\end{proof}
Similar as before, we can bound the size of $U$ via  the ``high vs.\ low degree'' trick. 

Given $N, \gamma, \delta$, we can use Claim~\ref{cl:set_inter2} by setting $n=N^{\frac{\delta+\gamma+1}{2}}, d=N^{\delta}$ and $t=N^{\delta+\gamma-1}$. We can check Claim~\ref{cl:set_inter2} gives the correct values of $|\mathcal{F}|$, size of each set, $q$, and the number of triangles to enumerate. The lower bound implied by Claim~\ref{cl:set_inter2} is $n^{2-o(1)}d / (td)=N^{2-o(1)}$. Note that we need $2 \le t$ and $td \le n^2$, so it works for all ranges of $\gamma, \delta$ where $\delta+\gamma \ge 1$ and $\delta \le 2$. 
\end{proof}

\section{Miscellaneous Remarks}
\label{sec:final}

The \ColorBMM\ problem we have introduced has several natural variants.
For example, given $n\times n$ Boolean matrices $A$ and $B$ and a mapping $\col:[n]\rightarrow \Gamma$, for each $i,j\in [n]$,
we may want to find the most frequent element in the multiset $\{\col(k): A[i,k]\wedge B[k,j],\ k\in [n]\}$.  Call this problem \ModeBMM.
It is not difficult to adapt the proofs in Section~\ref{sec:colorBMM} to obtain similar reductions from \APSP\ or \OV\ to \ModeBMM.

\ColorBMM\ and \ModeBMM\ are related to some new variants of equality matrix products.
Following the naming convention from~\cite{williamsxumono}, one could define the following
products of two $n\times n$ integer matrices $A$ and $B$, which to our knowledge have not been studied before:
\begin{itemize}
\item \DistinctEq: for each $i,j\in [n]$, compute 
the number of distinct elements in $\{ A[i,k]: A[i,k]=B[k,j],\ k \in [n] \}$;
\item \ModeEq: for each $i,j\in [n]$, find the most frequent element in the multiset 
$\{A[i,k]: A[i,k] =  B[k,j],\ k\in [n]\}$.
\end{itemize}

It is easy to see that \DistinctEq\ and \ModeEq\ can be reduced from \ColorBMM\ and \ModeBMM\ respectively (by setting
$A'[i,k]=\col(k)$ if $A[i,k]$ is true, and $A'[i,k]=\$$ otherwise, and
$B'[k,j]=\col(k)$ if $B[k,j]$ is true, and $B'[k,j]=\#$ otherwise, for two new symbols $\$$ and $\#$).  In particular, we immediately obtain near cubic lower bounds for these problems under OVH, by Theorem~\ref{thm:colorbmm:ov}.  One can similarly consider the
\DistinctPlus\ problem (computing the number of distinct elements in $\left\{ A[i,k]+B[k,j]: k \in [n] \right\}$).
In the case when the matrix entries are integers bounded by $n^\alpha$, all these products can be computed in $\OO(n^{2+\alpha+o(1)})$ time if $\omega=2$, and we have nearly matching lower bounds for all $\alpha\le 1$ under OVH, by Theorem~\ref{thm:colorbmm:ov}.

It is possible to obtain conditional lower bounds based on the conjectured hardness of \OddSUM\ for constant $k>1$, which is believed to require $n^{k+1-o(1)}$ time (the problem is to find  $2k+1$ input numbers summing to 0).
For example,
since \OddSUM\ is known to be reducible to an asymmetric instance of \ThreeSUM\ for three sets of sizes $n^k$, $n^k$, and $n$, Theorem~\ref{thm:3sum} with $\beta=1/k$ immediately implies:

\begin{corollary}
For any fixed $k$,
if \AESparseTri\ with $m$ edges and degeneracy $m^{\alpha}$ could be solved in $\OO(m^{1+\alpha-\eps})$ time
for some constant $\alpha\le 1/(3k+2)$, then \OddSUM\ could be solved in $\OO(n^{k+1-\eps(k+1)/(1+\alpha)})$ time using Las Vegas randomization.
\end{corollary}

\cite{barba2019subquadratic} have extended Gr{\o}nlund and Pettie's {\sf 3SUM} decision-tree result~\cite{gronlund2014} to certain algebraic generalizations of {\sf 3SUM} (given sets $A$, $B$, and $C$ of real numbers, does there exist a triple $(a,b,c)\in A\times B\times C$ such that $f(a,b)=c$ for a fixed bivariate polynomial $f$ of constant degree).  It seems possible to adapt our proof to obtain a reduction from such algebraic versions of {\sf 3SUM} to \AESparseTri\ (by replacing dyadic intervals with more advanced geometric range searching techniques), although the bounds would be worse.

\bibliographystyle{alpha}
\bibliography{ref}

\newcommand{\etalchar}[1]{$^{#1}$}
\begin{thebibliography}{dBCvKO08}

\bibitem[Abr87]{matching1}
Karl Abrahamson.
\newblock Generalized string matching.
\newblock {\em SIAM J. Comput.}, 16(6):1039–1051, December 1987.

\bibitem[ACLL14]{AmirCLL14}
Amihood Amir, Timothy~M. Chan, Moshe Lewenstein, and Noa Lewenstein.
\newblock On hardness of jumbled indexing.
\newblock In {\em Proc. 41st International Colloquium on Automata, Languages,
  and Programming (ICALP)}, volume 8572 of {\em Lecture Notes in Computer
  Science}, pages 114--125. Springer, 2014.

\bibitem[AEK05]{geo6}
Daniel Archambault, William Evans, and David Kirkpatrick.
\newblock Computing the set of all the distant horizons of a terrain.
\newblock {\em Int. J. Comput. Geometry Appl.}, 15:547--564, 12 2005.

\bibitem[AGMN92]{AlonGMN92}
Noga Alon, Zvi Galil, Oded Margalit, and Moni Naor.
\newblock Witnesses for {B}oolean matrix multiplication and for shortest paths.
\newblock In {\em Proc. 33rd IEEE Symposium on Foundations of Computer Science
  (FOCS)}, pages 417--426, 1992.

\bibitem[AGV15]{abboud2014centrality}
Amir Abboud, Fabrizio Grandoni, and Virginia {Vassilevska Williams}.
\newblock Subcubic equivalences between graph centrality problems, {APSP} and
  diameter.
\newblock In {\em Proc. 2015 Annual ACM-SIAM Symposium on Discrete Algorithms
  (SODA)}, pages 1681--1697, 2015.

\bibitem[AHI{\etalchar{+}}01]{geo5}
Manuel Abellanas, Ferran Hurtado, Christian Icking, Rolf Klein, Elmar
  Langetepe, Lihong Ma, Bel{\'e}n Palop, and Vera Sacrist{\'a}n.
\newblock Smallest color-spanning objects.
\newblock In {\em Proc. 9th Annual European Symposium on Algorithms (ESA)},
  pages 278--289, 2001.

\bibitem[AHP08]{geo7}
Boris Aronov and Sariel Har-Peled.
\newblock On approximating the depth and related problems.
\newblock {\em SIAM J. Comput.}, 38(3):899--921, 2008.

\bibitem[AV14]{abboud2014popular}
Amir Abboud and Virginia {Vassilevska Williams}.
\newblock Popular conjectures imply strong lower bounds for dynamic problems.
\newblock In {\em Proc. 55th IEEE Symposium on Foundations of Computer Science
  (FOCS)}, pages 434--443, 2014.

\bibitem[AV21]{alman2021refined}
Josh Alman and Virginia {Vassilevska Williams}.
\newblock A refined laser method and faster matrix multiplication.
\newblock In {\em Proc. 2021 ACM-SIAM Symposium on Discrete Algorithms (SODA)},
  pages 522--539, 2021.

\bibitem[AVW14]{AbboudWW14}
Amir Abboud, Virginia {Vassilevska Williams}, and Oren Weimann.
\newblock Consequences of faster alignment of sequences.
\newblock In {\em Proc. 41st International Colloquium on Automata, Languages,
  and Programming (ICALP), Part {I}}, volume 8572 of {\em Lecture Notes in
  Computer Science}, pages 39--51. Springer, 2014.

\bibitem[AVY18]{abboud2018matching}
Amir Abboud, Virginia {Vassilevska Williams}, and Huacheng Yu.
\newblock Matching triangles and basing hardness on an extremely popular
  conjecture.
\newblock {\em SIAM J. Comput.}, 47(3):1098--1122, 2018.
\newblock Preliminary version in STOC 2015.

\bibitem[AW15]{AlmanW15}
Josh Alman and Ryan Williams.
\newblock Probabilistic polynomials and {H}amming nearest neighbors.
\newblock In {\em Proc. 56th IEEE Symposium on Foundations of Computer Science
  (FOCS)}, pages 136--150, 2015.

\bibitem[AWY14]{abboud2014more}
Amir Abboud, Ryan Williams, and Huacheng Yu.
\newblock More applications of the polynomial method to algorithm design.
\newblock In {\em Proc. 26th ACM-SIAM Symposium on Discrete Algorithms (SODA)},
  pages 218--230, 2014.

\bibitem[AYZ97]{AlonYZ97}
Noga Alon, Raphael Yuster, and Uri Zwick.
\newblock Finding and counting given length cycles.
\newblock {\em Algorithmica}, 17(3):209--223, 1997.

\bibitem[BCI{\etalchar{+}}19]{barba2019subquadratic}
Luis Barba, Jean Cardinal, John Iacono, Stefan Langerman, Aur{\'e}lien Ooms,
  and Noam Solomon.
\newblock Subquadratic algorithms for algebraic {3SUM}.
\newblock {\em Discrete \& Computational Geometry}, 61(4):698--734, 2019.

\bibitem[BDP08]{baran2005subquadratic}
Ilya Baran, Erik~D. Demaine, and Mihai Patrascu.
\newblock Subquadratic algorithms for {3SUM}.
\newblock {\em Algorithmica}, 50(4):584--596, 2008.
\newblock Preliminary version in WADS 2005.

\bibitem[Beh46]{behrend1946sets}
Felix~A. Behrend.
\newblock On sets of integers which contain no three terms in arithmetical
  progression.
\newblock {\em Proc. Nat. Acad. Sci.}, 32(12):331--332, 1946.

\bibitem[Ben83]{Ben-Or83}
Michael Ben{-}Or.
\newblock Lower bounds for algebraic computation trees.
\newblock In {\em Proceedings of the 15th {ACM} Symposium on Theory of
  Computing (STOC)}, pages 80--86, 1983.

\bibitem[BGMW20]{BringmannGMW20}
Karl Bringmann, Pawe\l{} Gawrychowski, Shay Mozes, and Oren Weimann.
\newblock Tree edit distance cannot be computed in strongly subcubic time
  (unless {APSP} can).
\newblock {\em {ACM} Trans. Algorithms}, 16(4):48:1--48:22, 2020.

\bibitem[BH01]{geo9}
Gill Barequet and Sariel Har{-}Peled.
\newblock Polygon containment and translational min-{H}ausdorff-distance
  between segment sets are {3SUM}-hard.
\newblock {\em Int. J. Comput. Geom. Appl.}, 11(4):465--474, 2001.
\newblock Preliminary version in SODA 1999.

\bibitem[BRS{\etalchar{+}}18]{BackursRSWW18}
Arturs Backurs, Liam Roditty, Gilad Segal, Virginia {Vassilevska Williams}, and
  Nicole Wein.
\newblock Towards tight approximation bounds for graph diameter and
  eccentricities.
\newblock In {\em Proc. 50th ACM Symposium on Theory of Computing (STOC)},
  pages 267--280, 2018.

\bibitem[CDL{\etalchar{+}}16]{CyganDLMNOPSW16}
Marek Cygan, Holger Dell, Daniel Lokshtanov, D{\'{a}}niel Marx, Jesper
  Nederlof, Yoshio Okamoto, Ramamohan Paturi, Saket Saurabh, and Magnus
  Wahlstr{\"{o}}m.
\newblock On problems as hard as {CNF-SAT}.
\newblock {\em {ACM} Trans. Algorithms}, 12(3):41:1--41:24, 2016.

\bibitem[CEH07]{geo8}
Otfried Cheong, Alon Efrat, and Sariel Har{-}Peled.
\newblock Finding a guard that sees most and a shop that sells most.
\newblock {\em Discret. Comput. Geom.}, 37(4):545--563, 2007.
\newblock Preliminary version in SODA 2004.

\bibitem[CGI{\etalchar{+}}16]{carmosino2016nondeterministic}
Marco~L. Carmosino, Jiawei Gao, Russell Impagliazzo, Ivan Mihajlin, Ramamohan
  Paturi, and Stefan Schneider.
\newblock Nondeterministic extensions of the strong exponential time hypothesis
  and consequences for non-reducibility.
\newblock In {\em Proc. ACM Conference on Innovations in Theoretical Computer
  Science (ITCS)}, pages 261--270, 2016.

\bibitem[CH20]{ChanHe}
Timothy~M. Chan and Qizheng He.
\newblock Reducing {3SUM} to convolution-{3SUM}.
\newblock In {\em Proc. SIAM Symposium on Simplicity in Algorithms (SOSA)},
  pages 1--7, 2020.

\bibitem[Cha20]{chan3sum}
Timothy~M. Chan.
\newblock More logarithmic-factor speedups for {3SUM}, (median,+)-convolution,
  and some geometric {3SUM}-hard problems.
\newblock {\em ACM Trans. Algorithms}, 16(1):7:1--7:23, 2020.
\newblock Preliminary version in SODA 2018.

\bibitem[CIP09]{exact1}
Chris Calabro, Russell Impagliazzo, and Ramamohan Paturi.
\newblock The complexity of satisfiability of small depth circuits.
\newblock In {\em Proc. 4th International Workshop on Parameterized and Exact
  Computation (IWPEC)}, pages 75--85. Springer, 2009.

\bibitem[CIP13]{cip13}
Chris Calabro, Russell Impagliazzo, and Ramamohan Paturi.
\newblock On the exact complexity of evaluating quantified {$k$-CNF}.
\newblock {\em Algorithmica}, 65(4):817--827, 2013.
\newblock Preliminary version in IPEC 2010.

\bibitem[CKN18]{exact3}
Marek Cygan, Stefan Kratsch, and Jesper Nederlof.
\newblock Fast hamiltonicity checking via bases of perfect matchings.
\newblock {\em J. ACM}, 65(3), mar 2018.

\bibitem[CN85]{chiba1985}
Norishige Chiba and Takao Nishizeki.
\newblock Arboricity and subgraph listing algorithms.
\newblock {\em SIAM J. Comput.}, 14(1):210--223, 1985.

\bibitem[CW21]{ChanW21}
Timothy~M. Chan and R.~Ryan Williams.
\newblock Deterministic {APSP}, orthogonal vectors, and more: Quickly
  derandomizing {R}azborov-{S}molensky.
\newblock {\em {ACM} Trans. Algorithms}, 17(1):2:1--2:14, 2021.
\newblock Preliminary version in SODA 2016.

\bibitem[Dah16]{dahlgaard2016hardness}
S{\o}ren Dahlgaard.
\newblock On the hardness of partially dynamic graph problems and connections
  to diameter.
\newblock In {\em Proc. 43rd International Colloquium on Automata, Languages,
  and Programming (ICALP)}, volume~55 of {\em LIPIcs}, pages 48:1--48:14, 2016.

\bibitem[dBCvKO08]{BergCKO08}
Mark de~Berg, Otfried Cheong, Marc~J. van Kreveld, and Mark~H. Overmars.
\newblock {\em Computational Geometry: Algorithms and Applications}.
\newblock Springer, 3rd edition, 2008.

\bibitem[dBdGO97]{geo2}
Mark de~Berg, Marko~M. de~Groot, and Mark~H. Overmars.
\newblock Perfect binary space partitions.
\newblock {\em Computational Geometry}, 7(1):81--91, 1997.

\bibitem[DKPV20]{duraj2020equivalences}
Lech Duraj, Krzysztof Kleiner, Adam Polak, and Virginia {Vassilevska Williams}.
\newblock Equivalences between triangle and range query problems.
\newblock In {\em Proceedings of the Fourteenth Annual ACM-SIAM Symposium on
  Discrete Algorithms}, pages 30--47. SIAM, 2020.

\bibitem[DW10]{exact2}
Evgeny Dantsin and Alexander Wolpert.
\newblock On moderately exponential time for {SAT}.
\newblock In {\em Proc. 13th International Conference on Theory and
  Applications of Satisfiability Testing (SAT)}, pages 313--325. Springer,
  2010.

\bibitem[Eri99a]{erickson1995lower}
Jeff Erickson.
\newblock Lower bounds for linear satisfiability problems.
\newblock {\em Chic. J. Theor. Comput. Sci.}, 1999(8), 1999.
\newblock Preliminary version in SODA 1995.

\bibitem[Eri99b]{geo4}
Jeff Erickson.
\newblock New lower bounds for convex hull problems in odd dimensions.
\newblock {\em SIAM J. Comput.}, 28(4):1198--1214, 1999.

\bibitem[EvdHM20]{erickson-real}
Jeff Erickson, Ivor van~der Hoog, and Tillmann Miltzow.
\newblock Smoothing the gap between np and er.
\newblock In {\em Proc. 61st IEEE Symposium on Foundations of Computer
  Science}, pages 1022--1033. IEEE, 2020.
\newblock Full version available at \url{http://arxiv.org/abs/1912.02278}.

\bibitem[Flo62]{floyd1962algorithm}
Robert~W. Floyd.
\newblock Algorithm 97: shortest path.
\newblock {\em Commun. ACM}, 5(6):345, 1962.

\bibitem[FM71]{Fischer71}
Michael~J. Fischer and Albert~R. Meyer.
\newblock Boolean matrix multiplication and transitive closure.
\newblock In {\em Proc. 12th Annual Symposium on Switching and Automata
  Theory}, pages 129--131. IEEE, 1971.

\bibitem[FP74]{matching2}
Michael~J. Fischer and Michael~S. Paterson.
\newblock String matching and other products.
\newblock In {\em Complexity of Computation, RM Karp (editor), SIAM-AMS
  Proceedings}, volume~7, pages 113--125, 1974.

\bibitem[Fre76]{fredman1976new}
Michael~L. Fredman.
\newblock New bounds on the complexity of the shortest path problem.
\newblock {\em SIAM J. Comput.}, 5(1):83--89, 1976.

\bibitem[Fre17]{freund2017improved}
Ari Freund.
\newblock Improved subquadratic {3SUM}.
\newblock {\em Algorithmica}, 77(2):440--458, 2017.

\bibitem[Gal14]{legallmult}
Fran{\c{c}}ois~Le Gall.
\newblock Powers of tensors and fast matrix multiplication.
\newblock In {\em Proc. International Symposium on Symbolic and Algebraic
  Computation (ISSAC)}, pages 296--303, 2014.

\bibitem[GO95]{gajentaan1995class}
Anka Gajentaan and Mark~H. Overmars.
\newblock On a class of ${O}(n^2)$ problems in computational geometry.
\newblock {\em Computational Geometry}, 5(3):165--185, 1995.

\bibitem[GP18]{gronlund2014}
Allan Gr{\o}nlund and Seth Pettie.
\newblock Threesomes, degenerates, and love triangles.
\newblock {\em J. {ACM}}, 65(4):22:1--22:25, 2018.
\newblock Preliminary version in FOCS 2014.

\bibitem[GS17]{gold2017improved}
Omer Gold and Micha Sharir.
\newblock Improved bounds for {3SUM}, {$k$-SUM}, and linear degeneracy.
\newblock In {\em Proc. 25th European Symposium on Algorithms (ESA)}, volume~87
  of {\em LIPIcs}, pages 42:1--42:13, 2017.

\bibitem[GU18]{GawrychowskiU18}
Pawe\l{} Gawrychowski and Przemys\l{}aw Uzna\'nski.
\newblock Towards unified approximate pattern matching for hamming and {$L_1$}
  distance.
\newblock In {\em Proc. 45th International Colloquium on Automata, Languages,
  and Programming (ICALP)}, volume 107 of {\em LIPIcs}, pages 62:1--62:13,
  2018.

\bibitem[HS74]{HartmanisS74}
Juris Hartmanis and Janos Simon.
\newblock On the power of multiplication in random access machines.
\newblock In {\em Proc. 15th Annual Symposium on Switching and Automata
  Theory}, pages 13--23. {IEEE} Computer Society, 1974.

\bibitem[IP01]{ip01}
Russell Impagliazzo and Ramamohan Paturi.
\newblock On the complexity of {$k$-SAT}.
\newblock {\em J. Comput. Syst. Sci.}, 62(2):367--375, 2001.
\newblock Preliminary version in CoCo 1999.

\bibitem[JV16]{jafargholi2016mathrm}
Zahra Jafargholi and Emanuele Viola.
\newblock $\mathrm{3SUM}$, $\mathrm{3XOR}$, triangles.
\newblock {\em Algorithmica}, 74(1):326--343, 2016.

\bibitem[KLM19]{KaneLM19}
Daniel~M. Kane, Shachar Lovett, and Shay Moran.
\newblock Near-optimal linear decision trees for {$k$-SUM} and related
  problems.
\newblock {\em J. {ACM}}, 66(3):16:1--16:18, 2019.
\newblock Preliminary version in STOC 2018.

\bibitem[KPP16]{kopelowitz2016higher}
Tsvi Kopelowitz, Seth Pettie, and Ely Porat.
\newblock Higher lower bounds from the {3SUM} conjecture.
\newblock In {\em Proc. 27th ACM-SIAM Symposium on Discrete Algorithms (SODA)},
  pages 1272--1287, 2016.

\bibitem[KS99]{KahanS99}
Simon Kahan and Jack Snoeyink.
\newblock On the bit complexity of minimum link paths: {S}uperquadratic
  algorithms for problem solvable in linear time.
\newblock {\em Comput. Geom.}, 12(1-2):33--44, 1999.

\bibitem[LMS18]{exact4}
Daniel Lokshtanov, D\'{a}niel Marx, and Saket Saurabh.
\newblock Known algorithms on graphs of bounded treewidth are probably optimal.
\newblock {\em ACM Trans. Algorithms}, 14(2), apr 2018.

\bibitem[LPV20]{lincoln2020monochromatic}
Andrea Lincoln, Adam Polak, and Virginia {Vassilevska Williams}.
\newblock Monochromatic triangles, intermediate matrix products, and
  convolutions.
\newblock In {\em Proc. 11th Innovations in Theoretical Computer Science
  Conference (ITCS)}, volume 151 of {\em LIPIcs}, pages 53:1--53:18, 2020.

\bibitem[LU18]{legallurr}
Fran{\c{c}}ois {Le Gall} and Florent Urrutia.
\newblock Improved rectangular matrix multiplication using powers of the
  {C}oppersmith-{W}inograd tensor.
\newblock In {\em Proc. 29th {ACM-SIAM} Symposium on Discrete Algorithms
  (SODA)}, pages 1029--1046, 2018.

\bibitem[LVWW16]{lincoln2016deterministic}
Andrea Lincoln, Virginia {Vassilevska Williams}, Joshua~R. Wang, and R.~Ryan
  Williams.
\newblock Deterministic time-space trade-offs for {$k$-SUM}.
\newblock In {\em Proc. 43rd International Colloquium on Automata, Languages,
  and Programming (ICALP)}, volume~55 of {\em LIPIcs}, pages 58:1--58:14, 2016.

\bibitem[P{\u{a}}t10]{patrascu2010towards}
Mihai P{\u{a}}tra{\c{s}}cu.
\newblock Towards polynomial lower bounds for dynamic problems.
\newblock In {\em Proc. 42nd ACM Symposium on Theory of Computing (STOC)},
  pages 603--610, 2010.

\bibitem[RZ11]{RodittyZ04}
Liam Roditty and Uri Zwick.
\newblock On dynamic shortest paths problems.
\newblock {\em Algorithmica}, 61(2):389--401, 2011.
\newblock Preliminary version in ESA 2004.

\bibitem[Sch79]{Schonhage79}
Arnold Sch{\"{o}}nhage.
\newblock On the power of random access machines.
\newblock In {\em Proc. 6th Colloquium on Automata, Languages and Programming
  (ICALP)}, volume~71 of {\em Lecture Notes in Computer Science}, pages
  520--529. Springer, 1979.

\bibitem[Sei95]{seidel1995}
R.~Seidel.
\newblock On the all-pairs-shortest-path problem in unweighted undirected
  graphs.
\newblock {\em J. Comput. Syst. Sci.}, 51(3):400--403, 1995.

\bibitem[SEO03]{geo3}
Michael Soss, Jeff Erickson, and Mark Overmars.
\newblock Preprocessing chains for fast dihedral rotations is hard or even
  impossible.
\newblock {\em Computational Geometry}, 26(3):235--246, 2003.

\bibitem[{Vas}12]{vstoc12}
Virginia {Vassilevska Williams}.
\newblock Multiplying matrices faster than {C}oppersmith-{W}inograd.
\newblock In {\em Proc. 44th ACM Symposium on Theory of Computing (STOC)},
  pages 887--898, 2012.

\bibitem[{Vas}18]{virgisurvey}
Virginia {Vassilevska Williams}.
\newblock On some fine-grained questions in algorithms and complexity.
\newblock In {\em Proceedings of the ICM}, volume~3, pages 3431--3472. World
  Scientific, 2018.

\bibitem[VW13]{VWfindingcountingj}
Virginia {Vassilevska Williams} and Ryan Williams.
\newblock Finding, minimizing, and counting weighted subgraphs.
\newblock {\em {SIAM} J. Comput.}, 42(3):831--854, 2013.
\newblock Preliminary version in STOC 2009.

\bibitem[VW18]{focsyj}
Virginia {Vassilevska Williams} and R.~Ryan Williams.
\newblock Subcubic equivalences between path, matrix, and triangle problems.
\newblock {\em J. {ACM}}, 65(5):27:1--27:38, 2018.
\newblock Preliminary version in FOCS 2010.

\bibitem[VWY06]{VassilevskaWY06}
Virginia Vassilevska, Ryan Williams, and Raphael Yuster.
\newblock Finding the smallest \emph{H}-subgraph in real weighted graphs and
  related problems.
\newblock In {\em Proc. 33rd International Colloquium on Automata, Languages
  and Programming (ICALP), Part {I}}, volume 4051 of {\em Lecture Notes in
  Computer Science}, pages 262--273, 2006.

\bibitem[VX20]{williamsxumono}
Virginia {Vassilevska Williams} and Yinzhan Xu.
\newblock Monochromatic triangles, triangle listing and {APSP}.
\newblock In {\em Proc. 61st IEEE Symposium on Foundations of Computer Science
  (FOCS)}, pages 786--797, 2020.

\bibitem[Wil05]{Williams05}
Ryan Williams.
\newblock A new algorithm for optimal 2-constraint satisfaction and its
  implications.
\newblock {\em Theor. Comput. Sci.}, 348(2-3):357--365, 2005.
\newblock Preliminary version in ICALP 2004.

\bibitem[Wil14]{Williams14a}
Ryan Williams.
\newblock Faster all-pairs shortest paths via circuit complexity.
\newblock In David~B. Shmoys, editor, {\em Symposium on Theory of Computing,
  {STOC} 2014, New York, NY, USA, May 31 - June 03, 2014}, pages 664--673.
  {ACM}, 2014.

\bibitem[Wil18]{Williams18}
R.~Ryan Williams.
\newblock Faster all-pairs shortest paths via circuit complexity.
\newblock {\em {SIAM} J. Comput.}, 47(5):1965--1985, 2018.
\newblock Preliminary version in STOC 2014.

\bibitem[Yuv76]{Yuval76}
Gideon Yuval.
\newblock An algorithm for finding all shortest paths using $n^{2.81}$
  infinite-precision multiplications.
\newblock {\em Inf. Process. Lett.}, 4(6):155--156, 1976.

\end{thebibliography}

\appendix

\section{Model of Computation}
\label{sec:model}

Within fine-grained complexity it is standard to work in the Word RAM model of computation with $O(\log n)$-bit words. 
We will consider variants of {\sf APSP} and {\sf 3SUM} in which the numbers are reals, as opposed to integers. To handle real numbers, the usual model of computation is a version of the Real RAM (see e.g. Section 6 in the full version of ~\cite{erickson-real}) which supports unit cost comparisons and arithmetic operations (addition, subtraction, multiplication, division) on real numbers, unit cost casting integers into reals, in addition to the standard unit cost operations supported by an $O(\log n)$-bit Word RAM\@.
No conversions from real numbers to integers are allowed, and randomization only happens by taking random $O(\log n)$-bit integers, not random reals.

It is known~\cite{HartmanisS74} that unit cost multiplication makes RAM machines very powerful.
For fine-grained complexity purposes when exact polynomial running times matter, even additions of real numbers can give the Real RAM a lot of power over the standard Word RAM\footnote{For instance, 
using 
$O(n)$ additions one can create $n$-bit numbers at a cost of $O(n)$ and then arithmetic operations between these $n$-bit numbers also take unit times, giving the Real RAM a polynomial advantage over the Word RAM.}.

However, it is not difficult to formulate ``reasonable'' restricted versions of the Real RAM which would not make it more powerful (i.e., it can still be simulated efficiently on the Word RAM when the input numbers are $O(\log n)$-bit integers). In all the restricted versions of the Real RAM, we only restrict the set of unit cost operations available between real numbers, while keeping the Word RAM part of the model and the interplay between reals and integers the same as the original Real RAM. 
We describe several specific options below. 

\paragraph{``Reasonable'' Real RAM models.}
We start by defining two of the weakest (i.e., most restrictive) models, under which all our reductions hold. We emphasize that \emph{from the perspective of conditional lower bounds, reductions that work in weaker models are better}
(they automatically work in stronger models),
 and \emph{hypotheses based on weaker models are more believable}.
Thus, these two simplest models are sufficient for the purposes of this paper.

\begin{enumerate}
\item[(A)] \emph{Real RAM with 4-linear comparisons}.  In this model, the only operations
allowed on the input real numbers are comparisons of the form 
$a+b<c$ or $a+b < a'+b'$ where $a,b,a',b',c$ are input values.  

All of our reductions from \ThreeSUM, \ExactTri, and 
\MinPlus{} clearly work in this restricted model.  All known algorithms
for these problems also work in this model, and so does Fredman's and Gr\o{}nlund and Pettie's decision tree upper bound~\cite{fredman1976new,gronlund2014} (which holds in the 4-linear decision trees).  An exception is
Kane, Lovett, and Moran's decision tree upper bound~\cite{KaneLM19}, which requires a relaxation to 6-linear comparisons.

\item[(A$'$)] \emph{Real RAM with 4-linear comparisons and restricted additions}.
For \APSP, the model needs slight strengthening, since additions are needed to compute path lengths.  We allow registers storing intermediate real numbers.
We may add two register values, and perform comparisons of the form
$a+b<a'+b'$ where $a,b,a',b'$ are register values. 
However, we require that at any time, the value
of a register is equal to the length of some path. 

Note that register values cannot be exponentially large because of this requirement
(in particular, we cannot repeatedly double a number by adding it to itself).
\end{enumerate}

While the above models are sufficient for our purposes,
we mention two stronger models that are popularly used from the algorithm designers'
perspectives, though they may not have been explicitly stated in previous works:

\begin{enumerate}
\item[(B)] \emph{Real RAM with low-degree predicates}.  In this model, the only type of operation
allowed on the input real numbers is testing the sign of a constant-degree polynomial, with constant number of arguments and constant integer coefficients, evaluated at a constant number of the input values.  (Clearly, (B) extends (A), with degree~1.)

This version of the Real RAM model is usually sufficient for algorithms in computational geometry (e.g., for convex hulls, Voronoi diagrams, or line segment intersections).  It
is reminiscent of the standard
\emph{algebraic decision tree} model~\cite{Ben-Or83}.
For example, in 2D, testing whether three points are in clockwise order reduces to testing the sign of a degree-2 polynomial, and testing whether a point is in the circle through three given points reduces to testing the sign of a degree-4 polynomial~\cite{BergCKO08}.

\item[(B$'$)] \emph{Real RAM with low-degree computation}. 
In a still stronger model, we allow registers storing intermediate real numbers, support additions, subtractions, and multiplications on these registers, and allow comparisons of two register values.
However,  we require that at any time, the value
of a register is equal to the evaluation of a constant-degree
polynomial with $n$ arguments and constant integer coefficients, evaluated at the $n$ input values.  
(Clearly, (B$'$) extends (A$'$), with degree~1.)

Note that when the input numbers are $O(\log n)$-bit integers, intermediate numbers remain $O(\log n)$-bit integers because of this requirement; thus, computation can be simulated efficiently in the Word RAM\@.  We can further relax the model by allowing the degree to be $n^{o(1)}$ and coefficients to be bounded by $2^{n^{o(1)}}$, and computation can still be simulated in the Word RAM with an $n^{o(1)}$-factor slow-down.
\end{enumerate}

Though rarer, there are some examples of Real RAM algorithms in the computational geometry literature that do not fit in the above stronger models (e.g., for computing minimum-link paths in polygons~\cite{KahanS99}), because intermediate numbers
are generated iteratively from earlier numbers in a manner that forms a tree of large depth.  However, such algorithms are generally recognized as impractical due to precision issues.

\paragraph{How an unrestricted Real RAM may be unreasonable.}
An unrestricted Real RAM that supports arithmetic operations and the
floor function is known to have the ability to
solve PSPACE-hard problems in polynomial time~\cite{Schonhage79}.
It could also alter the fine-grained complexity of problems in P\@.
For example, in this model, it is possible to solve \AESparseTri\ in near-linear time (exercise left to the reader), and thus  solve \APSP\ in truly subcubic time and \ThreeSUM\ in truly subquadratic time by our reductions.

Though not as well known, an unrestricted Real RAM \emph{without} the floor function
can still do unusual things.  For example, it is possible to detect \emph{one}
triangle in a sparse graph in near-linear time, or 
solve \IntAPSP\ (but not necessarily \APSP)
in $O(n^{\omega+o(1)})$ time (similar to~\cite{Yuval76}, but with a small
modification to avoid the floor function).  It is not clear if the Real APSP or 3SUM
hypothesis could still be true in this model, but because of the unusualness of the model, we
are not willing to bet on it.

\paragraph{Implications to the large integer setting.}
Even to readers who are not interested in problems with real-valued inputs in the Real RAM Model, our techniques are still useful for versions of the problems for \emph{large integer} input in the standard Word RAM\@.  Our reductions hold in the 4-linear comparison model with restricted additions, and the cost of each 4-linear comparison and restricted addition is almost linear in the number of bits in the input integers.  For example, the proof
of Theorem~\ref{thm:apsp} implies that for any fixed $\delta>0$, if \AESparseTri\ could be solved in $\OO(m^{4/3-\eps})$ time, then \IntAPSP\ for $O(n^{1/2-\delta})$-bit
integers could be solved in truly subcubic time (since our reduction uses
$\OO(n^{5/2})$ comparisons and restricted additions, each now costing $\OO(n^{1/2-\delta})$).
The previous reduction from \IntAPSP~\cite{williamsxumono} can only handle $n^{o(1)}$-bit integers.  The hypothesis that \IntAPSP\ does not have truly subcubic algorithm  
for $O(n^{1/2-\delta})$-bit
integers is more believable than the original \IntAPSP\ hypothesis for $O(\log n)$-bit integers.

For \IntThreeSUM, one can use hashing to map large integers to $O(\log n)$-bit integers with randomization, and so the result here would not be new.  But our techniques work for certain inequality variants of \IntThreeSUM, for which hashing is not
directly applicable.

\section{\ThreeSUM\ \texorpdfstring{$\rightarrow$}{rightarrow} \ExactTri}
\label{sec:3sum2exactT}
\begin{theorem}
\ExactTri\ requires $n^{2.5-o(1)}$ time assuming the Real $3$SUM hypothesis. 
\end{theorem}
\begin{proof}
Given a \ThreeSUM\ instance, we first sort all the numbers, then split the sorted list into $g$ buckets $B_1,\ldots,B_g$ of size $n/g$ each for some $g$ to be determined later. Let $b_i$ be the first element in bucket $B_i$. Then all elements in $B_i$ are in the interval $[b_i,b_{i+1}]$.

Let's call a triple of buckets $B_i,B_j,B_k$ valid if $b_i+b_j+b_k\leq 0\leq b_{i+1}+b_{j+1}+b_{k+1}$. It was shown that 
number of valid triples is $O(g^2)$ (see e.g. \cite{VWfindingcountingj}), and we can find those triples in $O(g^2)$ time.

If some pair of buckets $B_i,B_j$ has at least $n^\eps$ buckets $B_k$ so that $B_i,B_j,B_k$ is a valid triple, then we go through each pair of elements in $B_i\cup B_j$ and check if the negative of their sum is in the list. The time for this over all such bucket pairs is asymptotically at most $\OO((g^2/n^\eps) (n/g)^2) = \OO(n^{2-\eps})$, which is truly subquadratic. 

The rest of the reduction follows from \cite{VWfindingcountingj}, but we include the full proof for completeness. 

Now for the remaining $B_i,B_j$ pairs, the number of $B_k$ that make a valid triple is at most $n^\eps$. 
We then create an \ExactTri\ instance for every $q\in [n^\eps]$ and $p$ in $[n/g]$ with three part of nodes. 

The first part represents buckets $B_i$, the second part represents buckets $B_k$, and the third part represents pairs $(s, t) \in [n/g] \times [n/g]$. 

Between $B_i$ and $B_i$, we add an edge with weight equal to the $p$-th number in the $q$-th bucket $B_k$ that forms a valid triple with $B_i,B_j$. Between $B_i$ and $(s, t)$, we add an edge with weight equal to the $s$-th number in bucket $B_i$. Between $B_j$ and $(s, t)$, we add an edge with weight equal to the $t$-th number in bucket $B_j$. 

Clearly, one of the \ExactTri\ instance has an exact triangle if and only if the remaining triples have a solution to the \ThreeSUM\ instance, since the set of triangles in these \ExactTri\ instances corresponds to all triples if numbers that could potentially sum up to $0$. 

The number of nodes is $O(g+(n/g)^2)$, and if we set $g=n^{2/3}$ we get $n^{\eps+1/3}$ instances of \ExactTri\ with $n^{2/3}$ nodes each. Suppose there exists an $O(n^{2.5-\delta})$ for $\delta > 0$ time algorithm for \ExactTri, then we can solve \ThreeSUM\ in $\OO(n^{2-\eps} + n^{\eps + 1/3} (n^{2/3})^{2.5-\delta})$ time, which is $\OO(n^{2-\delta/2})$ if we set $\eps = \delta / 2$, violating the Real $3$SUM hypothesis. 
\end{proof}

\end{document}